\documentclass[a4paper,12pt]{article}
\usepackage[utf8]{inputenc}
\usepackage{amsmath,amsfonts}
\usepackage{amsthm} 
\usepackage{amssymb}
\usepackage{multirow}
\usepackage{latexsym}
\usepackage{color}
\usepackage{graphicx} 
\usepackage{subcaption}
\usepackage[round]{natbib}
\usepackage[colorlinks=true,linkcolor=blue,urlcolor=black,citecolor=blue]{hyperref}
\usepackage{booktabs}
\usepackage[nodisplayskipstretch]{setspace} %

\usepackage{fullpage}

\newtheorem{theorem}{Theorem}
\newtheorem{lemma}{Lemma}[section]
\newtheorem{corollary}[lemma]{Corollary}
\newtheorem{remark}[lemma]{Remark}
\newtheorem{definition}[lemma]{Definition}
\newtheorem{proposition}[lemma]{Proposition}

\newcommand{\argmin}[1]{\underset{#1}{\operatorname{argmin}}}
\allowdisplaybreaks

\usepackage{bbm}

\newcommand{\IR}{\mathbb{R}}
\newcommand{\IE}{\mathbb{E}}
\newcommand{\IP}{\mathbb{P}}
\newcommand{\Ind}{\mathbbm{1}}
\newcommand{\IN}{\mathbb{N}}

\renewcommand{\tilde}{\widetilde}
\renewcommand{\epsilon}{\varepsilon}
\renewcommand{\phi}{\varphi}

 \numberwithin{equation}{section}

\usepackage{verbatim,color}

\def\boxit#1{\vbox{\hrule\hbox{\vrule\kern6pt
			\vbox{\kern6pt#1\kern6pt}\kern6pt\vrule}\hrule}}

\usepackage{chngcntr}
\usepackage{apptools}
\AtAppendix{\counterwithin{lemma}{section}}
\AtAppendix{\counterwithin{theorem}{section}}
\AtAppendix{\counterwithin{table}{section}}

 \makeatletter
\def\mythanks#1{%
	\protected@xdef \@thanks {\@thanks \protect \footnotetext [\the \c@footnote ]{#1}}%
}
\makeatother

\begin{document}
\onehalfspacing

\title{Inference in Regression Discontinuity Designs with High-Dimensional Covariates\mythanks{First version: October 26, 2021. 
		This version: \today.
		%We thank (...) for  helpful comments and suggestions.
		The authors gratefully acknowledge financial support by the European Research Council (ERC) through grant SH1-77202. Computations for this work were done in part using resources of the Leipzig University Computing Centre.
		 Author contact information:    
		Alexander Kreiß, Mathematical Institute, Leipzig University and Department of Statistics, London School of Economics. E-Mail: alexander.Kreiss@math.uni-leipzig.de. 
		Christoph Rothe, Department of Economics, University of Mannheim.
		E-Mail: rothe@vwl.uni-mannheim.de. Website: http://www.christophrothe.net.  }}

\author{Alexander Kreiß \and Christoph Rothe}

\date{}
\maketitle

\begin{abstract}
We study regression discontinuity designs in which many predetermined covariates, possibly much more than the number of observations, can be used to increase the precision of treatment effect estimates. We consider a two-step estimator which first selects a small number of ``important'' covariates through a localized Lasso-type procedure, and then, in a second step, estimates the treatment effect by including the selected covariates linearly into the usual local linear estimator. We provide an in-depth analysis of the algorithm's theoretical properties, showing that, under an approximate sparsity condition, the resulting estimator is  asymptotically normal, with asymptotic bias and variance that are conceptually similar to those obtained in low-dimensional settings. Bandwidth selection and inference can  be carried out using standard methods. We also provide simulations and an empirical application.

\end{abstract}

\newpage

\section{Introduction}

Regression discontinuity (RD) designs are widely used for estimating causal effects from observational data in economics and other social sciences. These designs exploit institutional settings in which a unit's treatment assignment is determined by whether its realization of a running variable falls above or below some known cutoff value. Estimates of causal effects are then obtained by comparing the outcomes of units that are close to, but on different sides of the cutoff. Methods based on local linear regression are widely used in this context, and their theoretical properties have been studied extensively in the literature \citep[e.g.][]{hahn2001identification, IK12, calonico2014robust, AK18}. 

While an empirical RD study can be carried out using only data on the outcome and the running variable, in practice researchers often want to incorporate additional covariates into their analysis to improve the precision of their  estimates. This is commonly done by including the covariates linearly and without localization in a standard local linear RD regression  \citep{CCFT19}. Such linear adjustment estimators generally work well if the number of available covariates is small relative to the sample size. However, they might yield distorted inference even with a rather moderate number of covariates; and they are generally  expected to break down in high-dimensional settings where the number of covariates is large and possibly even exceeds the number of observations. Such settings can occur, for example, when working with rich administrative data sets, but also if a large number of transformations, like interactions or polynomials, is applied to a low-dimensional set of underlying covariates.

In this paper, we study a two-step approach that addresses this problem. In the first step, we select a small subset of the covariates by adding an $\ell_1$ or Lasso penalty \citep[cf.][]{T96} to the local least squares problem that defines the linear adjustment estimator, and collect those variables with non-zero coefficient estimates. By construction, the selected covariates are strongly related to the outcome, and thus have the greatest potential to ``absorb'' some of its variance. In the second step, we then compute a standard linear adjustment estimator, but use only the selected covariates. We show that the resulting ``post-Lasso'' estimator is asymptotically normal under an approximate sparsity condition, with asymptotic bias and variance that are conceptually similar to those obtained in low-dimensional settings (here ``approximate sparsity'' means that only a small number of the covariates is particularly relevant for the empirical analysis, in the sense that including further covariates would not lead to meaningful improvements of estimation accuracy).  We also argue that one can  use standard methods for bandwidth selection and inference with the selected variables, making the method very easy to implement in practice.

Our estimator has many parallels with the well-known post-Lasso approach  for treatment effect estimation under unconfoundedness with high-dimensional controls in \citet{BCH14}, including the use of a similar notion of approximate sparsity. One important difference is that, in contrast to unconfoundedness, it is not necessary to select the ``right'' covariates in our RD framework in order to obtain a consistent estimator. This is because in our setting the purpose of controlling for covariates is only to increase efficiency, and not to address issues of selection bias.\footnote{This is a conceptual parallel between our setting and the use of covariates in randomized controlled experiments with a constant propensity score \citep[e.g.][]{wager2016high}. Note, however, that there is no explicit notion of random assignment in RD designs, and thus results from the literature on experiments do not simply carry over to our setting. } Our method therefore  only requires a single selection step that collects variables which are strongly related to the outcome, and not a ``double selection'' as in \citet{BCH14} that also selects variables related to treatment status. Our variable selection step is also not based on the standard Lasso, but on a ``Lasso-penalized'' local linear regression problem; and in contrast to unconfoundedness, one cannot make use of (conditionally) random treatment assignment in RD settings, but only exploit continuity conditions. The theoretical results that we derive in this paper therefore do not follow from existing arguments.

Our paper contributes to a growing literature that considers covariates in RD designs,  including \citet{AK18}, \citet{CCFT19}, \citet{froelich2019including}, \citet{noack2021flexible} and \citet{arai2021regression}.
In particular, \citet{arai2021regression} study an estimation approach that is very similar to the one we consider in this paper. However, their analysis relies on strong conditions, like exact sparsity and a ``$\beta$-min'' condition that puts a large lower bound on the coefficients of relevant covariates, which imply perfect model selection.\footnote{Roughly speaking, the conditions in \citet{arai2021regression} describe a setup in which the covariates can be partitioned into groups of ``very important'' and ``completely irrelevant'' ones, irrespective of the chosen bandwidth. Moreover,  the influence of the  ``very important'' ones is assumed to be large enough that they are selected by a localized Lasso procedure with near-certainty.
Such conditions seem unlikely to be satisfied in practice. We also note that the estimator studied by \citet{arai2021regression} differs from ours in that it does not use all the covariates with non-zero coefficients in the first stage for the ``post-Lasso'' step, but only those whose estimated coefficients exceed some positive bound, which depends on an additional tuning parameter. } Our paper uses an arguably more realistic framework, does not require perfect model selection, and develops a complete asymptotic theory for the final RD estimator. Our paper's technical arguments are also related to those in \citet{su2019non}, who use a localized Lasso to handle high-dimensional covariates in a nonparametric setup, namely a continuous treatment model.

The remainder of this paper is structured as follows. In Section~\ref{sec:setup} we introduce our model and our proposed estimator, and give an informal description of its theoretical properties. In Section~\ref{sec:theory} we give the assumptions for our theoretical analysis, and state and discuss our  main result. Section~\ref{sec:numerical_results} explains some implementation details of our procedure, and gives the results of a simulation study and an empirical illustration. Section~5 concludes. All proofs are collected in the Online Appendix.

\section{Setup and Method}
\label{sec:setup}

\subsection{Setup and Preliminaries}

Consider a sharp RD design to determine the causal effect of a binary treatment on some outcome variable of interest. The data are
an independent sample $\{(Y_i,X_i,Z_i), i=1,\ldots,n\}$ of size $n$ from some large population. Here  $Y_i\in\IR$ is the outcome variable, $X_i\in\IR$ is the running variable, and $Z_i \in \IR^p$ is a  vector of pre-treatment covariates. We particularly consider high-dimensional settings in which the covariate dimension $p$ can be large relative to, or indeed significantly larger than, the sample size $n$. We account for this in our framework by allowing $p=p_n$ to increase with the number of observations.
High-dimensional covariates occur of course if the researcher observes a large number of conceptually distinct variables for each unit, but also if  the researcher applies a large number of transformations from a dictionary of basis functions, that might create interactions or polynomials, to an underlying low-dimensional vector of covariates.

Units receive the treatment if and only if the running variable exceeds some known cutoff, which we normalize to zero without loss of generality. We denote the resulting treatment indicator by $T_i$, so that $T_i=\Ind(X_i \geq 0)$. Units also have potential outcomes $Y_i(t)$, for $t \in\{0, 1\}$, corresponding to the outcome unit $i$ would have experienced had it received treatment $t$, so that $Y_i = Y_i(T_i)$. The parameter of interest is the average treatment effect among units at the cutoff:
$$\tau_Y = \IE(Y_i(1)-Y_i(0)|X_i=0).$$
If $\IE(Y_i(t)|X_i=x)$ is continuous around the cutoff for $t \in\{0, 1\}$, this parameter is identified by the jump in the conditional expectation function $\IE(Y_i|X_i=x)$ of the observed outcome given the running variable at the threshold:
\begin{equation}\label{eq:estimand_limits}
\tau_Y = \lim_{x\downarrow 0} \IE(Y_i|X_i=x) - \lim_{x\uparrow 0}\IE(Y_i|X_i=x).
\end{equation}

Local linear regression \citep{fan1996local} is arguably the most popular framework for estimation and inference in RD designs. In the absence of covariates, the jump $\tau_Y$ is estimated by fitting a linear regression of $Y_i$ on $X_i$ locally around the cutoff, allowing for different intercepts and slopes on each side. This estimator is the baseline procedure for our analysis:
 \begin{equation}
 \label{eq:Base}
 \hat\tau_{h,\textnormal{Base}} = e_2^\top\argmin{\theta\in\IR^4}\sum_i^nK_h\left(X_i\right)\left(Y_i-V_i^\top\theta\right)^2,
 \end{equation}
with  $K$  a non-negative kernel function, $h>0$ a bandwidth, $K_h(x)=K(x/h)/h$, $V_i=\left(1,T_i,X_i/h,T_iX_i/h\right)^\top$ a vector of appropriate transformations of the running variable,  and $e_2 = (0,1,0,\ldots,0)^\top$  a unit vector of appropriate length.
 As discussed in \citet{CCFT19}, practitioners often augment the local regression in~\eqref{eq:Base} with additional covariates in a simple linear fashion, which yields the linear adjustment estimator
\begin{equation}
\label{eq:LS}
\hat\tau_{h,\textnormal{CCFT}} = e_2^\top\argmin{(\theta,\gamma)\in \IR^{4+p}}\sum_i^nK_h\left(X_i\right)\left(Y_i-V_i^\top\theta-Z_i^\top\gamma\right)^2.
\end{equation}
This estimator is consistent under standard regularity conditions if the dimension of the covariates is fixed and if their conditional  distribution given the running variable  changes smoothly around the cutoff, in the sense that  the conditional expectation of the covariates given the running variable does not jump:
\begin{equation}
\tau_Z = \lim_{x\downarrow 0} \IE(Z_i|X_i=x) - \lim_{x\uparrow 0}\IE(Z_i|X_i=x) =0. \label{eq:tauz}
\end{equation}

The linear adjustment estimator is typically more efficient than the baseline ``no covariates'' estimator. It is not uniquely defined, however, if the number of local parameters exceeds the number of observations that receive positive kernel weights in~\eqref{eq:LS}. Moreover, due to overfitting, asymptotic approximations based on a ``fixed $p$'' analysis might not provide adequate descriptions of the estimator's finite sample properties
 even in settings where the number of covariates is moderate relative to the effective sample size. For instance, in our simulations below we illustrate that conventional standard errors might severely underestimate the true variability of the linear adjustment estimator in a setting with 10--50 covariates and 1,000 data points. Linear adjustment estimators are therefore only appropriate for very low-dimensional settings.

\subsection{Proposed Method}
\label{subsec:method}

A natural way to extend linear adjustment estimators to high-dimensional settings is to consider versions that only use  a ``small'' active subset of the available covariates. Formally, with $J=\{j_1,\ldots,j_s\}\subset\{1,\ldots, p\}$ a generic subset of the covariates' indices of size $s\equiv|J|\ll p$, and $Z_{i}(J)= (Z^{(j_1)}_i,\ldots,Z^{(j_s)}_i)^{\top}$ the $s$-dimensional vector of components of $Z_i$ whose indices are collected in $J$, such estimators are given by
\begin{equation}
\label{eq:fixedJ}
\hat{\tau}_h(J) = e_2^\top\argmin{(\theta,\gamma)\in \IR^{4+s} } \sum_i^nK_h\left(X_i\right)\left(Y_i-V_i^\top\theta-Z_{i}(J)^\top\gamma\right)^2.
\end{equation}
Using arguments from \citet{CCFT19}, it is easily seen that such estimators are consistent for any fixed covariate subset $J$ under the appropriate regularity conditions. The choice of $J$ does affect the asymptotic variance, however, and using covariates that have high  correlation  with the outcome (locally at the cutoff) can generally be expected to yield more efficient estimates of $\tau_Y$. In practice, the identity of these ``most useful'' covariates is typically not known \emph{a priori}, but can potentially  be inferred in a data driven way. We therefore consider  estimators of the form $\hat{\tau}_h(\hat J_n)$, with $\hat J_n$ a data-dependent subset of the covariates' indices that is intended to contain the most relevant ones.

Our proposed RD estimator for settings with high-dimensional covariates determines the set $\hat J_n$ of active covariates through a ``localized'' version of a Lasso regression in a preliminary model selection step. We consider a version of the minimization problem in~\eqref{eq:LS} that includes an additional penalty on the sum  of the absolute values of the coefficients associated with the (appropriately standardized)  covariates, %\textcolor{blue}{What does this mean? Should it be \emph{the sum  of the absolute values of the (appropriately standardized) coefficients associated with the (centralized) covariates}},
 and define $\hat J_n$ as the set of covariate indices for which the  corresponding coefficient estimate is non-zero. Specifically, our procedure is as follows.
\begin{enumerate}
	\item Using a preliminary bandwidth $b$ and a penalty parameter $\lambda$, solve the
	following ``Lasso version'' of the weighted least squares problem in~\eqref{eq:fixedJ}:
		$$\left(\tilde{\theta}_n,\tilde{\gamma}_n\right)=\argmin{(\theta,\gamma)\in \IR^{4+p_n}}\sum_{i=1}^nK_b(X_i)\left(Y_i-V_i^{\top}\theta-\left(Z_i-\hat{\mu}_{Z,n}\right)^{\top}\gamma\right)^2+\lambda\sum_{k=1}^{p_n}\hat{w}_{n,k}|\gamma_k|,$$
	where
	\begin{align*}
	\hat{\mu}_{Z,n}=\frac{1}{n}\sum_{i=1}^nZ_iK_b\left(X_i\right) \textnormal{ and }
	\hat{w}_{n,k}^2=\frac{b}{n}\sum_{i=1}^n\left(K_b\left(X_i\right)Z_i^{(k)}-\mu_{Z,n}^{(k)}\right)^2
	\end{align*}
	 are the local sample mean and variance, respectively, of the covariates. Note that standardizing the covariates allows the penalty parameter $\lambda$ to be reasonably tuned to all coefficients simultaneously.

	\item Using a final bandwidth $h$, compute the restricted post-Lasso estimate of $\tau_Y$ as $\hat{\tau}_h(\hat{J}_n)$ as in~\eqref{eq:fixedJ}, where  $\hat{J}_n=\{k\in\{1,...,p_n\}:\tilde{\gamma}_n^{(k)}\neq 0\}$
	is the set of the indices of those covariates selected in the first step.
\end{enumerate}

The tuning parameters $b$ and $\lambda$ do not appear in the asymptotic distribution of $\hat{\tau}_h(\hat{J}_n)$, and hence some \emph{ad hoc} choices are needed. 
To choose $b$, we recommend  using a method for bandwidth choice designed for the baseline estimator without covariates, like the ones proposed in  \citet{IK12}, \citet{calonico2014robust} or \citet{AK18}. For $\lambda$, we compare different approaches in our simulations, and the results suggest that the plug-in procedure of \citet{BCH14} works well in practice. The final bandwidth $h$ can be chosen via any method suitable for linear adjustment estimators with low-dimensional covariates, such as those discussed in  \citet{CCFT19} or  \citet{AK18}.

\subsection{Overview of Main Result} 
\label{subsec:overview}
We now give an informal overview of the main theoretical result in this paper. For generic random vectors $A$ and $B$, we use the notation that 
$\mu_{A}(x)=\IE(A|X=x)$,  $\mu_{AB}(x)=\IE(AB^\top|X=x)$, $\sigma^2_{AB}(x)=\mu_{AB}(x)-\mu_{A}(x)\mu_{B}(x)^\top$;
and write $\sigma_{A}^2(x)=\sigma_{AA}^2(x)$ for simplicity. For a generic function $f$,
we also write $f_+\ = \lim_{x\downarrow 0}f(x)$ and $f_-\ = \lim_{x\uparrow 0}f(x)$ for its   right and left limit at zero, respectively, so that $\tau_Y=\mu_{Y+}-\mu_{Y-}$,
for example.

A key assumption for our analysis is that the covariates satisfy an \emph{approximate sparsity} condition, which intuitively means that only a small subset of the covariates is particularly relevant for the empirical analysis, and that including any further covariates would not lead to meaningful improvements of estimation accuracy. To state this notion more formally,  we define  the following population regression coefficients and corresponding residuals for any $J\subset\{1,\ldots,p_n\}$ and  bandwidth $h$:
\begin{equation}
\label{eq:optimal_regression}
\begin{array}{l}
(\theta_{0}(J,h),\gamma_{0}(J,h))=\argmin{(\theta,\gamma)}\quad \IE\left(K_h(X_i)\left(Y_i-V_i^{\top}\theta-Z_i(J)^{\top}\gamma\right)^2\right), \\
r_{i}(J,h)=Y_i-V_i^{\top}\theta_{0}(J,h)-Z_i(J)^{\top}\gamma_{0}(J,h).
\end{array}
\end{equation}
Approximate sparsity then means that there exist   deterministic \emph{target covariate sets} $J_n\subset \{1,\ldots,p_n\}$   that contain a ``small'' number  $s_n\equiv|J_n|\ll p_n$ of elements, and are such that the local correlation between 
the corresponding regression errors $r_{i}(J_n,h)$ and \emph{each} component of $Z_i$ is small relative to the estimation error:
\begin{align*}
\max_{j=1,...,p_n}\left|\IE\left(K_h(X_i) Z_{i}^{(j)}r_{i}(J_n,h)\right)\right|=O\left(\sqrt{\frac{\log p_n}{nh}}\right).
\end{align*}
 Moreover, this condition needs to be satisfied for an appropriate range of bandwidths, so that the sequence $J_n$ does  not depend on the exact choice of $h$.

Under this and other regularity conditions discussed below, one can show that the post-Lasso estimator $\hat{\tau}_h(\hat{J}_n)$ has the same first-order asymptotic properties as an infeasible estimator $\hat{\tau}_h(J_n)$ that uses the true target set, and then prove an asymptotic normality result for the latter. Taken together, this yields the main result of our paper, which is that the post-Lasso estimator $\hat{\tau}_h(\hat J_n)$ of $\tau_Y$ satisfies
\begin{align}
&\frac{\sqrt{nh}\left(\hat{\tau}_h(\hat J_n)-\tau_Y  -h^2\mathcal{B}_n\right)}{\mathcal{S}_n}\overset{d}{\to}\mathcal{N}\left(0,1\right),\label{eq:main_result}
\end{align}
with  asymptotic bias and variance, respectively,  such that
\begin{align}
\mathcal{B}_n &\approx\frac{C_\mathcal{B}}{2}\left(\mu_{\tilde Y+}''-\mu_{\tilde Y-}''\right)\quad\textnormal{and}\quad\mathcal{S}^2_n\approx\frac{C_\mathcal{S}}{f_X(0)}\left(\sigma^2_{\tilde Y+} +\sigma^2_{\tilde Y-} \right) \label{eq:approx_bias_var}
\end{align}
in a sense made precise below. Here $C_\mathcal{B}$ and $C_\mathcal{S}$ are constants that depend on the kernel function $K$ only, and
$$\tilde Y_i = Y_i-Z_{i}(J_n)^\top\gamma_n, \textnormal{ with }\gamma_n=\left(\sigma_{Z(J_n)-}^2+\sigma_{Z(J_n)+}^2\right)^{-1}\left(\sigma_{YZ(J_n)-}^2+\sigma_{YZ(J_n)+}^2\right),$$
is a ``covariate-adjusted'' version of the outcome variable that uses a vector $\gamma_n$ that can be thought of as an approximation of $\gamma_0(J_n,h)$ that is independent of the bandwidth. Our proposed estimator is thus first-order asymptotically equivalent to a ``baseline'' sharp RD estimator as in~\eqref{eq:Base}  with the covariate-adjusted outcome
$\tilde Y_i$
replacing the original outcome $Y_i$.\footnote{Note, however, that in contrast to the $Y_i$ the distribution of the $\tilde Y_i$ depends on the sample size $n$, and hence the properties of this estimator do not directly follow from existing ones for ``baseline'' sharp RD estimators. }

Note that, as in  \citet{CCFT19}, the continuity of $\mu_Z$ from~\eqref{eq:tauz} is necessary to establish this result. If the components of $\mu_Z$ could potentially have a jump at the cutoff,
the estimator $\hat{\tau}_h(\hat{J}_n)$ would generally not be consistent for $\tau_Y$, but satisfy $\hat{\tau}_h(\hat{J}_n) = (\tau_Y-\tau_Z^\top\gamma_n)(1+o_P(1)).$ In practice, researchers may want to investigate the plausibility of assuming~\eqref{eq:tauz} by running a series of sharp RD regressions in which the covariates take the role of the dependent variable, and testing whether the estimated jump at the cutoff is significantly different from zero via some approach that is appropriate for large-scale hypothesis testing \citep[e.g.][]{benjamini1995controlling}. Alternatively, one could carry out this exercise only for the selected covariates $Z_i(\hat J_n)$.

The formulas for the bias and variance of $\hat{\tau}_h(\hat{J}_n)$ in~\eqref{eq:approx_bias_var} are analogous to those obtained in \citet{CCFT19} for the case that $J_n \equiv J$ contains only a fixed number of covariates. This suggests that one can  select the final bandwidth $h$ and create a confidence interval for $\tau_Y$  by applying their proposed methods for low-dimensional setups  to the generated data set  $\{(Y_i,X_i,Z_i(\hat J_n)), i=1,\ldots,n\}$ that only contains the covariates selected by our algorithm. Similarly, given a bound on the second derivative of the function $\IE(\tilde Y_i|X_i=x)$, one can select the bandwidth $h$ and construct confidence intervals  for $\tau_Y$ by using the methods proposed by \citet{AK18} with the  generated data set. See Section~\ref{sec:implementation} for further discussion of implementation details.

\section{Theoretical Analysis}
\label{sec:theory}
\subsection{Assumptions}
\label{subsec:assumptions}

We impose the following assumptions in our theoretical analysis.

\smallskip

\noindent\textbf{Assumption (BW):} (Bandwidth). \emph{There are positive constants $c_{g,1},c_{g,2}$ such that $h,b\in[c_{g,1}g,c_{g,2}g]$, where $g\to0$ is a reference sequence such that $|J_n|\log p_n/\sqrt{ng}\to0$ and $|J_n|g^2\sqrt{\log p_n}\to0$.}

\smallskip

Assumption (BW) means that the bandwidths in the first and second stage of our procedure are such that $b\asymp h$, i.e., $b$ and $h$ converge to zero with the same speed. The exact role of this assumption is related to regularity conditions in Assumption (MS) and is discussed below. We emphasize here that (BW) can be achieved simply by selecting bandwidths from the range $[c_{g,1}g,c_{g,2}g]$. The condition that  $|J_n|\log p_n/\sqrt{ng}\to0$ is a version of a standard assumption in the Lasso literature (cf. Chapter 6 in \citet{vdGB11}), adapted to our locally penalized setup. The requirement that $|J_n|g^2\sqrt{\log p_n}\to0$ is needed to control for the bias. If the rate of the reference bandwidth $g$ is considered to be given, Assumption (BW) can be seen as imposing restrictions on the maximal number of covariates in the target set $J_n$. If, on the other hand, the growth of $|J_n|$ is considered to be given, this assumption can be interpreted as imposing limitations on the rate at which localization occurs.

\smallskip

\noindent\textbf{Assumption (AS):} (Approximate Sparsity).
\emph{It holds that $p_n\to\infty$ and, with $r_i(J_n,h)$ as in \eqref{eq:optimal_regression}, that
\begin{align}
\max_{k=1,...,p_n}\left|\IE\left(Z_{n,i}^{(k)}K_h(X_i)r_i(J_n,h)\right)\right|=O\left(\sqrt{\frac{\log p_n}{nh}}\right). \label{eq:as}
\end{align}
In addition, equation~\eqref{eq:as} remains true with $h$  replaced by $b$.}

\smallskip

Assumption (AS) is similar in nature to the notion of approximate sparsity in, for example, \citet{BCH14}.\footnote{Note that (AS) is required to make the Lasso work. Thus, it should not be read as a restriction, but it rather allows for data-dependent model selection (cf.\ our Remark \ref{rem:role_lasso} below for alternatives to (AS)). Online Appendix \ref{sec:non_sparse-simulations} provides a numerical example for the behavior of the Lasso in a non-sparse setting.}
 Note that it follows from the definition of $r_i(J_n,h)$  that $\IE(Z_{n,i}^{(k)}K_h(X_i)r_i(J_n,h))=0$ for all $k\in J_n$, and thus~\eqref{eq:as} only restricts the properties of covariates that are not part of the target set. Intuitively, (AS) means that the set $J_n$ contains ``essentially''  all relevant covariates, in the sense that any covariate which is not contained in $J_n$ is almost locally uncorrelated with the regression error $r_i(J_n,h)$. Note that in a setting with exact rather than approximate sparsity, condition~\eqref{eq:as} follows automatically as, by definition, $Z_{n,i}^{(k)}$ is uncorrelated with $(X_i,r_i(J_n,h))$ for $k\notin J_n$ in this case.

\smallskip

\noindent\textbf{Assumption (D):} (Differentiability). \emph{The density of $X_i$, $f_X$, is three times continuously differentiable in a neighborhood around zero and $f_X(0)>0$. Moreover, $\mu_Z$ is continuous and uniformly bounded in a neighborhood around zero. $\mu_Z$ and $\mu_Y$ are three times one-sided differentiable at $0$, i.e., $\mu_Z'$, $\mu_Z''$ and $\mu_Z'''$ exist on $(-\infty,0)\cup(0,\infty)$ and the left- and right sided limits at zero exist as well (and the same for $\mu_Y$). The functions $\mu_{ZZ}$ and $\mu_{ZY}$ are one-sided differentiable, and the derivatives fulfill
\begin{align*}
&\sup_{n\in\IN}\sup_{k\in\{1,...,p_n\}}\sup_{u\in[0,1]}\left|\mu_{Z^{(k)}}'(uh)\right|+\left|\mu_{Z^{(k)}}'(-uh)\right|<\infty, \\
&\sup_{n\in\IN}\sup_{k,l\in\{1,...,p_n\}}\sup_{u\in[0,1]}\left|\mu_{Z^{(k)}Z^{(l)}}'(uh)\right|+\left|\mu_{Z^{(k)}Z^{(l)}}'(-uh)\right|<\infty, \\
&\sup_{n\in\IN}\sup_{k\in J_n}\sup_{u\in[0,1]}\left|\mu_{Z^{(k)}Y}'(uh)\right|+\left|\mu_{Z^{(k)}Y}'(-uh)\right|<\infty, \\
&\sup_{n\in\IN}\sup_{k\in\{1,...,p_n\}}\sup_{u\in[0,1]}\left|\mu_{Z^{(k)}}''(uh)\right|+\left|\mu_{Z^{(k)}}''(-uh)\right|<\infty, \\
&\sup_{n\in\IN}\sup_{k\in J_n}\sup_{u\in[0,1]}\left|\mu_{Z^{(k)}}'''(uh)\right|+\left|\mu_{Z^{(k)}}'''(-uh)\right|<\infty.
\end{align*}}

\smallskip

In a finite dimensional setting, like in \citet{CCFT19}, the above conditions are implied by assuming existence and continuity of the one-sided derivatives. In the high-dimensional setting, the uniformity assumption is required in order to avoid pathological cases such as $\mu_{Z^{(k)}}$ getting increasingly steep as $k\to\infty$. Note that the conditions on the third derivative are only required on the target set $J_n$. 

To state the next assumptions, we define the matrix
$$M_n=\begin{pmatrix}
 \mu_{Z}(0) &
 \mathbf{0} &
h\mu_{Z-}' &
h\left(\mu_{Z+}'-\mu_{Z-}'\right)
\end{pmatrix}^\top\in\IR^{4\times p_n},$$
where $\mathbf{0}$ denotes a vector of zeros and put $\tilde{Z}_i=Z_i-M_n^\top V_i$.

\smallskip

\noindent\textbf{Assumption (TCS):} (Target Covariate Set).
\emph{It holds that
$$\left\|\IE\left(K_h(X_i)\tilde{Z}_i(J_n)\tilde{Z}_i^\top(J_n)\right)^{-1}\right\|_2=O(1),$$
and there are finite numbers $\delta,\sigma_l,\sigma_r,C>0$ such that
\begin{align}
&\lim_{n\to\infty}\sup_{u\in[0,1]}\left|\IE(r_i(J_n,h)^2|X_i=uh)-\sigma_r^2\right|=0, \label{eq:equicont1} \\
&\lim_{n\to\infty}\sup_{u\in[0,1]}\left|\IE(r_i(J_n,h)^2|X_i=-uh)-\sigma_l^2\right|=0, \label{eq:equicont2} \\
&\sup_{x\in[-h,h]}\sup_{n\in\IN}\left|\IE(|r_i(J_n,h)|^{2+\delta}|X_i=x)\right|<C, \label{eq:ubound1} \\
&\sup_{n\in\IN}\sup_{k\in J_n}\sup_{u\in[0,1]}\left|\IE\left(Z_i^{(k)}Y_i\big|X_i=uh\right)\right|+\left|\IE\left(Z_i^{(k)}Y_i\big|X_i=-uh\right)\right|<\infty. \label{eq:ubound2}
\end{align}
In addition, \eqref{eq:equicont1} and \eqref{eq:equicont2} also hold when $h$ is replaced by $b$.}

\smallskip

We call the requirements \eqref{eq:equicont1} and \eqref{eq:equicont2} in (TCS) equi-continuity from the right and left, respectively, and \eqref{eq:ubound1} and \eqref{eq:ubound2} are called uniform boundedness. In the proof the asymptotic normality of our final RD estimator, we use uniform boundedness to show a Lyapunov condition for the central limit theorem, and equi-continuity to ensure that the respective asymptotic variance converges to a finite and positive constant. While  boundedness seems to be unavoidable, the equi-continuity is assumed for convenience in our proofs. Removing it would potentially lead to a different convergence rate for our estimator by allowing settings where ``almost'' all variance is explained through the covariates in the limit. Given that this is not a realistic assumption, we do not consider adding this extra generality. Note  that it is necessary to distinguish the limits from left and right in  (TCS) because the conditional distribution of $r_i(J_n,h)$ given the running variable may experience a jump at zero.

\smallskip
 
\noindent\textbf{Assumption (K):} (Kernel). \emph{The kernel $K:\IR\to[0,\infty)$ integrates to one, is continuous, symmetric and is supported on $[-1,1]$.} 

\smallskip

Such conditions on the kernel are standard in the literature, and satisfied by the commonly used triangular and Epanechnikov kernels, for example. Kernels with unbounded support, like the Gaussian kernel, could be accommodated at the cost of slightly more involved theoretical arguments. Note that (K) implies that the following quantities are finite:
$$K^{(a)}=\int_{-\infty}^{\infty}u^aK(u)du,\quad K_+^{(a)}=\int_0^{\infty}u^aK(u)du, \quad a\in\{0,1,2,3,4\}.$$

For the following assumptions, we need further preliminary definitions. Let
\begin{align*}
&\mathbf{Y}=(Y_1,...,Y_n)',\quad \mathbf{K}_h=\textrm{diag}\left(K(X_1/h)/h,\ldots,K(X_n/h)/h\right), \\
&\mathbf{V}=\begin{pmatrix}
1 & T_1 & X_1/h & T_1X_1/h \\ \vdots & \vdots & \vdots & \vdots \\ 1 & T_n & X_n/h & T_nX_n/h
\end{pmatrix},\quad \mathbf{Z}(J)=\begin{pmatrix}
Z_1(J)^\top \\ \vdots \\ Z_n(J)^\top
\end{pmatrix},
\end{align*}
and for simplicity write $\mathbf{Z}=\mathbf{Z}(\{1,...,p_n\})$.

\begin{definition}
\label{def:comp_constant}
Let $c>0$ and $J\subseteq\{1,...,p_n\}$, define
\begin{equation}
\label{eq:comp_constant}
k(c,J)=\inf\frac{|J|\frac{1}{n}\left\|\mathbf{K}_b^{\frac{1}{2}}\begin{pmatrix}\mathbf{V} & \mathbf{Z}\end{pmatrix}\begin{pmatrix}\theta \\ \gamma \end{pmatrix}\right\|_2^2}{\left\|\begin{pmatrix} \theta & \gamma_{J} \end{pmatrix}'\right\|_1^2},
\end{equation}
where the infimum is taken over all vectors $(\theta\quad \gamma )'\in\IR^{p_n+4}$ for which
\begin{equation}
\label{eq:min_range}
\left\|\gamma_{J^c}\right\|_1\leq c\left\|\begin{pmatrix}\theta \\ \gamma_{J} \end{pmatrix}\right\|_1
\end{equation}
and, for any $\gamma\in\IR^{p_n}$, $\gamma_J^{(j)}=\gamma^{(j)}$ for $j\in J$ and $\gamma_J^{(j)}=0$ for $j\in J^c$. We say that the compatibility condition  $\textrm{CC}(c,J_n)$ holds for a possibly random sequence $J_n\subseteq\{1,...,p_n\}$  if $k(c,J_n)^{-1}=O_P(1)$.
\end{definition}
The constant $k(c,J)$ differs from the \emph{compatibility constants} known from the classical Lasso literature (cf.\ Chapter 6.13 in \citet{vdGB11}) only in the additional kernel weight and in the fact that the vector $\theta$ is not penalized in our setup. In order to give some intuition, we rewrite \eqref{eq:comp_constant} as follows:
\begin{equation}
\label{eq:cc_intuition}
k(c,J)=\frac{|J|}{nh}\inf\sum_{i=1}^nK\left(\frac{X_i}{h}\right)\left(V_i^\top\theta+Z_i(J)^\top\gamma_{J}-\left(-Z_i(J^c)^\top\gamma_{J^c}\right)\right)^2,
\end{equation}
where the infimum is taken over all pairs $(\theta,\gamma)$ for which \eqref{eq:min_range} and additionally $\|\theta\|_1+\|\gamma_{J}\|_1=1$ hold. Thus, $k(c,J)$ is bounded away from zero if the covariates in $J^c$ with small coefficients are unable to linearly represent the RD design vectors $V_i$ or the active covariates $Z_i(J)$.

\begin{definition}
\label{def:RSE}
Define for a sequence $m_n\in\IN$ and set $J_n$
\begin{equation}
\label{eq:rse}
\phi(m_n,J_n)=\inf\frac{\frac{1}{n}\left\|\mathbf{K}_h^{\frac{1}{2}}\begin{pmatrix} \mathbf{V} & \mathbf{Z}\end{pmatrix}\alpha\right\|_2^2}{\|\alpha\|^2}\leq\sup\frac{\frac{1}{n}\left\|\mathbf{K}_h^{\frac{1}{2}}\begin{pmatrix} \mathbf{V} & \mathbf{Z}\end{pmatrix}\alpha\right\|_2^2}{\|\alpha\|^2}=:\Phi(m_n,J_n),
\end{equation}
where $\inf$ and $\sup$ are taken over all vectors $\alpha=(\theta, \gamma)^\top\in\IR^{p_n+4}\setminus\{0\}$ such that $|\{i\in J_n^c:\,\gamma_i\neq0\}|\leq m_n$. We say that the restricted sparse eigenvalue condition $\textrm{RSE}(m_n,J_n,h)$ holds for a (random) sequence $m_n$ and a sequence of index sets $J_n$ if $\phi(m_n,J_n)^{-1}=O_P(1)$ and $\Phi(m_n,J_n)=O_P(1)$.
\end{definition}
Following the pattern of the compatibility constant in Definition~\ref{def:comp_constant}, we extend the concept of restricted sparse eigenvalues to localized problems. Continuing with the analogy, we can write down an equivalent formulation of \eqref{eq:rse} in the fashion of \eqref{eq:cc_intuition} to see that CC and RSE are similar in terms of their interpretation. The restricted sparse eigenvalue assumption is often required when it comes to Lasso estimators. See for example \citet{BCH14} for a discussion for non-localized estimators (Comment 3.2 therein) or Lemma 1 in \citet{BC13}. The localized case with the additional kernel changes the problem to a conditional instead of an unconditional variance.

\smallskip

\noindent\textbf{Assumption (RSE \& CC):} (Restricted Sparse Eigenvalues and Compatibility). \emph{The matrix $\tilde{\mathbf{Z}}(\hat{J}_n)^\top\mathbf{K}_h\tilde{\mathbf{Z}}(J_n)$ is almost surely invertible. The conditions RSE$(|J_n|\log n,J_n,h)$ for $\tilde{\mathbf{Z}}$, RSE$(|J_n|\log n,J_n,b)$, RSE$(0,J_n,h)$ and CC$(\bar{w},J_n)$ hold true for $\bar{w}=3w^{(u)}/w^{(l)}$, where  $w^{(l)}$ and $w^{(u)}$ are lower and upper bounds, respectively, on the weights $\hat{w}_{n,k}$ given in Lemma \ref{lem:weights}.}
\smallskip

Recalling the discussion after Definitions \ref{def:comp_constant} and \ref{def:RSE}, the assumption above means that $Z_i(J_n^c)$ cannot be used to represent $V_i$ or $Z_i(J_n)$. It can therefore be understood as excluding collinearity between the covariates. In order to establish standard consistency results for the Lasso (like Lemma \ref{lem:reg_error} in the Online Appendix) we only require the compatibility condition. The restricted eigenvalue assumptions are required for our asymptotic normality results, where they  guarantee that the number of selected covariates is growing slowly and that results about the model selection step carry over to the RD step. 

Recall the notation from \eqref{eq:optimal_regression}  for the following assumption.

\smallskip

\noindent\textbf{Assumption (MS):} (Model Smoothness). 
\emph{There is a sub-sequence $J_{0,n}\subseteq J_n$, a sequence $\eta_n\to\infty$ and a constant $C>0$ such that for any bandwidth $\mathcal{H}$ which fulfills $c_{g,1}g\leq\mathcal{H}\leq c_{g,2}g$, with $g,c_{g,1},c_{g,2}$ are as in (BW), we have that for any $k\in J_{0,n}$
\begin{equation}
\label{eq:vic}\left|\gamma_0^{(k)}(J_n,\mathcal{H})\right|\geq\eta_n\sqrt{\frac{|J_n|\log p_n}{ng}}
\end{equation}
and for any $k\in J_n\setminus J_{0,n}$
\begin{equation}
\label{eq:vuc}
\left|\gamma_0^{(k)}(J_n,\mathcal{H})\right|\leq C\sqrt{\frac{|J_n|\log p_n}{ng}}.
\end{equation}
Moreover,
$$\sqrt{|J_n\setminus J_{0,n}|}\cdot\frac{|J_n|\log p_n}{\sqrt{ng}}\to0\quad\textrm{and}\quad\sqrt{|J_n\setminus J_{0,n}|}\cdot|J_n|g^2\sqrt{\log p_n}\to0.$$}

This assumption rules out pathological settings in which a covariate's relevance within the target set is strongly affected by minor changes of the bandwidth.\footnote{Note that \eqref{eq:vic} and \eqref{eq:vuc} are mutually exclusive but not exhaustive. It would be possible to formulate mutually exclusive conditions by introducing a constant $C_0>0$ which depends on many unknown quantities. This extra generality would thus bring no meaningful practical benefit.} To see this, fix $\mathcal{H}$ and consider sets $J_{0,n}(\mathcal{H})\subseteq J_n$ such that \eqref{eq:vic} holds for $k\in J_{0,n}(\mathcal{H})$ while for $k\in J_n\setminus J_{0,n}(\mathcal{H})$, \eqref{eq:vuc} is true. Assumption (MS) reads then as: The mapping $\mathcal{H}\mapsto J_{0,n}(\mathcal{H})$ is constant for each $n\in\IN$. In other words, the identity of the covariates with large population regression coefficients in $J_n$ remains the same when the bandwidth is slightly altered. There might also be covariates in the target set $J_n$ with relatively small coefficients, namely those for which \eqref{eq:vuc} holds. Therefore, assumption (MS) is different from a $\beta$-min assumption \citep[cf.][]{vdGB11}.

For the next assumption, we define for $k\in\{1,...,p_n\}$ and $m\in\IN$:
$$\mu_{k,m}(x)=\IE\left(\left|Z_i^{(k)}\right|^m\Big|X_i=x\right)\textrm{ and }\mu_{k,m}^{(r)}(x)=\IE\left(\left|Z_i^{(k)}r_i(J_n,b)\right|^m\Big|X_i=x\right)$$

\noindent\textbf{Assumption (CTB):} (Covariate Tail Behavior).
\emph{The functions $\mu_{k,1},\mu_{k,2}$ and $\mu_{k,1}^{(r)}$ are uniformly bounded in a neighborhood around zero. There are finite numbers $\sigma_a^2,c_a,c_a^*>0$ for $a=0,1,2$ such that for all  $m\in\IN$ 
\begin{align}
\int_{\IR}\left(1+|u|^m\right)K(u)^m\mu_{k,m}(ub)f_X(ub)du\leq&\frac{m!}{2}\sigma_0^2c_0^{m-2}, \label{eq:B0} \\
\int_{\IR}K(u)^{2m}\mu_{k,2m}(ub)f_X(ub)du\leq&\frac{m!}{2}\sigma_1^2c_1^{m-2}, \label{eq:B2} \\
\int_{\IR}(1+|u|^m)K(u)\mu^{(r)}_{k,m}(ub)f_X(ub)du\leq&\frac{m!}{2}\sigma_2^2c_2^{m-2}. \label{eq:B4} \end{align}
Equations \eqref{eq:B0} and \eqref{eq:B4} hold also when $b$ is replaced by $h$.}

\smallskip

Conditions \eqref{eq:B0} and \eqref{eq:B2} hold, for example, for bounded covariates, or covariates which fulfill a local sub-Gaussianity condition. Similarly, \eqref{eq:B4} is implied by a sub-gaussianity or boundedness condition on the covariates and the residuals. We use these tail-constraints to prove a Bernstein-type concentration result (cf.\ Proposition~\ref{prop:bernstein} and the reference there for a general statement and Lemma~\ref{lem:expbound} for a formulation of the statement which is tailored to our setting). Condition \eqref{eq:B2} is a specific requirement for the model selection step and is thus formulated in terms of the model selection bandwidth~$b$.

%\alex{Ich ein paar Integrale in (CTB) durch Stetigkeit ersetzt. Dadurch wird es etwas kürzer. Außerdem tauchen $\tilde{Z}_i$ nicht mehr auf.}

\smallskip

\noindent\textbf{Assumption (CV):} (Covariate Variance).
\emph{It holds that
\begin{align*}
&\min_{n\in\IN}\min_{k\in\{1,...,p_n\}}\IE\left(\frac{1}{b}K\left(\frac{X_i}{b}\right)^2\left(Z_i^{(k)}\right)^2\right)>0, \\
&\max_{n\in\IN}\max_{k\in\{1,...,p_n\}}\IE\left(\frac{1}{b}K\left(\frac{X_i}{b}\right)^2\left(Z_i^{(k)}\right)^2\right)<\infty.
\end{align*}}

\smallskip

This assumption ensures that no single covariate has, asymptotically, either a negligible or dominating variance, and thus  that all covariates have a roughly similar scale.

\subsection{Main Result and Discussion}
\label{subsec:result}
The main result of this paper shows asymptotic normality of our estimator under the conditions stated in Section~\ref{subsec:assumptions}. In order to state it, we define the following constants, which depend on the kernel function only:
\begin{align*}
	C_\mathcal{B}=&\frac{K_+^{(3)}-2K_+^{(1)}K_+^{(2)}}{K_+^{(2)}-2\left(K_+^{(1)}\right)^2},\quad C_{\mathcal{S}}=\frac{(K^2)^{(0)}(K_+^{(2)})^2+(K^2)_+^{(2)}(K_+^{(1)})^2-2(K^2)_+^{(1)}K_+^{(1)}}{\left[(K_+^{(1)}-\frac{1}{2}K_+^{(2)}\right]^2}.
\end{align*}

\begin{theorem}
\label{thm:main}
Suppose the assumptions from Section \ref{subsec:assumptions} hold, and that the penalty parameter $\lambda$ is chosen such that $\lambda=O(\sqrt{\log p_n / (ng)})$. Then there are sequences $\mathcal{B}_n$ and $\mathcal{S}_n$  satisfying
\begin{align}
\mathcal{B}_n =&\frac{C_\mathcal{B}}{2}\left(\mu_{\tilde Y+}''-\mu_{\tilde Y-}''\right)+o\left(|J_n|^{1/2}\right)\textnormal{ and} \label{eq:bias_formula} \\
\mathcal{S}^2_n=&\frac{C_\mathcal{S}}{f_X(0)}\left(\sigma^2_{\tilde Y+} +\sigma^2_{\tilde Y-} \right)+o(1) \label{eq:var_formula}
\end{align}
such that~\eqref{eq:main_result} holds:
$$\frac{\sqrt{nh}\left(\hat{\tau}_h(\hat J_n)-\tau_Y  -h^2\mathcal{B}_n\right)}{\mathcal{S}_n}\overset{d}{\to}\mathcal{N}\left(0,1\right).$$
\end{theorem}

The theorem's proof is given in Online Appendix \ref{app:overview}. The following remarks discuss its implications and possible extensions.

\begin{remark}[Asymptotic Bias]\normalfont
\label{rem:asymptotic_bias}
	The first term  in the expansion of the asymptotic bias $\mathcal{B}_n $
	in~\eqref{eq:bias_formula} is proportional to
	$$\mu_{\tilde Y+}''-\mu_{\tilde Y-}'' = (\mu_{Y+}''-\mu_{Y-}'') - (\mu_{Z(J_n)+}''-\mu_{Z(J_n)-}'')^\top \gamma_n \equiv A - B_n.$$
	The term $A=\mu_{Y+}''-\mu_{Y-}''$ is the one we would obtain for the baseline estimator, and the term $B_n=(\mu_{Z(J_n)+}''-\mu_{Z(J_n)-}'')^\top \gamma_n$  captures the covariates' contribution to the bias. Depending on the curvature of the components of $\mu_{Z(J_n)}$, in theory our estimator's overall asymptotic bias could thus be larger or smaller than that of the baseline estimator; and since $B_n$  could potentially be of order $O(|J_n|^{1/2})$
	our estimator's overall asymptotic bias
	  could also vanish at the same or a slightly slower rate than that of the baseline estimator.
	  
	We would argue that one should not be too concerned though that including covariates could increase the  bias in an empirical context. In particular, the general notion of covariates being predetermined prior to treatment assignment is plausibly compatible with a strengthening of~\eqref{eq:tauz} that assumes that not only the levels but also the second derivatives of $\mu_Z$ are continuous around the cutoff. In this case, we have $B_n=0$, and the leading bias in~\eqref{eq:bias_formula} simplifies to that of the baseline RD estimator.
	
	We note that there are some pathological cases in which the first term in~\eqref{eq:bias_formula} would actually be of smaller order than the remainder term, and thus not be leading. The following Lemma provides conditions that rule this out.
	\begin{lemma}
		\label{lem:bias_formula}
		Suppose that, in addition to the assumptions of Theorem \ref{thm:main}, there is a constant $\eta>0$ such that
		\begin{equation}
			\label{eq:bias1}
			\left|\left(\mu_{Z(J_n)+}''-\mu_{Z(J_n)-}''\right)^\top \gamma_n-\left(\mu_{Y+}''-\mu_{Y-}''\right)\right|\geq\eta
		\end{equation}
		and, if $\|\mu_{Z(J_n)+}''-\mu_{Z(J_n)-}''\|_2\to\infty$, there is a constant $c>0$ such that
		\begin{equation}
			\label{eq:bias2}
			\|\mu_{Z(J_n)+}''-\mu_{Z(J_n)-}''\|_2\leq c\left|\left(\mu_{Z(J_n)+}''-\mu_{Z(J_n)-}''\right)^\top\gamma_n\right|.
		\end{equation}
		Then
		\begin{equation}
			\label{eq:bias_good_formula}
			\mathcal{B}_n=\frac{C_\mathcal{B}}{2}\left(\mu_{\tilde Y+}''-\mu_{\tilde Y-}''\right)(1+o\left(1\right)).
		\end{equation}
	\end{lemma}
 Requirements \eqref{eq:bias1} and \eqref{eq:bias2} are the natural extensions of the standard assumption that $\mu_{Y+}''-\mu_{Y-}''\neq0$ to our setting. They exclude pathological cases in which the covariates' contribution to the bias $B_n$ and the ``no covariates'' component $A$ happen to  cancel each other asymptotically, and cases in which some of the components of the vector  $\mu_{Z(J_n)+}''-\mu_{Z(J_n)-}''$ are large, but happen to exactly offset each other such that their linear combination $B_n$ vanishes.
\end{remark}

\begin{remark}[Asymptotic Variance]\normalfont
	\label{rem:asymp_variance}
	The leading term of our estimator's asymptotic variance $\mathcal{S}_n^2$ converges to a positive constant under our assumptions, and 	is guaranteed not to exceed that of the baseline estimator from~\eqref{eq:Base}, or that of any estimator of the form in~\eqref{eq:fixedJ} that uses a strict subset of the the target set $J_n$. This can be seen by noting first that  it depends on the covariates only through the term
	$$\sigma^2_{\tilde Y+} +\sigma^2_{\tilde Y-} = \lim_{x\downarrow 0} \textrm{Var}(Y_i-Z_i(J_n)^{\top}\gamma_n|X_i=x)+\lim_{x\uparrow 0}\textrm{Var}(Y_i-Z_i(J_n)^{\top}\gamma_n|X_i=x);$$
	and second that $\gamma_n$ minimizes the function
	$$\gamma\mapsto \lim_{x\downarrow 0}\textrm{Var}(Y_i-Z_i(J_n)^{\top}\gamma|X_i=x)+\lim_{x\uparrow 0}\textnormal{Var}(Y_i-Z_i(J_n)^{\top}\gamma|X_i=x).$$
	This holds because the function is  quadratic in $\gamma$, has a positive
	definite Hessian matrix, and its Jacobian is set to zero by $\gamma_n=(\lim_{x\downarrow 0}\textrm{Var}(Z_i(J_n)|X_i=x)+\lim_{x\uparrow 0}\textnormal{Var}(Z_i(J_n|X_i=x))^{-1}(\lim_{x\downarrow 0}\textrm{Cov}(Y_i,Z_i(J_n)|X_i=x)+\lim_{x\uparrow 0}\textnormal{Cov}(Y_i,Z_i(J_n)|X_i=x))$, as defined above.\footnote{The remarks in \citet[][Section IV.B]{CCFT19} seem to suggest that this result should only hold under additional conditions, but this does not appear to be the case here.}
	
	At first glance, this result might seem to suggest that a sequence of larger target sets must always lead to a smaller asymptotic variance. The following Lemma  shows that this is not the case. Indeed, it shows the stronger statement that if there are two covariate sets $J_{1,n}, J_{2,n}\subseteq\{1,...,p_n\}$ that both satisfy our assumptions, the corresponding estimators of the form in~\eqref{eq:fixedJ} must have the same  asymptotic variance.

\begin{lemma}
\label{lem:optimal_variance}
Let $J_{1,n}, J_{2,n}\subseteq\{1,...,p_n\}$ be two sequences of covariate indices such that
$$\left\|\IE\left(K_h(X)\left(V_i^\top \quad Z_i(J_{1,n}\cup J_{2,n})^\top\right)^\top\left(V_i^\top \quad Z_i(J_{1,n}\cup J_{2,n})^\top\right)\right)^{-1}\right\|_2=O(1)$$
and let $|J_{j,n}|\log p_n/nh\to0$ for $j=1,2$. Suppose Assumptions (AS) and (TCS) hold with $J_n$ replaced by either $J_{1,n}$ or $J_{2,n}$, and let $\mathcal{S}_{j,n}^2$ be the asymptotic variance of $\hat\tau_h(J_{j,n})$ as in Theorem~\ref{thm:main}, for $j=1,2$. Then $\mathcal{S}_{1,n}^2-\mathcal{S}_{2,n}^2=o(1)$.
	\end{lemma}

	We conjecture that even in a setting in which (AS) does not hold for any target set, our estimator continues to have the smallest asymptotic variance  among all linear adjustment estimators that only use a moderate (in some appropriate sense) number of the available covariates, given some suitable choice of $\lambda$. However, proving such a result would require developing a ``non-sparse'' theory for the Lasso, which is beyond the scope of this paper.

\end{remark}

\begin{remark}[Double Selection] \normalfont As an alternative to our proposed estimator, one could also consider a ``double selection''  procedure which, as in \citet{BCH14}, additionally includes covariates that are predictive for treatment status into the active set. That is, one could redefine the set $\hat{J}_n$ as $\hat{J}_n=\{k\in\{1,...,p_n\}:\tilde{\gamma}_n^{(k)}\neq 0 \textnormal{ or }\bar{\gamma}_n^{(k)}\neq 0\}$, where
	$$\bar{\gamma}_n=\argmin{\gamma\in \IR^{p_n}}\sum_{i=1}^nK_b(X_i)\left(T_i-\left(Z_i-\hat{\mu}_{Z,n}\right)^{\top}\gamma\right)^2+\lambda\sum_{j=k}^{p_n}\hat{w}_{n,k}|\gamma_k|,$$
	and $\tilde{\gamma}_n$ is as defined above. This change would not affect the properties of the final RD estimator, however, as under~\eqref{eq:tauz} the covariates are not informative about treatment status among units local to the cutoff (a conceptually similar phenomenon appears in randomized experiments).
	For analogous reasons, the large sample properties of our final RD estimator would also be the same as that of an alternative procedure that excludes the predictors $V_i$ from the model selection step (we choose to include them because that  substantially simplifies the algebra in some steps of our proofs).
\end{remark}

\begin{remark}[Role of the Lasso]
\label{rem:role_lasso}\normalfont
	Our assumptions do not imply that  the Lasso recovers the target set with very high probability, in the sense that $\IP(\hat{J}_n=J_n)\to1$. Existing results suggest that such a property could only be established under substantially stronger conditions, including exact rather than approximate sparsity and a so-called ``$\beta$-min'' condition that imposes a substantial lower bound on the values of all non-zero coefficients (cf. Section 2.6 in \citet{vdGB11} and the references therein).	
	Consistent estimation of the target set is also not required for our results. In fact, similarly as in \citet{BC13}, we only require that the covariates obtained via the model selection step lead to regression residuals that are ``similar'' to those corresponding to the target set. Precise requirements are given in Assumption (CMS) in the Online Appendix, and we show in the proof of Theorem \ref{thm:main} that these are satisfied by the Lasso. Any model selection procedure that satisfies Assumption (CMS) could be used instead of the Lasso in our procedure as well. The simplest example being the \emph{ad-hoc} selection of a pre-defined set of covariates in which case $p_n\equiv p_0$ and $\hat{J}_n\equiv\{1,...,p_0\}$.
\end{remark}

\section{Numerical Results}
\label{sec:numerical_results}
\subsection{Implementation}
\label{sec:implementation}

In order to implement our proposed method in practice, one has to choose the initial bandwidths $b$, the Lasso penalty $\lambda$, and the final bandwidth $h$. As discussed in Section~2, the latter can in principle be chosen by applying any approach deemed suitable for settings with low-dimensional covariates to the generated data set  $\{(Y_i,X_i,Z_i(\hat J_n)), i=1,\ldots,n\}$, such as those proposed by  \citet{CCFT19} or  \citet{AK18}. We conduct our simulations using both frameworks, with results based on the methods in  \citet{AK18} reported in this section, and results based on the methods in \citet{CCFT19} reported in Online Appendix \ref{sec:rdrobust_empirics}.

The choice of $b$ and $\lambda$ is complicated by the fact that these quantities do not appear in the limiting distribution of our final RD estimator. Our heuristic recommendation is to use a method for bandwidth choice designed for settings without covariates, like the ones proposed in  \citet{IK12}, \citet{calonico2014robust} or \citet{AK18} to select $b$; and we focus on the method proposed in \citet{AK18} in this section. We also consider choosing $\lambda$ via adaptions of three methods for non-localized Lasso estimators to our RD setting: standard cross-validation, the plug-in procedure of \citet{BCH14}, and a recently proposed bootstrap-based method by  \citet{VL20}. The three procedures are described formally in Online Appendix \ref{sec:app_implementation}, and we refer to them by the acronyms (CV), (BCH) and (LV), respectively, below.  %We note that adapting the  method of \citet{BCH14} requires a slightly different choice of weights $\hat{w}_{n,k}$, but this does not cause any major issues.
Our computations in this section use the {\tt R} packages  {\tt glmnet} for implementing Lasso-based covariate selection, and {\tt RDHonest} for bandwidth selection, standard errors, and confidence intervals.

\subsection{Simulations}
\label{sec:simulations}
For the simulations, we consider the following DGP, which is a variation of ``Model~2'' in \citet{CCFT19} and corresponds to an RD setting with $p=200$ covariates and parameter of interest $\tau_Y=0.02$:
\begin{align*}
X&\sim 2\cdot\textrm{beta}(2,4)-1, \quad T=\Ind(X\geq0), \quad (\epsilon,  Z^\top)^\top\sim \mathcal{N}\left(\mathbf{0},\Sigma\right), \quad \Sigma=\begin{pmatrix}
	\sigma_ {\epsilon}^2 & v^\top \\ v & \sigma_Z^2I_{200}
\end{pmatrix},\\
Y&=\epsilon+\left\{\begin{array}{lll}
0.36+0.96\cdot X+5.47\cdot X^2 && \\ 
\quad\quad\quad\quad+15.28\cdot x^3+15.87\cdot X^4+5.14\cdot X^5+0.22\cdot Z^\top\alpha,& \textrm{if }T=0, \\
0.38+0.62\cdot X-2.84\cdot X^2 && \\
\quad\quad\quad\quad+8.42\cdot X^3-10.24\cdot X^4+4.31\cdot X^5+0.28\cdot Z^\top \alpha,& \textrm{if }T=1,\end{array}\right.
\end{align*}
with $\sigma_{\epsilon}^2=0.1295^2$, $\sigma_Z^2=0.1353^2$, $I_{200}$ denoting the $200\times200$ identity matrix, $v\in\IR^{200}$ a vector whose $k$th component is equal to $v_k=0.8\sqrt{6}\sigma_{\epsilon}^2/\pi k$, and  $\alpha\in\IR^{200}$ a vector whose $k$th component is equal to $\alpha_k=2/k^2$.  This choice of $\alpha$ implies that the $k$th covariate $Z^{(k)}$ becomes less important for variance reduction as $k$ increases. 
We  consider the sample size $n=1,000$ and set the number of Monte Carlo replications to $10,000$. %We set the smoothness bound   required by {\tt RDHonest} to $40$.

We report results for our post-Lasso procedure with the penalty parameter $\lambda$ selected  via either of (CV), (BCH) and (VL). For comparison, we also consider linear adjustment estimators that use different fixed subsets of the covariates, namely either no covariates, only the first covariate, only the first 10 covariates, only the first 30 covariates, only the first 50 covariates, or only the ``optimal'' linear combination of covariates  $Z^\top\alpha$. Note that the covariates are ordered by their ``importance'' in our DGP, the procedures that use a fixed non-zero number of covariates are  infeasible, as in practice the econometrician would generally not know which covariates are the most important ones. The ``optimal covariate'' estimator is not feasible as well, as the vector $\alpha$ is generally unknown in applications. These estimators serve as (oracle) performance benchmarks in our simulation study.
 
  \begin{table}[!t]
 	\centering\caption{Simulation Results}\label{tab1}\small
 	\begin{tabular}{l|cccccc}
 		\toprule
 		Covariate Selection &  \#Cov. & Bias  & SD    & Avg.\ SE &  CI Length & Coverage           \\
 		\toprule
 		Lasso (CV)   & 9.5 & 0.0054 & 0.0464 & 0.0325 & 0.1559 & 87.7 \\ 
 		Lasso (BCH)  & 1.2 & 0.0051 & 0.0512 & 0.0488 & 0.2189 & 96.1 \\ 
 		Lasso (LV)   & 1.9 & 0.0047 & 0.0482 & 0.0445 & 0.2005 & 95.6 \\ 
 		\midrule
 		Fixed: No Covariates     & -- & 0.0070 & 0.0744 & 0.0735 & 0.3278 & 96.7 \\ 
 		Fixed: Most Important Covariate Only        & -- & 0.0050 & 0.0516 & 0.0499 & 0.2236 & 96.3 \\ 
 		Fixed: 10 Most Important Covariates   & -- & 0.0044 & 0.0427 & 0.0364 & 0.1670 & 94.6 \\ 
 		Fixed: 30 Most Important Covariates  & -- & 0.0055 & 0.0447 & 0.0282 & 0.1390 & 87.5 \\ 
 		Fixed:  50 Most Important Covariates  & -- & 0.0063 & 0.0487 & 0.0213 & 0.1176 & 77.0 \\ 
 		Fixed: Optimal Covariate & -- & 0.0042 & 0.0441 & 0.0424 & 0.1900 & 96.2 \\ 
 		\bottomrule
 	\end{tabular}
 
 	\raggedright \footnotesize Results based on $10,000$ Monte Carlo replications. For each estimator, the table shows shows average number of selected covariates (\#Cov.), the bias (Bias), the standard deviation (SD), the average value of the final estimator's standard error (SE), the average length of the corresponding confidence interval for the parameter of interest (CI Length), and the share of simulation runs in which the respective confidence interval covered the true parameter value (Coverage). 	
 \end{table}

  Our simulation results are summarized in Table~\ref{tab1}. Regarding our proposed procedures, we see that choosing the Lasso penalty via~(CV) leads to substantially more covariates being selected relative to~(BCH) or~(LV), with the latter being roughly similar. All Lasso-based estimators have similarly low bias and similar empirical standard deviations; and the latter are both lower than that of the ``no covariates'' baseline, and close to that of the ``optimal covariate'' oracle estimator. Standard errors and CI coverage are accurate with (BCH) and (LV). However, for (CV) the standard errors notably underestimate the true standard deviation, which leads to slight CI undercoverage. This phenomenon seems to occur because the slightly larger number of covariates that (CV) tends to select already leads to overfitting in the post-Lasso stage.
  
  A similar effect occurs with the linear adjustment estimator that uses fixed sets of the covariates. While the performance of this estimator is good in our simulations if only the single most important covariate is used, with the 10, 30 or 50 most important covariates we see a progressively severe downward bias in the standard error, with corresponding CI undercoverage.
   
  Overall, the simulation results are in line with our asymptotic theory, and show that our procedures can obtain near-oracle performance in practice. They also highlight the need for working with a small number of covariates to obtain reliable inference, and thus the need for covariate selection even if the number of available covariates is only moderate relative to the sample size.

\subsection{Empirical Application}
\label{subsec:application}

In this section, we apply our methodology to data on Austrian workers from \citet{CCW07}, to whom we refer for an extensive description of its construction. During the sample period, workers are eligible for severance payments when losing their job if they have at least 36 months of job tenure at the time of separation.
 One part of the analysis in \citet{CCW07} concerns the question whether severance payments lead to higher wages in future jobs (by enabling workers to search longer for a new position, and thus find better matches). We use our method to reanalyze this question, taking previous job tenure as the  running variable, with a cutoff at 36 months, and the difference in log wages between old and new jobs as the outcome. The data include a large number of covariates containing information about workers' socio-demographic characteristics and the nature of their employment. We select 60 of these covariates, and  split them into a basic and an extended set as follows:
\begin{description}
	\item[Basic Covariates:] Gender, marital status, Austrian nationality, ``blue collar'' occupation, age and its square, $\log$ of previous wage and its square, indicators for month and year of job termination (38 covariates)
	\item[Additional Covariates:] Work experience and its square, number of employees in firm at job just lost, indicator of having a job before the one just lost, ``blue collar'' status at job prior to the one lost, indicator of having a prior spell of nonemployment, duration of last nonemployment, total number of spells of nonemployment in career, indicator of being recalled to the job before the one just lost, indicator for higher education, indicators for industry sector and region (22 covariates)
\end{description}
We also create further covariates by including all non-trivial interaction terms and trigonometric series transformations of all non-dummy variables, which are of the form $\sin(2\pi k\times\textrm{variable})$ and $\cos(2\pi k\times\textrm{variable})$ for $k=1,...,5$. This results in a total of 1,958 covariates. After removing all observations with at least one missing covariate value, data on 288,175 workers is available for the empirical analysis. We then compute an estimate of the RD parameter, with associated standard error and confidence intervals. In view of our simulation results, we only consider (BCH) for selecting the penalty parameter. We compare the result to those based on the baseline estimator or linear adjustment estimators that either use only the basic set of covariates, or both the basic and the additional set of covariates. 

The results are shown in Table~\ref{tab3}. All four methods produce similar point estimates close to zero, which is in line with the results in \citet{CCW07}. Our method's standard error is about $11\%$ lower than that of the baseline estimator without covariates, showing that the use of a small number of carefully selected covariates can meaningfully improve estimation accuracy. The corresponding confidence interval is also shorter by a similar factor. Importantly, (BCH) only selects three of the 1,958 covariates, all of which are interactions of the squared logarithm of the previous wage with some other variable. The results for the two remaining linear adjustment estimators show that controlling for a larger number of covariates does not yield meaningfully smaller standard errors, which suggest that approximate sparsity is a reasonable assumption in this context. In view of our simulation results above, there  also is concern that using 38 or 60 covariates could lead to downward biased standard errors for the corresponding linear adjustment estimators.

\begin{table}
\centering\caption{Estimation Results}\label{tab3}
\begin{tabular}{l|c|cc|ccc}
\toprule
            &&  &      & \multicolumn{3}{c}{Confidence Interval}    \\
   Covariate Selection         &\#Cov.     &  Estimator         &   SE     & Lower   & Upper  & Length \\
\toprule
Lasso (BCH)       & 3            & 0.0387 & 0.0207 & -0.0066 & 0.0839 & 0.0905 \\ 
\midrule
Fixed: None                  &0 & -0.0063 & 0.0233 & -0.0570 & 0.0443 & 0.1013 \\ 
Fixed: Basic only   &38         & 0.0116 & 0.0200 & -0.0322 & 0.0554 & 0.0876 \\ 
Fixed: Basic + Additional&60  & 0.0168 & 0.0203 & -0.0277 & 0.0612 & 0.0890 \\ 

\bottomrule
\end{tabular}

	\raggedright \footnotesize Results based on 288,175 observations. Column \#Cov. shows the number of selected covariates. Based on our simulation results above, there are concerns that SEs for estimators ``Fixed: Basic only'' and ``Fixed: Basic + Additional'' could be downward biased.
\end{table}

\section{Concluding Remarks}

Our results on sharp RD estimation with a potentially large number of covariates can be extended to other settings, such as fuzzy RD or regression kink (RK) designs. In fuzzy RD designs, for instance, units are assigned to treatment if their realization of the running variable falls above the threshold value, but they do not necessarily comply with this assignment. The conditional treatment probability hence jumps at the cutoff, but in contrast to sharp RD designs it generally does not jump from zero to one.  The parameter of interest in fuzzy RD designs is 
	$$
	\tau_{\textnormal{fuzzy}} = \frac{\tau_Y}{\tau_T}   \equiv \frac{\mu_{Y+} - \mu_{Y-} }{\mu_{T+} - \mu_{T-}},
	$$
	which is the ratio of two sharp RD estimands. It can be estimated by running our proposed procedure twice, once with $Y_i$ and once with $T_i$ as the dependent variable, and taking the ratio of the two estimates. If the assumptions from Section \ref{subsec:assumptions} also hold with $T_i$ replacing $Y_i$, and $\tau_T$ is bounded away from zero, the asymptotic normality of the resulting estimator of $\tau_{\textnormal{fuzzy}}$ simply follows from the delta method.

\bibliography{mybib}{}

\newpage
\appendix

\section{Proofs}
\subsection{Notation and Overview}
\label{app:overview}
For subsets $J\subseteq\{1,...,p_n\}$ we denote by $|J|$ its size and for an arbitrary vector $a\in\IR^{p_n}$ we denote its restricted version by $a_J$, i.e., $a_J^{(i)}=a^{(i)}$ for $i\in J$ and $a_J^{(i)}=0$ for $i\notin J$, where $a^{(i)}$ refers to the $i$-th entry of $a$. The restricted version of $Z_i$ will be denoted by $Z_i(J)$. Furthermore,
\begin{align*}
\mathbf{Y}=&(Y_1,...,Y_n)^{\top},\quad \mathbf{K}_h=\textrm{diag}\left(h^{-1}K(X_1h^{-1}),...,h^{-1}K(X_nh^{-1})\right), \\
\mathbf{V}=&\begin{pmatrix}
1 & T_1 & X_1/h & T_1X_1/h \\ \vdots & \vdots & \vdots & \vdots \\ 1 & T_n & X_n/h & T_nX_n/h
\end{pmatrix},\quad \mathbf{Z}(J)=\begin{pmatrix}
Z_1(J)^{\top} \\ \vdots \\ Z_n(J)^{\top}
\end{pmatrix}, \\
\mathbf{r}_n(J_n,h)=&(r_1(J_n,h),...,r_n(J_n,h))^{\top},
\end{align*}
where $\textrm{diag}(v)$ denotes a diagonal matrix which diagonal is given by the vector $v$. In this notation, the (constrained to $J$) minimizer in the $\argmin{}$ of \eqref{eq:LS} is given by
$$\left(\hat{\theta}_n(J),\hat{\gamma}_n(J)\right)=\underset{\gamma_{J^c=0}}{\argmin{(\theta,\gamma)\in\IR^{4+p_n}}}\frac{1}{n}\left\|\mathbf{K}_h^{\frac{1}{2}}\left(\mathbf{Y}-\mathbf{V}\theta_0-\mathbf{Z}\gamma\right)\right\|^2$$
and $\hat{\tau}_n(J)$ is the second entry of $\hat{\theta}_n(J)$. Here $\mathbf{K}_h^{\frac{1}{2}}$ denotes a diagonal matrix whose diagonal equals the square roots of the corresponding diagonal elements of $\mathbf{K}_h$. Note that $\hat{\tau}_n=\hat{\tau}_n(\hat{J}_n)$ and $\hat{\theta}_n=\hat{\theta}_n(\hat{J}_n)$. Moreover, we can find an explicit formula for $\hat{\theta}_n(J)$:
\begin{align*}
\hat{\theta}_n(J)&=\left[\left(\mathbf{V}-\mathbf{Z}(J)\hat{\gamma}_{V,n}(J)\right)^{\top}\mathbf{K}_h\left(\mathbf{V}-\mathbf{Z}(J)\hat{\gamma}_{V,n}(J)\right)\right]^{-1} \\
&\quad\quad\quad\quad\times\left(\mathbf{V}-\mathbf{Z}(J)\hat{\gamma}_{V,n}(J)\right)^{\top}\mathbf{K}_h\left(\mathbf{Y}-\mathbf{Z}(J)\hat{\gamma}_n^*(J)\right), \\
\hat{\gamma}_n^*(J)&=\left(\mathbf{Z}(J)^{\top}\mathbf{K}_h\mathbf{Z}(J)\right)^{-1}\mathbf{Z}(J)^{\top}\mathbf{K}_h\mathbf{Y}, \quad \hat{\gamma}_{V,n}(J)=\left(\mathbf{Z}(J)^{\top}\mathbf{K}_h\mathbf{Z}(J)\right)^{-1}\mathbf{Z}(J)^{\top}\mathbf{K}_h\mathbf{V}.
\end{align*}
Note that $\hat{\gamma}_n^*(J)$ is the regression coefficient in a regression of $Y_i$ on $Z_i(J)$, and $\hat{\gamma}_{V,n}(J)$ is a matrix of regression coefficients from a regression of each component of the vector $V_i=(1,T_i,X_i/h,T_iX_i/h)^{\top}$ on $Z_i(J)$.

In the proof of Theorem \ref{thm:main} we will begin by generalizing the results from \citet{CCFT19} to a setting where there is a large (possibly larger than $n$) set of covariates available of which we choose  a subset which can be reasonably handled with the given number of observations in a data driven way. Therefore, we have to extend the considerations from \citet{CCFT19} to the case where the covariates are a random subset of a large covariate set which grows as the number of observations grows. While doing this, we do not specify the exact model selection algorithm but we study the asymptotic behavior of a general model selection procedure. In the proof of Theorem \ref{thm:main} we will formulate conditions that a general model selection procedure has to fulfill and we will show that the Lasso procedure which we introduced in the main text has these properties.

\begin{proof}[Proof of Theorem \ref{thm:main}]
We first show how we can reduce the situation to the case where   $\mu_Z$ is continuously differentiable. Recall that $\tilde{Z}_i=Z_i-M_n^\top V_i$ and define $\tilde{\mathbf{Z}}$ like $\mathbf{Z}$ but where $Z_i$ is replaced by $\tilde{Z}_i$. We can directly see that
\begin{align*}
\mu_{\tilde{Z}}(x)=\IE(\tilde{Z}_i|X_i=x)=&\mu_Z(x)-\mu_{Z-}-x\mu_{Z-}'-x\Ind(x\geq0)(\mu_{Z+}'-\mu_{Z-}')
\end{align*}
is differentiable with $\mu_{\tilde{Z}}(0)=\mu_{\tilde{Z}}'(0)=0$. We compare an estimator based on $V_i$ and $\tilde{Z}_i$,
$$\left(\check{\theta}_n,\check{\gamma}_n\right)=\argmin{(\theta,\gamma)}\sum_{i=1}^nK_h(X_i)\left(Y_i-V_i^{\top}\theta-\tilde{Z}_i^{\top}\gamma\right)^2,$$
with our estimator (note that it is no problem that we have here $V_i$ rather than $V_i$ in \eqref{eq:fixedJ} because only the second component of $\hat{\theta}_n$ will be of interest later):
\begin{align*}
\left(\hat{\theta}_n,\hat{\gamma}_n\right)=&\argmin{(\theta,\gamma)}\sum_{i=1}^nK_h(X_i)\left(Y_i-V_i^{\top}\theta-Z_i^{\top}\gamma\right)^2 \\
=&\argmin{(\theta,\gamma)}\sum_{i=1}^nK_h(X_i)\left(Y_i-V_i^{\top}\left(\theta+M_n\gamma\right)-\tilde{Z}_i^{\top}\gamma\right)^2,
\end{align*}
where all $\textrm{argmin}$ above are over the set of all $(\theta,\gamma)$ with $\gamma_{\hat{J}_n^c}=0$. Since both estimators are otherwise unconstrained, it is true that the optimal values of the two objective functions above are identical. From this we conclude that
$$\check{\gamma}_n=\hat{\gamma}_n\quad\textrm{and}\quad\check{\theta}_n=\hat{\theta}_n+M_n\hat{\gamma}_n.$$
If the optima are not unique, such solutions exist which we keep for the remainder of the proof.

In a completely analogous fashion we also compare the population quantities
$$\left(\check{\theta}_0(J_n,h),\check{\gamma}_0(J_n,h)\right)=\argmin{(\theta,\gamma)}\,\IE\left(K_h(X_i)\left(Y_i-V_i^{\top}\theta-\tilde{Z}_i^{\top}\gamma\right)^2\right)$$
and
$$\left(\theta_0(J_n,h),\gamma_0(J_n,h)\right)=\argmin{(\theta,\gamma)}\,\IE\left(K_h(X_i)\left(Y_i-V_i^{\top}\left(\theta+M_n\gamma\right)-\tilde{Z}_i^{\top}\gamma\right)^2\right),$$
where the $\textrm{argmin}$ are over the set of all $(\theta,\gamma)$ for which $\gamma_{J_n^c}=0$. Then, we obtain
$$\check{\gamma}_0(J_n,h)=\gamma_0(J_n,h)\quad\textrm{and}\quad\check{\theta}_0(J_n,h)=\theta_0(J_n,h)+M_n\gamma_0(J_n,h).$$
In particular, we read from these formulas
$$\check{\theta}_n^{(2)}(J_n,h)=\hat{\theta}_n^{(2)}(J_n,h)+\left(M_n\right)_{2\cdot}\hat{\gamma}_n(J_n,h)=\hat{\theta}_n^{(2)}(J_n,h).$$
We formulate now the general condition on a model selection procedure. In principle, we can distinguish two different types of such selection procedures: ones that (like the Lasso) generate parameter estimates $\tilde{\theta}_n$ and $\tilde{\gamma}_n$ and then put $\hat{J}_n=\{j:\tilde{\gamma}^{(j)}\neq0\}$; and ones that generate a set $\hat{J}_n$ directly. Depending on the type of the procedure at least one of the following quantities is well defined:
\begin{align}
B_n=&\frac{1}{n}\left\|\mathbf{K}_h^{\frac{1}{2}}\left(\mathbf{Y}-\mathbf{V}\check{\theta}_0(J_n,h)-\tilde{\mathbf{Z}}\tilde{\gamma}_n\right)\right\|_2^2 \nonumber \\
&\quad\quad\quad-\frac{1}{n}\left\|\mathbf{K}_h^{\frac{1}{2}}\left(\mathbf{Y}-\mathbf{V}\check{\theta}_0(J_n,h)-\tilde{\mathbf{Z}}\check{\gamma}_0(J_n,h)\right)\right\|_2^2, \label{eq:defB} \\
C_n(I_n)=&\frac{1}{n}\left\|\mathbf{K}_h^{\frac{1}{2}}\left(\mathbf{Y}-\mathbf{V}\check{\theta}_0(J_n,h)-\tilde{\mathbf{Z}}\left(\check{\gamma}_0(J_n,h)\right)_{\hat{J}_n\cap I_n}\right)\right\|_2^2 \nonumber \\
&\quad\quad\quad-\frac{1}{n}\left\|\mathbf{K}_h^{\frac{1}{2}}\left(\mathbf{Y}-\mathbf{V}\check{\theta}_0(J_n,h)-\tilde{\mathbf{Z}}\check{\gamma}_0(J_n,h)\right)\right\|_2^2, \label{eq:defC}
\end{align}
where $I_n$ is an arbitrary sequence of subsets of indices. If we use a procedure which does not output $\tilde{\gamma}_n$, we just put $B_n=\infty$. Both of the quantities $B_n$ and $C_n(I_n)$ can function as performance measures for model selection. Let $\delta_n$ be such that (recall the definition of $r_i(J_n,h)$ from \eqref{eq:optimal_regression}):
\begin{equation}
\label{eq:cc}
\sup_{k=1,...,p_n}\frac{1}{nh}\sum_{i=1}^nK\left(\frac{X_i}{h}\right)\tilde{Z}_i^{(k)}r_i(J_n,h)=O_P(\delta_n(h)).
\end{equation}
The assumption (CMS) below formulates precisely what is required of the model selection $\hat{J}_n$ in order to obtain asymptotic normality of the final RD estimator.

\smallskip

\noindent\textbf{Assumption (CMS):} (Conditions on Model Selection) 
\emph{$\tilde{\mathbf{Z}}(\hat{J}_n)\mathbf{K}_h\tilde{\mathbf{Z}}(\hat{J}_n)$ is almost surely invertible and RSE$(|\hat{J}_n|,J_n,h)$ holds for $\tilde{Z}_i$. Denote $\zeta_n=\sqrt{\max(0,\min(B_n,C_n(I_n)))}$ (with $Z_i$ replaced by $\tilde{Z}_i$ in the definition) for some subset $I_n\subseteq\{1,...,p_n\}$ and suppose that
\begin{align*}
\beta_{1,n}&=\frac{|\hat{J}_n|\log p_n}{nh}=o_P(1),\quad\beta_{2,n}=|\hat{J}_n|h^4=o_P(1), \\
\beta_{3,n}&=\left(\sqrt{\log p_n}+\sqrt{nh^5}\right)\left(|J_n|+|\hat{J}_n|\right)^{\frac{1}{2}}\zeta_n=o_P(1), \\
\alpha_n&=\left(\sqrt{\log p_n}+\sqrt{nh^5}\right)\left(|J_n|+|\hat{J}_n|\right)\left(\frac{1}{\sqrt{nh}}+\delta_n\right)=o_P(1), \\
&\sqrt{|\hat{J}_n|}\left(\sqrt{\beta_{1,n}}+\sqrt{\beta_{2,n}}\right)\left(\alpha_n+\beta_{3,n}\right)=o_P(1)
\end{align*}}

Suppose for the moment that (CMS) holds. Our notation from above yields (where $\mathcal{S}_n^2$ is defined in Theorem \ref{thm:GeneralAsymptoticNormality}):
\begin{align}
\sqrt{\frac{nh}{\mathcal{S}_n^2}}\left(\hat{\theta}_n^{(2)}-\check{\theta}_0^{(2)}(J_n,h)\right)=&\sqrt{\frac{nh}{\mathcal{S}_n^2}}\left(\check{\theta}_n^{(2)}-\check{\theta}_0^{(2)}(J_n,h)\right). \label{eq:check+err}
\end{align}
Hence, the asymptotics of the left hand side are determined by the asymptotics of the right hand side. Note that
$$\check{r}_i(J_n,h)=Y_i-V_i^{\top}\check{\theta}_n-\tilde{Z}_i^{\top}\check{\gamma}_n=Y_i-V_i^{\top}\hat{\theta}_n-Z_i^{\top}\hat{\gamma}_n=r_i(J_n,h).$$
Hence, all assumptions which we make on $r_i(J_n,h)$ carry over to $\check{r}_i(J_n,h)$. Moreover, $\mu_Z''=\mu_{\tilde{Z}}''$. Thus, all assumptions of Theorem \ref{thm:GeneralAsymptoticNormality} hold true and we may thus apply this theorem with $Z_i=\tilde{Z}_i$ to show that
$$\sqrt{\frac{nh}{\mathcal{S}_n^2}}\left(\check{\theta}_n^{(2)}-\check{\theta}_0(J_n,h)^{(2)}\right)\to\mathcal{N}(0,1).$$
We also see that all assumptions of Lemma \ref{lem:bias} hold for $Z_i=\tilde{Z}_i$ and can hence invoke this lemma to show that $\check{\theta}_0(J_n,h)^{(2)}=\tau+h^2\mathcal{B}_n$. This completes the asymptotic normality. We show now that \eqref{eq:bias_formula}, \eqref{eq:var_formula} and (CMS) hold.

\smallskip
\noindent\textbf{Discussion of Bias:}
Define for ease of notation
$$\check{\beta}_n=\IE\left(K_h(X_i)\tilde{Z}_i(J_n)\tilde{Z}_i(J_n)^\top\right)^{-1}\IE\left(K_h(X_i)\tilde{Z}_i(J_n)Y_i\right),$$
and use the formula for $\mathcal{B}_n$ from Lemma \ref{lem:bias} below. In order to prove \eqref{eq:bias_formula}, we have to show that
\begin{align*}
&\frac{C_\mathcal{B}}{2}\left(\mu_{Y+}''-\mu_{Y-}''-\sum_{k\in J_n}\left(\mu_{Z^{(k)}+}''-\mu_{Z^{(k)}-}''\right)\gamma_n^{(k)}\right)+o\left(|J_n|^{1/2}\right) \\
&\quad=\frac{C_{\mathcal{B}}}{2}\left(\mu_{Y+}''-\mu_{Y-}''-\sum_{k\in J_n}\left(\mu_{Z^{(k)}+}''-\mu_{Z^{(k)}-}''\right)\check{\beta}_n^{(k)}\right)+o(1)+O(|J_n|h^2)+O(|J_n|^{\frac{1}{2}}h) \\
\Leftrightarrow&\sum_{k\in J_n}\left(\mu_{Z^{(k)}+}''-\mu_{Z^{(k)}-}''\right)\left(\check{\beta}_n^{(k)}-\gamma_n^{(k)}\right)=o(1)+O(|J_n|h^2)+O(|J_n|^{\frac{1}{2}}h)+o\left(|J_n|^{1/2}\right).
\end{align*}
Since we assume $|J_n|h^2\to0$ in (BW), the dominating term on the right hand side is $o(|J_n|^{1/2})$. By the Cauchy-Schwarz Inequality we have
$$\left|\sum_{k\in J_n}\left(\mu_{Z^{(k)}+}''-\mu_{Z^{(k)}-}''\right)\left(\check{\beta}_n^{(k)}-\gamma_n^{(k)}\right)\right|\leq\left(\sum_{k\in J_n}\left(\mu_{Z^{(k)}+}''-\mu_{Z^{(k)}-}''\right)^2\right)^{\frac{1}{2}}\left\|\check{\beta}_n-\gamma_n\right\|_2.$$
From the boundedness assumptions in (D) we conclude that
$$\left(\sum_{k\in J_n}\left(\mu_{Z^{(k)}+}''-\mu_{Z^{(k)}-}''\right)^2\right)^{\frac{1}{2}}=O\left(|J_n|^{1/2}\right)$$
and hence we have left to prove that $\left\|\check{\beta}_n-\gamma_n\right\|_2\to0$. We have
\begin{align*}
\left\|\check{\beta}_n-\gamma_n\right\|_2=&\left\|\IE\left(K_h(X_i)\tilde{Z}_i(J_n)\tilde{Z}_i(J_n)^\top\right)^{-1}\IE\left(K_h(X_i)\tilde{Z}_i(J_n)\left(Y_i-\tilde{Z}_i(J_n)^\top\gamma_n\right)\right)\right\|_2 \\
=&O(1)\left\|\IE\left(K_h(X_i)\tilde{Z}_i(J_n)\left(Y_i-\tilde{Z}_i(J_n)^\top\gamma_n\right)\right)\right\|_2
\end{align*}
by assumption. Define
$$\mu(x)=\IE\left((\tilde{Z}_i(J_n)\left(Y_i-\tilde{Z}_i(J_n)^\top\gamma_n\right)\big|X_i=x\right).$$
It holds that
$$\mu_+=\sigma_{Z(J_n)Y+}^2-\sigma_{Z(J_n)+}^2\gamma_n\quad\textrm{and}\quad\mu_+=\sigma_{Z(J_n)Y+}^2-\sigma_{Z(J_n)-}^2,$$
which in turn implies $\mu_++\mu_-=0$. Using an argument as in \eqref{eq:kernel} and the bounded derivatives we obtain that each entry of $\mu(x)$ can be bounded in absolute value by $Ch$, where $C$ is a suitable constant. This yields
$$\left\|\IE\left(K_h(X_i)\tilde{Z}_i(J_n)\left(Y_i-\tilde{Z}_i(J_n)^\top\gamma_n\right)\right)\right\|_2^2\leq C^2|J_n|h^2\to0$$
by the assumptions on $h$ and $|J_n|$. This completes the proof of \eqref{eq:bias_formula}.

\smallskip

\noindent\textbf{Discussion of Variance:}
Recall that $\sigma_l^2,\sigma_r^2$ denote the left and right limits of $\IE(r_i(J_n,h)^2|X_i=x)$ at $x=0$, respectively, in (TCS). It can be computed that for $h\to0$
$$\mathcal{S}_n^2\to \frac{C_{\mathcal{S}}}{f_X(0)}\left(\sigma_l^2+\sigma_r^2\right).$$
Thus, in order to show that \eqref{eq:var_formula} holds, we have to prove that $\sigma_{\tilde Y+}^2\to\sigma_r^2$ and $\sigma_{\tilde Y-}^2\to\sigma_l^2$. In the beginning of the proof of Lemma \ref{lem:bias}, cf. \eqref{eq:begin_lem12}, we provide a formula for $\check{\theta}_0(J_n,h)$ (note that each $Z_i$ in Lemma \ref{lem:bias} has to be replaced by $\tilde{Z}_i$). By combining this with \eqref{eq:kernel_exp}, Lemma \ref{lem:help1} and the boundedness assumptions on the derivatives we obtain (recall that $\mu_{\tilde{Z}(J_n)}$ is continuous with $\mu_{\tilde{Z}(J_n)}(0)=0$)
$$\check{\theta}_0(J_n,h)\to\begin{pmatrix}
\mu_{Y-} & \tau & 0  & 0
\end{pmatrix}^\top.$$
We can obtain a formula for $\gamma_0(J_n,h)=\check{\gamma}_0(J_n,h)$ in the fashion of \eqref{eq:begin_lem12} simply by interchanging the roles of $V_i$ and $\tilde{Z}_i$. Doing this, we obtain $\gamma_0(J_n,h)=\gamma_n+O(h)$ uniformly. Using this and the relation between $\theta_0(J_n,h)$ and $\check{\theta}_0(J_n,h)$ we obtain (recall that $|J_n|h\to0$):
\begin{align*}
&\theta_0(J_n,h)-\begin{pmatrix}
\mu_{\tilde Y-} & \mu_{\tilde Y+}-\mu_{\tilde Y-} & 0 & 0
\end{pmatrix}^\top \\
=&\check{\theta}_0(J_n,h)-M_n\gamma_0(J_n,h)-\begin{pmatrix}
\mu_{Y-}-\mu_{Z-}^{\top}\gamma_n & \tau & 0 & 0
\end{pmatrix}^\top\to0.
\end{align*}
Using this convergence we obtain
\begin{align*}
&\IE\left(r(J_n,h)^2|X_i=0+\right)-\textrm{Var}(\tilde Y_i|X_i=0+) \\
=&\theta_0(J_n,h)^{\top}\IE(V_iV_i^{\top}|X_i=0+)\theta_0(J_n,h)-\left(\mu_{Y+}-\mu_{Z(J_n)+}^{\top}\gamma_n\right)^2 \\
&+\gamma_0(J_n,h)^{\top}\mu_{Z(J_n)Z(J_n)+}\gamma_0(J_n,h)-\gamma_n^{\top}\mu_{Z(J_n)Z(J_n)+}\gamma_n \\
&-2\IE(Y_iV_i^{\top}|X_i=0+)\theta_0(J_n,h)+2\mu_{Y+}^2-2\mu_{Z(J_n)+}^\top\mu_{Y+}\gamma_n \\
&-2\IE(Y_iZ_i(J_n)^{\top}|X_i=0+)\gamma_0(J_n,h)+2\mu_{Y(J_n)Z+}^\top\gamma_n \\
&+2\theta_0(J_n,h)^{\top}\IE(V_iZ_i(J_n)^{\top}|X_i=0+)\gamma_0(J_n,h)+2\gamma_n^{\top}\mu_{Z(J_n)+}\mu_{Z(J_n)+}^{\top}\gamma_n-2\mu_{Y+}\mu_{Z(J_n)+}^{\top}\gamma_n \\
\to&0.
\end{align*}
Thus, by definition of $\sigma_r^2$ in (TCS) we conclude that $\textrm{Var}(\tilde Y_i|X_i=0+)\to\sigma_r^2$. We can prove in a similar fashion that $\textrm{Var}(\tilde Y_i|X_i=0-)\to\sigma_l^2$ which concludes the proof of \eqref{eq:var_formula}.

\smallskip 
\noindent\textbf{Discussion of (CMS):} 
We show finally that the Lasso model selection fulfils the Assumption (CMS). We assume that $\tilde{\mathbf{Z}}(\hat{J}_n)\mathbf{K}_h\tilde{\mathbf{Z}}(\hat{J}_n)$ is almost surely invertible. By Theorem \ref{thm:my_sparsity} we have that $|\hat{J}_n|=O_P(|J_n|)$ which means in particular $\IP(|\hat{J}_n|\leq\log n\cdot |J_n|)\to1$. From this and $\textrm{RSE}(|J_n|\log n,J_n,h)$ for $\tilde{Z}_i$, we conclude that $\textrm{RSE}(|\hat{J}_n|,J_n,h)$ holds as well. From $|\hat{J}_n|=O(|J_n|)$ and $\delta_n=O(\sqrt{\log p_n/nh})$ which we get from Lemma \ref{lem:ZR} for $h$ we conclude that $\alpha_n,\beta_{1,n},\beta_{2,n}\to0$ by the assumptions on the rates. By employing Theorem \ref{thm:LassoCn} we find that also $\beta_{3,n}\to0$ and the last condition of (CMS) holds too.
\end{proof}

In the following we discuss the use of the Lasso as model selector in Theorem \ref{thm:main}, that is, we need to show that (CMS) is true. To this end, we use many standard arguments for the Lasso and the post Lasso. See, for example, \citet{vdGB11,BC13,BRT09}. Note that in the following our interest will lie solely in $\gamma$. Therefore, similarly to the discussion in the proof of Theorem \ref{thm:main}, we note that we may shift $Z_i$ by $\alpha'V_i$ for any $\alpha$ without changing the value of $\tilde{\gamma}_n$ (this is because $\theta$ is not penalized). In contrast to Theorem~\ref{thm:main} we shift for the theoretical analysis in a way such that $\IE(K_b(X_i)Z_i)=0$. This makes some notation simpler. However, for computational stability it might be useful to centralize the covariates by their empirical mean. This practice is hence not problematic. In order to understand well what the Lasso estimator is doing we recall in particular the notation in \eqref{eq:optimal_regression}. Note moreover that the bandwidth $b$ is different from $h$. Recall the definition of $\delta_n(h)$ in \eqref{eq:cc}. Under approximate sparsity, we can prove an exact rate.
\begin{lemma}
\label{lem:ZR}
Let (CTB, \ref{eq:B4}) and (AS) hold. Let $p_n\to\infty$, $b\to0$ and $\log p_n/nb\to0$. Then, for $C>0$ large enough
\begin{align*}
&\IP\left(\sup_{k=1,...,p_n}\left|\frac{1}{nb}\sum_{i=1}^nZ_i^{(k)}K\left(\frac{X_i}{b}\right)r_i(J_n,b)\right|>C\sqrt{\frac{\log p_n}{nb}}\right)\to0.
\end{align*}
Particularly, if in addition, (TCS \eqref{eq:equicont1}, \eqref{eq:equicont2}) and (D conditions on $\mu_Z$ and $\mu_Z'$) hold for $h$ we have
\begin{equation}
\label{eq:deltan}
\delta_n(b)=O\left(\sqrt{\frac{\log p_n}{nb}}\right).
\end{equation}
If, moreover, all of (CTB) and $\IE(K_b(X_i)Z_i)=0$ hold (but possibly not (TCS)), then
\begin{align*}
&\IP\left(\sup_{k=1,...,p_n}\left|\frac{1}{nb}\sum_{i=1}^n\hat{\omega}_{n,k}^{-1}Z_i^{(k)}K\left(\frac{X_i}{b}\right)r_i(J_n,b)\right|>C\sqrt{\frac{\log p_n}{nb}}\right)\to0.
\end{align*}
\end{lemma}
\begin{proof}
We only show the proof for the case with weights. Otherwise, put $\hat{\omega}_{n,k}=1$ below and carry out the steps analogously. Since the conditions of Lemma \ref{lem:weights} hold we have that $\IP(\mathcal{A}_n^c)\to0$ for
$$\mathcal{A}_n=\left\{\forall k\in\{1,...,p_n\}: \hat{\omega}_{n,k}\geq w^{(l)} \right\}.$$
Thus, we get for any $C>0$
\begin{align}
&\IP\left(\sup_{k=1,...,p_n}\left|\frac{1}{nb}\sum_{i=1}^n\hat{\omega}_{n,k}^{-1}Z_i^{(k)}K\left(\frac{X_i}{b}\right)r_i(J_n,b)\right|>C\sqrt{\frac{\log p_n}{nb}}\right) \nonumber \\
=&\IP\left(\exists k\in\{1,...,p_n\}:\left|\frac{1}{nb}\sum_{i=1}^nZ_i^{(k)}K\left(\frac{X_i}{b}\right)r_i(J_n,b)\right|>C\sqrt{\frac{\log p_n}{nb}}\hat{\omega}_{n,k}\right) \nonumber \\
\leq&\IP\left(\exists k\in\{1,...,p_n\}:\left|\frac{1}{nb}\sum_{i=1}^nZ_i^{(k)}K\left(\frac{X_i}{b}\right)r_i(J_n,b)\right|>C\sqrt{\frac{\log p_n}{nb}}w^{(l)}\right)+\IP(\mathcal{A}_n^c) \nonumber \\
\leq&p_n\max_{k=1,...,p_n}\IP\left(\left|\frac{1}{nb}\sum_{i=1}^nZ_i^{(k)}K\left(\frac{X_i}{b}\right)r_i(J_n,b)\right|>C\sqrt{\frac{\log p_n}{nb}}w^{(l)}\right)+\IP(\mathcal{A}_n^c) \nonumber \\
\leq&p_n\max_{k=1,...,p_n}\IP\Bigg(\left|\frac{1}{nb}\sum_{i=1}^nZ_i^{(k)}K\left(\frac{X_i}{b}\right)r_i(J_n,b)-\IE\left(Z_i^{(k)}\frac{1}{b}K\left(\frac{X_i}{b}\right)r_i(J_n,b)\right)\right| \nonumber \\
&\quad\quad\quad\quad\quad>C\sqrt{\frac{\log p_n}{nb}}w^{(l)}\Bigg) \label{eq:aps} \\
&+p_n\max_{k=1,...,p_n}\IP\left(\left|\IE\left(Z_i^{(k)}\frac{1}{b}K\left(\frac{X_i}{b}\right)r_i(J_n,b)\right)\right|>C\sqrt{\frac{\log p_n}{nb}}w^{(l)}\right)+\IP(\mathcal{A}_n^c). \nonumber
\end{align}
We have just argued that $\IP(\mathcal{A}_n^c)\to0$ and we assume that (cf. (AS)) the expectations are smaller than $C\sqrt{\frac{\log p_n}{nb}}$ for a large enough choice of $C>0$. Thus, we get that the second line converges to zero. The first line converges to zero by Lemma \ref{lem:expbound} which we may apply because we assume the moment conditions in (CTB, \eqref{eq:B4}).

Note that when proving the statement without weights, we do not need to rely on Lemma \ref{lem:weights} and hence we do not need to require $\IE(K_b(X_i)Z_i)=0$.

In order to see what $\delta_n(b)$ is, we note that
$$\frac{1}{nh}\sum_{i=1}^nK\left(\frac{X_i}{h}\right)\tilde{Z}_i^{(k)}r_i(J_n,h)=\frac{1}{nh}\sum_{i=1}^nK\left(\frac{X_i}{h}\right)Z_i^{(k)}r_i(J_n,h)-\left[M_n\right]_{\cdot k}^\top\frac{1}{nh}\sum_{i=1}^nV_ir_i(J_n,h).$$
We have just shown that the first part is of order $\sqrt{\log p_n/nb}$ after taking the $\sup$. The second part will be shown in Lemma \ref{lem:VR} to be of order $1/\sqrt{nb}$ because $M_n$ remains bounded by assumption.
\end{proof}
We finish this section with the proofs of Lemmas \ref{lem:bias_formula} and \ref{lem:optimal_variance} which we stated in the main text.
\begin{proof}[Proof of Lemma \ref{lem:bias_formula}]
Let $a_n$ be such that the following is true ($\check{\beta}_n$ was defined in the proof of Theorem \ref{thm:main} and the big-O and small-o terms are the same as in the definition of $\mathcal{B}_n$ in Lemma \ref{lem:bias}):
\begin{align*}
&\frac{C_{\mathcal{B}}}{2}\left(\mu_{\tilde{Y}+}''-\mu_{\tilde{Y}-}''\right)(1+a_n)=\frac{C_{\mathcal{B}}}{2}\Bigg[\mu_{Y+}''-\mu_{Y-}''-\sum_{k\in J_n}\left(\mu_{Z^{(k)}+}''-\mu_{Z^{(k)}-}''\right)\check{\beta}_n^{(k)}\Bigg] \\
&\quad\quad\quad\quad\quad\quad\quad\quad\quad\quad\quad\quad\quad\quad+o(1)+O(|J_n|h^2)+O(|J_n|^{\frac{1}{2}}h) \\
\Leftrightarrow& \left(\mu_{\tilde{Y}+}''-\mu_{\tilde{Y}-}''\right)a_n=\Bigg[-\sum_{k\in J_n}\left(\mu_{Z^{(k)}+}''-\mu_{Z^{(k)}-}''\right)\left(\check{\beta}_n^{(k)}-\gamma_n^{(k)}\right)\Bigg] \\
&\quad\quad\quad\quad\quad\quad\quad\quad\quad\quad+o(1)+O(|J_n|h^2)+O(|J_n|^{\frac{1}{2}}h) \\
\end{align*}
In order to prove \eqref{eq:bias_good_formula}, we need to show that $a_n\to0$. Note next, that we assume $|J_n|h^2\to0$ in (BW) and hence all big-O-terms above are $o(1)$. Since \eqref{eq:bias1} holds by assumption, we have left to show that
\begin{align*}
&\left|\frac{\sum_{k\in J_n}\left(\mu_{Z^{(k)}+}''-\mu_{Z^{(k)}-}''\right)\left(\check{\beta}_n^{(k)}-\gamma_n^{(k)}\right)}{\mu_{\tilde{Y}+}''-\mu_{\tilde{Y}-}''}\right| \\
\leq&\left|\frac{\sum_{k\in J_n}\left(\mu_{Z^{(k)}+}''-\mu_{Z^{(k)}-}''\right)^2}{\left(\mu_{\tilde{Y}+}''-\mu_{\tilde{Y}-}''\right)^2}\right|^{\frac{1}{2}}\left\|\check{\beta}_n-\gamma_n\right\|_2\to0
\end{align*}
(use the Cauchy-Schwarz Inequality). The first part is $O(1)$ by \eqref{eq:bias1} and \eqref{eq:bias2} and $\|\check{\beta}_n-\gamma_n\|_2\to0$ converges to zero which was proven in the proof of Theorem \ref{thm:main}.
\end{proof}

\begin{proof}[Proof of Lemma \ref{lem:optimal_variance}]
Suppose firstly that $J_{1,n}\supseteq J_{2,n}$. By definition (cf. \ref{thm:GeneralAsymptoticNormality})
$$\mathcal{S}_n^2(J_{1,n})=\frac{1}{h}\IE\left(K\left(\frac{X_i}{h}\right)^2\xi\left(\frac{X_i}{h}\right)r_i(J_{1,n},h)^2\right),$$
where $\xi$ is a given function. Hence, we need to study the limit of the following difference (below $\gamma_0(J_{2,n},h)$ as a vector in $\IR^{|J_{1,n}|}$ which has zeros at the places $J_{1,n}\setminus J_{2,n}$)
\begin{align*}
&\frac{1}{h}\IE\left(K\left(\frac{X_i}{h}\right)^2\xi\left(\frac{X_i}{h}\right)r_i(J_{1,n},h)^2\right)-\frac{1}{h}\IE\left(K\left(\frac{X_i}{h}\right)^2\xi\left(\frac{X_i}{h}\right)r_i(J_{2,n},h)^2\right) \\
=&\frac{1}{h}\IE\left(K\left(\frac{X_i}{h}\right)^2\xi\left(\frac{X_i}{h}\right)\left(\begin{pmatrix}\theta_0(J_{2,n},h)-\theta_0(J_{1,n},h) \\  \gamma_0(J_{2,n},h)-\gamma_0(J_{1,n},h) \end{pmatrix}^\top \begin{pmatrix} V_i \\ Z_i(J_{1,n})\end{pmatrix}\right)^2\right) \\
&+\frac{2}{h}\IE\left(K\left(\frac{X_i}{h}\right)^2\xi\left(\frac{X_i}{h}\right)r_i(J_{2,n},h)\begin{pmatrix}\theta_0(J_{2,n},h)-\theta_0(J_{1,n},h) \\  \gamma_0(J_{2,n},h)-\gamma_0(J_{1,n},h) \end{pmatrix}^\top \begin{pmatrix} V_i \\ Z_i(J_{1,n})\end{pmatrix}\right).
\end{align*}
We suppose in (TCS) that $\IE(r_i(J_{a,n},h)^2|X_i=x)$ behaves nicely around $x=0$ for $a=1,2$ and hence the qualitative behavior of the above is determined by $\left\|\theta_0(J_{1,n},h)-\theta_0(J_{2,n},h)\right\|_2^2$ and $\left\|\gamma_0(J_{1,n},h)-\gamma_0(J_{2,n},h)\right\|_2^2$. These can be controlled by applying least squares algebra as follows:
\begin{align*}
&\begin{pmatrix}
\theta_0(J_{1,n},h) \\ \gamma_0(J_{1,n},h)
\end{pmatrix}-\begin{pmatrix}
\theta_0(J_{2,n},h) \\ \gamma_0(J_{2,n},h)
\end{pmatrix} \\
=&\IE\left(K_h(X_i)\begin{pmatrix}
V_i \\ Z_i(J_{1,n})
\end{pmatrix}\begin{pmatrix}
V_i \\ Z_i(J_{1,n})
\end{pmatrix}^\top \right)^{-1} \\
&\quad\quad\quad\quad\times\IE\left(K_h(X_i)\begin{pmatrix}
V_i \\ Z_i(J_{1,n})
\end{pmatrix}\left(Y_i-\begin{pmatrix}
V_i \\ Z_i(J_{1,n})
\end{pmatrix}^\top\begin{pmatrix}
\theta_0(J_{2,n},h) \\ \gamma_0(J_{2,n},h)
\end{pmatrix}\right)\right) \\
=&\IE\left(K_h(X_i)\begin{pmatrix}
V_i \\ Z_i(J_{1,n})
\end{pmatrix}\begin{pmatrix}
V_i \\ Z_i(J_{1,n})
\end{pmatrix}^\top \right)^{-1}\IE\left(K_h(X_i)\begin{pmatrix}
V_i \\ Z_i(J_{1,n})
\end{pmatrix}r_i(J_{n,2},h)\right).
\end{align*}
Since
$$\left\|\IE\left(K_h(X_i)\begin{pmatrix}
V_i \\ Z_i(J_{1,n})
\end{pmatrix}\begin{pmatrix}
V_i \\ Z_i(J_{1,n})
\end{pmatrix}^\top \right)^{-1}\right\|_2=O(1)$$
and by Assumption (AS) and noting that $\IE(K_h(X_i)V_ir_i(J_{2,n},h))=0$ we obtain
\begin{align*}
&\left\|\begin{pmatrix} \theta_0(J_{1,n},h) \\ \gamma_0(J_{1,n},h)
\end{pmatrix}-\begin{pmatrix}
\theta_0(J_{2,n},h) \\ \gamma_0(J_{2,n},h)
\end{pmatrix}\right\|_2^2=O(1)\left\|\IE\left(K_h(X_i)\begin{pmatrix}
V_i \\ Z_i(J_{1,n})
\end{pmatrix}r_i(J_{n,2},h)\right)\right\|_2^2 \\
=&O\left(\frac{\log p_n}{nh}|J_{1,n}|\right)
\end{align*}
which converges to zero. If $J_{2,n}\supseteq J_{1,n}$ we can apply the same arguments with the roles of $J_{1,n}$ and $J_{2,n}$ interchanged. If $J_{1,n}$ and $J_{2,n}$ are not nested, define $J_{3,n}=J_{1,n}\cup J_{2,n}$. (AS) and (TCS) continue to hold for $J_{3,n}$ and by the above argument, $\mathcal{S}_n^2(J_{3,n})-\mathcal{S}_n^2(J_{1,n})\to0$ and $\mathcal{S}_n^2(J_{3,n})-\mathcal{S}_n^2(J_{2,n})\to0$ which implies $\mathcal{S}_n^2(J_{1,n})-\mathcal{S}_n^2(J_{2,n})\to0$ and the proof is complete.
\end{proof}

\subsection{Preliminary Results}
Let $L:\IR\to\IR$ be an arbitrary function and recall the definitions
$$L_-^{(\alpha)}=\int_{-\infty}^0L(u)u^{\alpha}du,\quad L_+^{(\alpha)}=\int_0^{\infty}L(u)u^{\alpha}du,\quad L^{(\alpha)}=\int_{-\infty}^{\infty}L(u)u^{\alpha}du$$
for $\alpha\in\{0,1,2,3,4\}$. If $L$ is a symmetric second order kernel, we have $L^{(0)}=1$, $L_-^{(0)}=L_+^{(0)}=\frac{1}{2}$, $L_-^{(1)}=-L_+^{(1)}$ and $L_-^{(2)}=L_+^{(2)}$. This proves the following lemma.

\begin{lemma}
\label{lem:kappa}
For every symmetric kernel $K$ with $K^{(2)}<\infty$ the matrix
$$\kappa(K)=\begin{pmatrix}
K^{(0)}   & K_+^{(0)} & K^{(1)}   & K_+^{(1)} \\
K_+^{(0)} & K_+^{(0)} & K_+^{(1)} & K_+^{(1)} \\
K^{(1)}   & K_+^{(1)} & K^{(2)}   & K_+^{(2)} \\
K_+^{(1)} & K_+^{(1)} & K_+^{(2)} & K_+^{(2)}
\end{pmatrix}$$
is invertible.
\end{lemma}
\begin{proof}
By using the above relations, we find that the determinant of $\kappa(K)$ is given by
$$\left(\left(K_+^{(1)}\right)^2-\frac{1}{2}K^{(2)}\right)^2.$$
By using Jensen's Inequality for integrals we get that $\left(K_+^{(1)}\right)^2<\frac{1}{2}K^{(2)}$ and we conclude that the determinant is strictly positive which completes the proof.
\end{proof}
The matrix $\kappa(K)$ will play a role in Lemma \ref{lem:SigmaV} below. We will also often need computations of the following type and therefore we do it here once as a reference: Suppose that $f:\IR\to\IR$ is twice one-sided differentiable at zero, that is, the first two derivatives of $f_m:(-\infty,0)\to\IR, x\mapsto f(x)$ and $f_p:(0,\infty)\to\IR, x\mapsto f(x)$ exist and can be continuously extended to zero. We write $f_-$ for the extension of $f_m$ to zero and $f_+$ for the extension of $f_p$ at zero. Define $f_-',f_+',f_-''$, and $f_+''$ in a similar way. Then,
\begin{align}
&\IE\left(\frac{1}{h}L\left(\frac{X_i}{h}\right)f(X_i)\right)=\int_{-\infty}^0L(u)f(uh)f_X(uh)du+\int_0^{-\infty}L(u)f(uh)f_X(uh)du \nonumber \\
=&L_-^{(0)}f_-f_X(0)+L_+^{(0)}f_+f_X(0)+h\left[L_-^{(1)}\left(f\cdot f_X\right)_-'+L_+^{(1)}\left(f\cdot f_X\right)_+'\right] \nonumber \\
&\quad+\frac{1}{2}h^2\left[L_-^{(2)}\left(f\cdot f_X\right)_-''+L_+^{(2)}\left(f\cdot f_X\right)_+''\right]+o(h^2). \label{eq:kernel}
\end{align}
A simple consequence of the above is the following Lemma:
\begin{lemma}
\label{lem:SigmaV}
Let $f_X$ be twice continuously differentiable and let $K^{(\alpha)}$ for $\alpha\in\{0,...,4\}$ and $\left(K^2\right)^{(\alpha)}$ for $\alpha\in\{0,1,2\}$ be finite. If $h\to0$ and $nh\to\infty$, then
\begin{align}
&\frac{1}{n}\sum_{i=1}^nK_h(X_i)V_iV_i^{\top}=\IE\left(K_h(X_i)V_iV_i^{\top}\right)+O_P\left(\frac{1}{nh}\right), \label{eq:kernel_conv} \\
&\IE(K_h(X_i)V_iV_i^{\top})= f_X(0)\kappa(K)+f_X'(0)h\begin{pmatrix}
K^{(1)}   & K^{(1)}_+ & K^{(2)}   & K_+^{(2)} \\
K_+^{(1)} & K_+^{(1)} & K_+^{(2)} & K_+^{(2)} \\
K^{(2)}   & K_+^{(2)} & K^{(3)}   & K_+^{(3)} \\
K_+^{(2)} & K_+^{(2)} & K_+^{(3)} & K_+^{(3)}
\end{pmatrix} \nonumber \\
&\quad+\frac{h^2}{2}f_X''(0)\begin{pmatrix}
K^{(2)}   & K_+^{(2)} & K^{(3)}   & K_+^{(3)} \\
K_+^{(2)} & K_+^{(2)} & K_+^{(3)} & K_+^{(3)} \\
K^{(3)}   & K_+^{(3)} & K^{(4)}   & K_+^{(4)} \\
K_+^{(3)} & K_+^{(3)} & K_+^{(4)} & K_+^{(4)}
\end{pmatrix}+o(h^2) \label{eq:kernel_exp}
\end{align}
\end{lemma}
\begin{proof}
For the proof we have to compute the expectation and the variance of the average by means of \eqref{eq:kernel}.
\end{proof}

Lastly, we consider in this preliminary discussion a local version of the Bernstein Inequality. For the convenience of the reader, we state the regular Bernstein Inequality as it can be found e.g. in \citet{GN16}.
\begin{proposition}
\label{prop:bernstein}
Let $A_i$, $i=1,...,n$ be a sequence of independent, centered random variables such that there are numbers $c$ and $\sigma_i$ such that for all $m\in\IN$ $\IE(|A_i|^m|)\leq\frac{m!}{2}\sigma_i^2c^{m-2}$. Set $\sigma^2=\sum_{i=1}^n\sigma_i^2$, $S_n=\sum_{i=1}^nA_i$. Then, for all $t\geq0$, $\IP(S_n\geq t)\leq\exp\left(-\frac{t^2}{2(\sigma^2+ct)}\right)$.
\end{proposition}
In our setting, the following local version will be relevant.
\begin{lemma}
\label{lem:expbound}
Let $B_1,...,B_n$ be iid and denote $\mu_m(x)=\IE\left(\left|B_i\right|^m\big|X_i=x\right)$ (which is independent of $i=1,...n$). Suppose that for all $m\geq2$ and $n\in\IN$
\begin{align*}
&\int_{\IR}K(u)^m\mu_{m}(uh)f_X(uh)du\leq \frac{m!}{2}\sigma_0^2c^{m-2}, \\
&\left|\int_{\IR}K(u)\IE(B_i|X_i=uh)f_X(uh)du\right|\leq c^*
\end{align*}
for some constants $\sigma_0^2,c,c^*>0$. Let furthermore $p_n$ and $h$ be such that (after possibly increasing $\sigma_0$ and $c$) for all $m\geq2$ and all $n\in\IN$
$$\frac{m!}{2}\sigma_0^2c^{m-2}\left(c^*\right)^{-m}\geq h^{m-1}\textrm{ and }\sqrt{\frac{\log p_n}{nh}}\leq\frac{1}{8c}.$$
It holds for all $x\geq16\sigma_0^2$ that
$$\IP\left(\frac{1}{n}\sum_{i=1}^n\left(K_h(X_i)B_i-\IE\left(K_h(X_i)B_i\right)\right)>x\sqrt{\frac{\log p_n}{nh}}\right)\leq\left(\frac{1}{p_n}\right)^x.$$
\end{lemma}
\begin{proof}
We begin by rewriting the term of interest as follows. Let $\epsilon_n=x\sqrt{\frac{\log p_n}{nh}}$. Then,
\begin{align*}
&\IP\left(\frac{1}{n}\sum_{i=1}^n\left(K_h(X_i)B_i-\IE\left(K_h(X_i)B_i\right)\right)>\epsilon_n\right)=\IP\left(\sum_{i=1}^nA_{n,i}>nh\epsilon_n\right),
\end{align*}
where
$$A_{n,i}=K\left(\frac{X_i}{h}\right)B_i-\IE\left(K\left(\frac{X_i}{h}\right)B_i\right).$$
We apply now Bernstein's Inequality (cf. Proposition \ref{prop:bernstein}) to $A_{n,i}$. Since for $a,b\geq0$, we have that $(a+b)^m\leq2^{m-1}(a^m+b^m)$, we have for any $i=1,...,n$ and any $m\in\IN$ by assumption
\begin{align*}
&\IE\left(\left|A_{i,n}\right|^m\right)\leq2^{m-1}\left(\IE\left(K\left(\frac{X_i}{h}\right)^m\left|B_i\right|^m\right)+\left|\IE\left(K\left(\frac{X_i}{h}\right)B_i\right)\right|^m\right) \\
\leq&2^{m-1}\left(h\int_{\IR}K(u)^m\mu_{m}(uh)f_X(uh)du+\left|h\int_{\IR}K(u)\IE(B_i|X_i=uh)f_X(uh)du\right|^m\right) \\
\leq&2^{m-1}\left(\frac{m!}{2}h\sigma_0^2c^{m-2}+h^m\left(c^*\right)^m\right)\leq\frac{m!}{2}\cdot4h\sigma_0^2\cdot (2c)^{m-2}.
\end{align*}
We may thus apply Proposition \ref{prop:bernstein} with $\sigma^2=4nh\sigma_0^2$ and "$c=2c$". We conclude
\begin{align*}
&\IP\left(\frac{1}{n}\sum_{i=1}^n\left(K_h(X_i)B_i-\IE\left(K_h(X_i)B_i\right)\right)>\epsilon_n\right) \\
\leq&\exp\left(-\frac{n^2h^2\epsilon_n^2}{2\left(4nh\sigma_0^2+2cnh\epsilon_n\right)}\right)=\exp\left(-\frac{x^2nh\log p_n}{2\left(4nh\sigma_0^2+2cx\sqrt{nh\log p_n}\right)}\right) \\
\leq&\exp\left(-\frac{x^2\log p_n}{2\left(4\sigma_0^2+2cx\sqrt{\frac{\log p_n}{nh}}\right)}\right)\leq \exp\left(-x\log p_n\right)
\end{align*}
by the assumptions on $x$ and $\log p_n/nh$.
\end{proof}

\subsection{The Result for a General Model Selection Algorithm}
\subsubsection{Statement of the Result and Proof Structure}
\begin{theorem}
\label{thm:GeneralAsymptoticNormality}
Let $p_n\to\infty$, $h\to0$, $nh\to\infty$ and $\log p_n/nh\to0$ and let $K$ be symmetric, compactly supported with $K^{(4)},(K^2)^{(2)}<\infty$. Suppose that (CMS) holds and let $f_X$ and $\mu_{Z^{(k)}}$ be twice differentiable in a neighbourhood around zero with $f_X(0)>0$ and $f_X''$ being continuous at zero and $\mu_{Z^{(k)}}(0)=\mu_{Z^{(k)}}'(0)=0$ as well as
\begin{equation}
\label{eq:sd_bound}
\sup_{n\in\IN}\sup_{k\in\{1,...,p_n\}}\sup_{u\in[0,1]}\left|\mu_{Z^{(k)}}''(uh)\right|+\left|\mu_{Z^{(k)}}''(-uh)\right|<\infty.
\end{equation}
Moreover, suppose that for $\mu_{k,m}(x)=\IE\left(\left|Z_i^{(k)}\right|^m\Big|X_i=x\right)$ there are finite numbers $\sigma_0^2,c,c^*$ such that for all natural $m\geq2$ and all $k$
\begin{align}
\int_{\IR}\left(1+|u|^m\right)K(u)^m\mu_{k,m}(uh)f_X(uh)du\leq&\frac{m!}{2}\sigma_0^2c^{m-2}, \label{eq:mombound1} \\
\int_{\IR}\left(1+|u|\right)K(u)^m\mu_{k,1}(uh)f_X(uh)du\leq& c^*. \label{eq:mombound2}
\end{align}
Suppose that for the target set $J_n$, there are $\delta>0$ and finite numbers $\sigma_l,\sigma_r,C>0$ such that
\begin{align}
&\lim_{n\to\infty}\sup_{u\in[0,1]}\left|\IE(r_i(J_n,h)^2|X_i=uh)-\sigma_r^2\right|=0, \nonumber \\
&\quad\quad\quad\quad\lim_{n\to\infty}\sup_{u\in[0,1]}\left|\IE(r_i(J_n,h)^2|X_i=-uh)-\sigma_l^2\right|=0, \label{eq:equi-cont} \\
&\sup_{n\in\IN}\sup_{x\in[-h,h]}\left|\IE(|r_i(J_n,h)|^{2+\delta}|X_i=x)\right|<C. \label{eq:uboundr}
\end{align}
Set $w=\left(\left(f_X(0)\kappa(K)\right)^{-1}\right)_{2\cdot}^\top$ to be the scaled second row of the inverse of $\kappa(K)$. Define
\begin{align*}
\mathcal{S}_n^2=\frac{1}{h}\IE\left(K\left(\frac{X_i}{h}\right)^2\left(w^\top V_i\right)^2r_i(J_n,h)^2\right).
\end{align*}
Then,
\begin{align*}
&\sqrt{\frac{nh}{\mathcal{S}^2_n}}\left(\hat{\tau}_n(\hat{J}_n)-\theta_{0,n}^{(2)}\right)\overset{d}{\to}\mathcal{N}\left(0,1\right).
\end{align*}
\end{theorem}

\begin{proof}[Proof of Theorem \ref{thm:GeneralAsymptoticNormality}]
Note that the conditions of this theorem contain all conditions of the supporting results (from the following section) or imply them (e.g. \eqref{eq:equi-cont} is stronger than \eqref{eq:rsqbounded}). Therefore we can use all results from Section \ref{subsubsec:sup_res_gc}. Let
\begin{equation}
\label{eq:defM}
\mathbf{M}_n(\hat{J}_n)=\mathbf{I}_n-\mathbf{K}_h^{\frac{1}{2}}\mathbf{Z}(\hat{J}_n)\left(\mathbf{Z}(\hat{J}_n)^{\top}\mathbf{K}_h\mathbf{Z}(\hat{J}_n)\right)^{-1}\mathbf{Z}(\hat{J}_n)^{\top}\mathbf{K}_h^{\frac{1}{2}}
\end{equation}
denote the projection matrix on the chosen covariates ($\mathbf{I}_n$ denotes $n$-dimensional identity matrix). We write here $\mathbf{r}_n=\mathbf{r}(J_n,h)$ and $\gamma_{0,n}=\gamma_{0,n}(J_n,h)$ because $J_n$ and $h$ will be the same sequences throughout the proof. Moreover, $\gamma_{0,n}$ will be understood as element of $\IR^{p_n}$ with $\gamma_{0,n}^{(k)}=0$ for $k\notin J_n$. Note that $\mathbf{M}_n\mathbf{K}_h^{\frac{1}{2}}\mathbf{Z}(\hat{J}_n)=0$. We thus obtain by calculation (or the Frisch-Waugh-Lovell Theorem for weighted regression) the following representation of our estimator
\begin{align*}
\hat{\theta}_n=&\left(\mathbf{V}^{\top}\mathbf{K}_h^{\frac{1}{2}}\mathbf{M}_n(\hat{J}_n)\mathbf{K}_h^{\frac{1}{2}}\mathbf{V}\right)^{-1}\mathbf{V}^{\top}\mathbf{K}_h^{\frac{1}{2}}\mathbf{M}_n(\hat{J}_n)\mathbf{K}_h^{\frac{1}{2}}\mathbf{Y} \\
=&\theta_{0,n}+\left(\mathbf{V}^{\top}\mathbf{K}_h^{\frac{1}{2}}\mathbf{M}_n(\hat{J}_n)\mathbf{K}_h^{\frac{1}{2}}\mathbf{V}\right)^{-1}\mathbf{V}^{\top}\mathbf{K}_h^{\frac{1}{2}}\mathbf{M}_n(\hat{J}_n)\mathbf{K}_h^{\frac{1}{2}}\left(\mathbf{Z}\gamma_{0,n}+\mathbf{r}_n\right).
\end{align*}
Suppose for the moment that we know that
\begin{equation}
\label{eq:last_step}
\frac{1}{n}\mathbf{V}^{\top}\mathbf{K}^{\frac{1}{2}}M_n(\hat{J}_n)\mathbf{K}_h^{\frac{1}{2}}\mathbf{Z}\gamma_{0,n}=o_P\left(\frac{1}{\sqrt{nh}}\right).
\end{equation}
If the above is true we obtain together with Proposition \ref{prop:VV}, Lemma \ref{lem:VR}, the assumptions on $\delta_n$ and the fact that $\kappa(K)$ is invertible by Lemma \ref{lem:kappa} that
\begin{align*}
&\left(\frac{1}{n}\mathbf{V}^{\top}\mathbf{K}_h^{\frac{1}{2}}\mathbf{M}_n(\hat{J}_n)\mathbf{K}_h^{\frac{1}{2}}\mathbf{V}\right)^{-1}=\left(f_X(0)\kappa(K)\right)^{-1}+o_P(1), \\
&\frac{1}{n}\mathbf{V}^{\top}\mathbf{K}_h^{\frac{1}{2}}\mathbf{M}_n(\hat{J}_n)\mathbf{K}_h^{\frac{1}{2}}\left(\mathbf{Z}\gamma_{0,n}+\mathbf{r}_n\right)=\frac{1}{n}\mathbf{V}^{\top}\mathbf{K}_h\mathbf{r}_n+o_P\left(\frac{1}{\sqrt{nh}}\right), \\
&\frac{1}{n}\mathbf{V}^{\top}\mathbf{K}_h\mathbf{r}_n=O_P\left(\frac{1}{\sqrt{nh}}\right).
\end{align*}
Hence,
$$\sqrt{nh}\left(\hat{\theta}_n-\theta_{0,n}\right)=\left(f_X(0)\kappa(K)\right)^{-1}\sqrt{nh}\frac{1}{n}\mathbf{V}^{\top}\mathbf{K}_h\mathbf{r}_n+o_P(1).$$
Thus, to find the asymptotics of the estimator for the treatment effect we study the second entry of the above vector. Recall to this end that $w$ denotes the second row of the kernel matrix written as a column. Thus we have
$$\sqrt{nh}\left(\hat{\tau}_n-\theta_{0,n}^{(2)}\right)=\sqrt{nh}\frac{1}{n}w^{\top}\mathbf{V}^{\top}\mathbf{K}_h\mathbf{r}_n+o_P(1)=\frac{1}{\sqrt{nh}}\sum_{i=1}^nK\left(\frac{X_i}{h}\right)w^{\top}V_ir_i(J_n,h)+o_P(1).$$
For simplicity of notation we write $\nu(X_i/h)=w^{\top}V_i$. Now we can employ Lyapunov's central limit theorem (cf. Lemma 15.41 and Theorem 15.43 in \citet{K08}). We have by definition
\begin{align*}
&\textrm{Var}\left(\frac{1}{\sqrt{nh}}\sum_{i=1}^nK\left(\frac{X_i}{h}\right)\nu\left(\frac{X_i}{h}\right)r_i(J_n,h)\right) \\
=&\frac{1}{nh}\sum_{i=1}^n\IE\left(K\left(\frac{X_i}{h}\right)^2\nu\left(\frac{X_i}{h}\right)^2r_i(J_n,h)^2\right) \\
=&\frac{1}{h}\IE\left(K\left(\frac{X_i}{h}\right)^2\nu\left(\frac{X_i}{h}\right)^2r_i(J_n,h)^2\right)=\mathcal{S}_n^2.
\end{align*}
By the continuity and boundedness assumptions \eqref{eq:equi-cont} and \eqref{eq:uboundr} on $\IE(r_i(J_n,h)^2|X_i)$ and $\IE(|r_{n,ij}|^{2+\delta}|X_i)$, we conclude that there are constants $\alpha_0,\alpha_1>0$ such that for some $\delta>0$ and $n\to\infty$
\begin{align*}
&\frac{1}{h}\IE\left(K\left(\frac{X_i}{h}\right)^{2}\nu\left(\frac{X_i}{h}\right)^2r_i(J_n,h)^{2}\right)\to\alpha_0, \\
&\frac{1}{h}\IE\left(K\left(\frac{X_i}{h}\right)^{2+\delta}\left|\nu\left(\frac{X_i}{h}\right)\right|^{2+\delta}|r_i(J_n,h)|^{2+\delta}\right)\leq\alpha_1.
\end{align*}
Thus the Lyapunov condition is fulfilled: For $n\to\infty$
\begin{align*}
\sum_{i=1}^n\frac{n^{-\frac{2+\delta}{2}}h^{-\frac{2+\delta}{2}}\IE\left(K\left(\frac{X_i}{h}\right)^{2+\delta}\left|\nu\left(\frac{X_i}{h}\right)\right|^{2+\delta}|r_i(J_n,h)|^{2+\delta}\right)}{\mathcal{S}_n^{2+\delta}}\leq\frac{\alpha_1(nh)^{-\frac{\delta}{2}}}{\left(\frac{\alpha_0}{2}\right)^{\frac{2+\delta}{2}}}\to0
\end{align*}
and we conclude
$$\frac{1}{\sqrt{nh\mathcal{S}_n^2}}\sum_{i=1}^nK\left(\frac{X_i}{h}\right)\nu\left(\frac{X_i}{h}\right)r_i(J_n,h)\to\mathcal{N}(0,1).$$
In order to finish the proof, we have to show \eqref{eq:last_step} which we will do next. Recall to this end the post selection estimators 
$$\left(\hat{\theta}_n,\hat{\gamma}_n\right)=\argmin{\theta,\gamma:\gamma_{\hat{J}_n^c}=0}\left\|\mathbf{K}_h^{\frac{1}{2}}\left(\mathbf{Y}-\mathbf{V}\theta-\mathbf{Z}\gamma\right)\right\|_2^2.$$
Note that $\hat{\gamma}_n$ is constrained to be zero for covariates not included in $\hat{J}_n$. Therefore, we have $M_n(\hat{J}_n)\mathbf{K}^{\frac{1}{2}}_h\mathbf{Z}\hat{\gamma}_n=0$. We can hence write
\begin{align}
&\frac{1}{n}\mathbf{V}^{\top}\mathbf{K}_h^{\frac{1}{2}}M_n(\hat{J}_n)\mathbf{K}_h^{\frac{1}{2}}\mathbf{Z}\gamma_{0,n}=\frac{1}{n}\mathbf{V}^{\top}\mathbf{K}_h^{\frac{1}{2}}M_n(\hat{J}_n)\mathbf{K}_h^{\frac{1}{2}}\mathbf{Z}\left(\gamma_{0,n}-\hat{\gamma}_n\right) \nonumber \\
=&\frac{1}{n}\mathbf{V}^{\top}\mathbf{K}_h\mathbf{Z}\left(\gamma_{0,n}-\hat{\gamma}_n\right)+\frac{1}{n}\mathbf{V}^{\top}\mathbf{K}_h\mathbf{Z}(\hat{J}_n)\left(\mathbf{Z}(\hat{J}_n)^{\top}\mathbf{K}_h\mathbf{Z}(\hat{J}_n)\right)^{-1}\mathbf{Z}(\hat{J}_n)^{\top}\mathbf{K}_h\left(\mathbf{Z}\gamma_{0,n}-\mathbf{Z}\hat{\gamma}_n\right) \nonumber \\
=&\frac{1}{n}\mathbf{V}^{\top}\mathbf{K}_h\mathbf{Z}\left(\gamma_{0,n}-\hat{\gamma}_n\right) \label{eq:l1} \\
&+\frac{1}{n}\mathbf{V}^{\top}\mathbf{K}_h\mathbf{Z}(\hat{J}_n)\left(\mathbf{Z}(\hat{J}_n)^{\top}\mathbf{K}_h\mathbf{Z}(\hat{J}_n)\right)^{-1}\mathbf{Z}(\hat{J}_n)^{\top}\mathbf{K}_h\mathbf{V}\left(\hat{\theta}_n-\theta_{0,n}\right) \label{eq:l2} \\
&-\frac{1}{n}\mathbf{V}^{\top}\mathbf{K}_h\mathbf{Z}(\hat{J}_n)\left(\mathbf{Z}(\hat{J}_n)^{\top}\mathbf{K}_h\mathbf{Z}(\hat{J}_n)\right)^{-1}\mathbf{Z}(\hat{J}_n)^{\top}\mathbf{K}_h\mathbf{r}_n \label{eq:l3} \\
&+\frac{1}{n}\mathbf{V}^{\top}\mathbf{K}_h\mathbf{Z}(\hat{J}_n)\left(\mathbf{Z}(\hat{J}_n)^{\top}\mathbf{K}_h\mathbf{Z}(\hat{J}_n)\right)^{-1}\mathbf{Z}(\hat{J}_n)^{\top}\mathbf{K}_h\left(\mathbf{Y}-\mathbf{Z}\hat{\gamma}_n-\mathbf{V}\hat{\theta}_n\right). \label{eq:l4}
\end{align}
Note that \eqref{eq:l4} equals zero because it contains the empirical correlation of a covariate  with the empirical residuals (which is zero). From the definition of $M_n(\hat{J}_n)$ it follows from \eqref{eq:mat2} of Proposition \ref{prop:VV} that $\eqref{eq:l3}=o_P(1/\sqrt{nh})$ by the conditions on $\delta_n$. Hence, in order to prove \eqref{eq:last_step}, we have to prove that \eqref{eq:l1} and \eqref{eq:l2} are both of order $o_P(1/\sqrt{nh})$. We do this by studying the rate of convergence of $\hat{\theta}_n$ and $\hat{\gamma}_n$. We note firstly that
$$\frac{1}{n}\begin{pmatrix}
\theta_{0,n}-\hat{\theta}_n \\ \gamma_{0,n}-\hat{\gamma}_n
\end{pmatrix}^{\top}\begin{pmatrix}
\mathbf{V}^{\top} \\ \mathbf{Z}^{\top}
\end{pmatrix}\mathbf{K}_h\begin{pmatrix}
\mathbf{V} & \mathbf{Z}
\end{pmatrix}\begin{pmatrix}
\theta_{0,n}-\hat{\theta}_n \\ \gamma_{0,n}-\hat{\gamma}_n
\end{pmatrix}\geq\phi(|\hat{J}_n|,J_n)\left\|\begin{pmatrix}
\theta_{0,n}-\hat{\theta}_n \\ \gamma_{0,n}-\hat{\gamma}_n
\end{pmatrix}\right\|_2^2,$$
where $\phi$ is defined as in Definition \ref{def:RSE}. Since we assume the restricted sparse eigenvalue condition $\textrm{RSE}(|\hat{J}_n|,J_n,h)$ we have that $\phi(|\hat{J}_n|,J_n)^{-1}=O_P(1)$ and we conclude that
\begin{equation}
\label{eq:rate1}
\left\|\begin{pmatrix}
\theta_{0,n}-\hat{\theta}_n \\ \gamma_{0,n}-\hat{\gamma}_n
\end{pmatrix}\right\|_2=O_P\left(\frac{1}{\sqrt{n}}\left\|\mathbf{K}_h^{\frac{1}{2}}\begin{pmatrix}
\mathbf{V} & \mathbf{Z}
\end{pmatrix}\begin{pmatrix}
\theta_{0,n}-\hat{\theta}_n \\ \gamma_{0,n}-\hat{\gamma}_n
\end{pmatrix}\right\|_2\right).
\end{equation}
The proof strategy is now similar to \citet{BC13}. At first we note that, by definition for any index set $I_n$
\begin{align}
&\frac{1}{n}\left\|\mathbf{K}^{\frac{1}{2}}\left(\mathbf{Y}-\mathbf{V}\hat{\theta}_n-\mathbf{Z}\hat{\gamma}_n\right)\right\|_2^2-\frac{1}{n}\left\|\mathbf{K}^{\frac{1}{2}}\left(\mathbf{Y}-\mathbf{V}\theta_{0,n}-\mathbf{Z}\gamma_{0,n}\right)\right\|_2^2 \nonumber \\
\leq&\left\{\begin{array}{l}
\frac{1}{n}\left\|\mathbf{K}^{\frac{1}{2}}\left(\mathbf{Y}-\mathbf{V}\theta_{0,n}-\mathbf{Z}\tilde{\gamma}_n\right)\right\|_2^2-\frac{1}{n}\left\|\mathbf{K}^{\frac{1}{2}}\left(\mathbf{Y}-\mathbf{V}\theta_{0,n}-\mathbf{Z}\gamma_{0,n}\right)\right\|_2^2=:B_n \\
\frac{1}{n}\left\|\mathbf{K}^{\frac{1}{2}}\left(\mathbf{Y}-\mathbf{V}\theta_{0,n}-\mathbf{Z}\left(\gamma_{0,n}\right)_{\hat{J}_n\cap I_n}\right)\right\|_2^2 \\
\quad\quad\quad\quad\quad\quad\quad-\frac{1}{n}\left\|\mathbf{K}^{\frac{1}{2}}\left(\mathbf{Y}-\mathbf{V}\theta_{0,n}-\mathbf{Z}\gamma_{0,n}\right)\right\|_2^2=C_n(I_n)
\end{array}\right.. \label{eq:BC}
\end{align}
Since the index set $I_n$ will not change during this proof, we write in the following $C_n=C_n(I_n)$. Let now, $\hat{\alpha}_n=(\hat{\theta}_n^{\top} , \hat{\gamma}_n^{\top})^{\top}-(\theta_{0,n}^{\top} , \gamma_{0,n}^{\top})^{\top}$. Then, $\left\|\left(\hat{\theta}_n-\theta_{0,n}\right)_{J_n^c}\right\|_0\leq|\hat{J}_n|$ and hence we can apply Lemma \ref{lem:Q} with $m_n=|\hat{J}_n|$ to obtain that
\begin{align*}
&\frac{1}{n}\Bigg|\left\|\mathbf{K}_h^{\frac{1}{2}}\left(\mathbf{Y}-\begin{pmatrix} \mathbf{V} & \mathbf{Z}\end{pmatrix}\begin{pmatrix}
\hat{\theta}_n \\ \hat{\gamma}_n
\end{pmatrix}\right)\right\|_2^2-\left\|\mathbf{K}_h^{\frac{1}{2}}\left(\mathbf{Y}-\begin{pmatrix} \mathbf{V} & \mathbf{Z}\end{pmatrix}\begin{pmatrix}
\theta_{0,n} \\ \gamma_{0,n}
\end{pmatrix}\right)\right\|_2^2 \\
&\quad\quad\quad\quad\quad-\left\|\mathbf{K}_h^{\frac{1}{2}}\begin{pmatrix} \mathbf{V} & \mathbf{Z}\end{pmatrix}\hat{\alpha}_n\right\|_2^2\Bigg| \\
\leq&\rho_n(|\hat{J}_n|)\frac{1}{\sqrt{n}}\left\|\mathbf{K}_h^{\frac{1}{2}}\begin{pmatrix} \mathbf{V} & \mathbf{Z}\end{pmatrix}\hat{\alpha}_n\right\|_2
\end{align*}
which in turn implies together with \eqref{eq:BC}
\begin{align*}
&\frac{1}{n}\left\|\mathbf{K}_h^{\frac{1}{2}}\begin{pmatrix} \mathbf{V} & \mathbf{Z}\end{pmatrix}\hat{\alpha}_n\right\|_2^2-\rho_n(|\hat{J}_n|)\frac{1}{\sqrt{n}}\left\|\mathbf{K}_h^{\frac{1}{2}}\begin{pmatrix} \mathbf{V} & \mathbf{Z}\end{pmatrix}\hat{\alpha}_n\right\|_2 \\
\leq&\frac{1}{n}\left\|\mathbf{K}_h^{\frac{1}{2}}\left(\mathbf{Y}-\begin{pmatrix} \mathbf{V} & \mathbf{Z}\end{pmatrix}\begin{pmatrix}
\hat{\theta}_n \\ \hat{\gamma}_n
\end{pmatrix}\right)\right\|_2^2-\frac{1}{n}\left\|\mathbf{K}_h^{\frac{1}{2}}\left(\mathbf{Y}-\begin{pmatrix} \mathbf{V} & \mathbf{Z}\end{pmatrix}\begin{pmatrix}
\theta_{0,n} \\ \gamma_{0,n}
\end{pmatrix}\right)\right\|_2^2\leq\min(B_n,C_n).
\end{align*}
It is elementary to prove that $x_n^2-p_nx_n\leq q_n$ for non-negative sequences $x_n,p_n$ and an arbitrary sequence $q_n$ implies $x_n\leq p_n+\sqrt{\max(q_n,0)}$. Hence, the above gives us
$$\frac{1}{\sqrt{n}}\left\|\mathbf{K}_h^{\frac{1}{2}}\begin{pmatrix} \mathbf{V} & \mathbf{Z}\end{pmatrix}\hat{\alpha}_n\right\|_2\leq \rho_n(|\hat{J}_n|)+\sqrt{\max(0,\min(B_n,C_n))}.$$
Thus, we obtain from \eqref{eq:rate1} that
\begin{equation}
\label{eq:rtg}
\left\|\begin{pmatrix}
\theta_{0,n}-\hat{\theta}_n \\ \gamma_{0,n}-\hat{\gamma}_n
\end{pmatrix}\right\|_2=O_P\left(\rho_n(|\hat{J}_n|)+\sqrt{\max(0,\min(B_n,C_n))}\right).
\end{equation}
We use this to show that \eqref{eq:l1} and \eqref{eq:l2} are both $o_P(1/\sqrt{nh})$. Note that \eqref{eq:l1} and \eqref{eq:l2} are both vectors of length $4$ and hence we may use any norm to study the asymptotics. Moreover, for any matrix $M$, we denote by $|M|_{\infty}$ the supremum over the absolute values of the entries of $M$. Set moreover $\zeta_n=\sqrt{\max(0,\min(B_n,C_n))}$. We begin with \eqref{eq:l1}. Note therefore that $\left\|\gamma_{0,n}-\hat{\gamma}_n\right\|_1\leq(|J_n|+|\hat{J}_n|)^{\frac{1}{2}}\left\|\gamma_{0,n}-\hat{\gamma}_n\right\|_2$. Then we obtain for some constant $C>0$ from Lemma \ref{lem:ZV} and by using \eqref{eq:rtg} in which we substitute the rate of $\rho(|\hat{J}_n|)$ from Lemma \ref{lem:Q}:
\begin{align*}
&\sqrt{nh}\|\eqref{eq:l1}\|_1\leq\sqrt{nh}\left|\frac{1}{n}\mathbf{V}^{\top}\mathbf{K}_h\mathbf{Z}\right|_{\infty}\left\|\gamma_{0,n}-\hat{\gamma}_n\right\|_1 \\
\leq&C\sqrt{nh}\left(\sqrt{\frac{\log p_n}{nh}}+h^2\right)\left(|J_n|+|\hat{J}_n|\right)^{\frac{1}{2}}\left(\left(|J_n|+|\hat{J}_n|\right)^{\frac{1}{2}}\left(\frac{1}{\sqrt{nh}}+\delta_n\right)+\zeta_n\right) \\
=&C\left(\sqrt{\log p_n}+\sqrt{nh^5}\right)\left(\left(|J_n|+|\hat{J}_n|\right)\left(\frac{1}{\sqrt{nh}}+\delta_n\right)+\left(|J_n|+|\hat{J}_n|\right)^{\frac{1}{2}}\zeta_n\right).
\end{align*}
The above converges to zero by the assumptions on the rates in (CMS). Hence, $\eqref{eq:l1}=o_P(1/\sqrt{nh})$. For \eqref{eq:l2} we obtain, for a different constant $C>0$, by applying the bound \eqref{eq:stand21} together with Lemma \ref{lem:ZV} and \eqref{eq:rtg} where we again replace $\rho(|\hat{J}_n|)$ as in Lemma \ref{lem:Q})
\begin{align*}
&\sqrt{nh}\|\eqref{eq:l2}\|_1\leq\sqrt{nh}\left|\frac{1}{n}\mathbf{V}^{\top}\mathbf{K}_h\mathbf{Z}(\hat{J}_n)\left(\mathbf{Z}(\hat{J}_n)^{\top}\mathbf{K}_h\mathbf{Z}(\hat{J}_n)\right)^{-1}\mathbf{Z}(\hat{J}_n)^{\top}\mathbf{K}_h\mathbf{V}\right|_{\infty}\left\|\hat{\theta}_n-\theta_{0,n}\right\|_1 \\
\leq&C|\hat{J}_n|\sqrt{nh}\sqrt{\Phi\left(\left|\hat{J}_n\right|,J_n\right)}\left(\sqrt{\frac{\log p_n}{nh}}+h^2\right)^2 \\
&\quad\quad\quad\quad\quad\times\left(\left(|J_n|+|\hat{J}_n|\right)\left(\frac{1}{\sqrt{nh}}+\delta_n\right)+\left(|J_n|+|\hat{J}_n|\right)^{\frac{1}{2}}\zeta_n\right) \\
\leq&C|\hat{J}_n|\sqrt{\Phi\left(\left|\hat{J}_n\right|,J_n\right)}\left(\sqrt{\frac{\log p_n}{nh}}+h^2\right) \\
&\quad\quad\quad\times \left(\sqrt{\log p_n}+\sqrt{nh^5}\right)\left(\left(|J_n|+|\hat{J}_n|\right)\left(\frac{1}{\sqrt{nh}}+\delta_n\right)+\left(|\hat{J}_n|+|J_n|\right)^{\frac{1}{2}}\zeta_n\right).
\end{align*}
Since we assume $\textrm{RSE}(|\hat{J}_n|,J_n,h)$ and the rates in (CMS), we have that the above converges to zero and the proof is complete.
\end{proof}

\subsubsection{Supporting Results}
\label{subsubsec:sup_res_gc}
In the following Lemma the difference between regression discontinuity design and and regular treatment effects becomes visible: We assume only asymptotically that the covariates and the treatment indicator are uncorrelated, moreover we do not model explicitly the relation between them. Thus we can only obtain that the correlation converges to zero.
\begin{lemma}
\label{lem:ZV}
Let $K$ be second order and compactly supported and let $h\to0$, $p_n\to\infty$ and $\log p_n/nh\to0$. Suppose that $f_X$ and $\mu_{Z^{(k)}}(x)$ are twice differentiable on a neighborhood around $0$ with $\mu_{Z^{(k)}}(0)=\mu_{Z^{(k)}}'(0)=0$ and $f_X''$ continuous and $\mu_{Z^{(k)}}''$ fulfill \eqref{eq:sd_bound}. Let furthermore \eqref{eq:mombound1} and \eqref{eq:mombound2} hold. Then, we have
\begin{align}
\sup_{a\in\{1,...,4\}}\left\|\IE\left(\frac{1}{n}\mathbf{Z}^{\top}\mathbf{K}_h\mathbf{V}_{\cdot a}\right)\right\|_{\infty}=O(h^2), \label{eq:ZV1} \\
\sup_{a\in\{1,...,4\}}\left\|\frac{1}{n}\mathbf{Z}^{\top}\mathbf{K}_h\mathbf{V}_{\cdot a}\right\|_{\infty}=O_P\left(\sqrt{\frac{\log p_n}{nh}}\right)+O(h^2). \label{eq:ZV2}
\end{align}
\end{lemma}
\begin{proof}
Since $a=1,...,4$ can only take finitely many values, it suffices to prove the statements \eqref{eq:ZV1} and \eqref{eq:ZV2} for an arbitrary $a\in\{1,...,4\}$. Let $a$ be thus fixed. Note that
$$\frac{1}{n}\mathbf{Z}^{\top}\mathbf{K}_h\mathbf{V}_{\cdot a}=\frac{1}{n}\sum_{i=1}^nK_h(X_i)Z_iV_i^{(a)}.$$
Note that the bounds on $\mu_{Z^{(k)}}''$ and the differentiability of $f_X$ imply a similar bound for $\left(\mu_{Z^{(k)}}f_X\right)''$, i.e.,
$$C=\sup_{n\in\IN}\sup_{k\in\{1,...,p_n\}}\sup_{u\in[0,1]}\left|\left(\mu_{Z^{(k)}}f_X\right)''(uh)\right|+\left|\left(\mu_{Z^{(k)}}f_X\right)''(-uh)\right|<\infty.$$
We begin by computing the expectation in \eqref{eq:ZV1}. By using that $\mu_{Z^{(k)}}(x)$ is differentiable and that $\mu_{Z^{(k)}}(0)=\mu_{Z^{(k)}}'(0)=0$ we obtain for all $k=1,...,p_n$ by an argument as in \eqref{eq:kernel} that:
\begin{align}
\left|\IE\left(K_h(X_i)Z_i^{(k)}V_i^{(a)}\right)\right|\leq&\frac{1}{2}h^2\cdot C. \label{eq:mom1} 
\end{align}
This implies \eqref{eq:ZV1} and that
\begin{equation}
\label{eq:step1}
\left\|\frac{1}{n}\mathbf{Z}^{\top}\mathbf{K}_h\mathbf{V}_{\cdot a}\right\|_{\infty}\leq\left\|\frac{1}{n}\mathbf{Z}^{\top}\mathbf{K}_h\mathbf{V}_{\cdot a}-\IE\left(\frac{1}{n}\mathbf{Z}^{\top}\mathbf{K}_h\mathbf{V}_{\cdot a}\right)\right\|_{\infty}+O(h^2).
\end{equation}
For the first part above we will apply Lemma \ref{lem:expbound} with $B_i=Z_i^{(k)}V_i^{(a)}$. The integral conditions hold true by assumption because $\IE\left(\left|Z_i^{(k)}V_i^{(a)}\right|^m\big|X_i=x\right)\leq (1+(|x|/h)^m)\mu_{k,m}(x)$. The conditions on $n,p_n,h$ follow because we assume that $h\to0$ and $\log p_n/nh\to0$. Note furthermore that all constants in the assumptions do not depend on $k$. We may thus apply Lemma \ref{lem:expbound} simultaneously for all $k=1,...,p_n$ and obtain for $x$ large enough as in Lemma \ref{lem:expbound} (but $x>1$) for $\epsilon_n=x\sqrt{\log p_n/nh}$ by the union bound that
\begin{align}
&\IP\left(\max_{k\in\{1,...,p_n\}}\frac{1}{n}\sum_{i=1}^n\left(K_h(X_i)Z_i^{(k)}V_i^{(a)}-\IE\left(K_h(X_i)Z_i^{(k)}V_i^{(a)}\right)\right)>\epsilon_n\right) \nonumber \\
\leq&p_n\max_{k\in\{1,...,p_n\}}\IP\left(\frac{1}{n}\sum_{i=1}^n\left(K_h(X_i)Z_i^{(k)}V_i^{(a)}-\IE\left(K_h(X_i)Z_i^{(k)}V_i^{(a)}\right)\right)>\epsilon_n\right) \nonumber \\
\leq& p_n^{1-x}\to0
\end{align}
which implies
$$\left\|\frac{1}{n}\mathbf{Z}^{\top}\mathbf{K}_h\mathbf{V}_{\cdot a}-\IE\left(\frac{1}{n}\mathbf{Z}^{\top}\mathbf{K}_h\mathbf{V}_{\cdot a}\right)\right\|_{\infty}=O_P\left(\sqrt{\frac{\log p_n}{nh}}\right).$$
This, together with \eqref{eq:step1}, completes the proof of \eqref{eq:ZV2}.
\end{proof}

The previous results are used in the following way to understand the behaviour of the projection matrix $\mathbf{M}_n(\hat{J}_n)$.
\begin{proposition}
\label{prop:VV}
Let all notation be as above and let the assumptions of Lemmas \ref{lem:SigmaV} and \ref{lem:ZV} be true. Suppose that $\textrm{RSE}(|\hat{J}_n|,J_n,h)$ holds and $\frac{|\hat{J}_n|\log p_n}{nh}=o_P(1)$ as well as $|\hat{J}_n|h^4=o_P(1)$. We have then that
\begin{align}
&\frac{1}{n}\mathbf{V}^{\top}\mathbf{K}_h^{\frac{1}{2}}\mathbf{M}_n(\hat{J}_n)\mathbf{K}_h^{\frac{1}{2}}\mathbf{V}\overset{\IP}{\to}f_X(0)\kappa(K) \label{eq:mat1} \\
&\frac{1}{n}\mathbf{V}^{\top}\mathbf{K}_h^{\frac{1}{2}}\mathbf{M}_n(\hat{J}_n)\mathbf{K}_h^{\frac{1}{2}}\mathbf{r}_n=\frac{1}{n}\mathbf{V}^{\top}\mathbf{K}_h\mathbf{r}_n+O_P\left(\left\|\frac{1}{n}\mathbf{Z}^{\top}\mathbf{K}_h\mathbf{r}_n\right\|_{\infty}\cdot|\hat{J}_n|\left(\sqrt{\frac{\log p_n}{nh}}+h^2\right)\right). \label{eq:mat2}
\end{align}
\end{proposition}
\begin{proof}
By definition \eqref{eq:defM} we have
$$\frac{1}{n}\mathbf{V}^{\top}\mathbf{K}_h^{\frac{1}{2}}\mathbf{M}_n(\hat{J}_n)\mathbf{K}_h^{\frac{1}{2}}\mathbf{V}=\frac{1}{n}\mathbf{V}^{\top}\mathbf{K}_h\mathbf{V}-\frac{1}{n}\mathbf{V}^{\top}\mathbf{K}_h\mathbf{Z}(\hat{J}_n)\left(\mathbf{Z}(\hat{J}_n)^{\top}\mathbf{K}_h\mathbf{Z}(\hat{J}_n)\right)^{-1}\mathbf{Z}(\hat{J}_n)^{\top}\mathbf{K}_h\mathbf{V}.$$
The first term above converges by Lemma \ref{lem:SigmaV} to the quantity we claim in the lemma. It remains to show that the second part converges to zero in probability. Recall to this end the definition of $\Phi(m_n,J_n)$ in Definition \ref{def:RSE}. We note that for all $a,b\in\{1,...,4\}$ (below $\|.\|$ denotes the Euclidean norm for vectors and the spectral norm for matrices)
\begin{align}
&\left[\frac{1}{n}\mathbf{V}^{\top}\mathbf{K}_h\mathbf{Z}(\hat{J}_n)\left(\mathbf{Z}(\hat{J}_n)^{\top}\mathbf{K}_h\mathbf{Z}(\hat{J}_n)\right)^{-1}\mathbf{Z}(\hat{J}_n)^{\top}\mathbf{K}_h\mathbf{V}\right]_{a,b} \nonumber \\
\leq&\left\|\left(\frac{1}{n}\mathbf{Z}(\hat{J}_n)^{\top}\mathbf{K}_h\mathbf{Z}(\hat{J}_n)\right)^{-1}\frac{1}{n}\mathbf{Z}(\hat{J}_n)^{\top}\mathbf{K}_h\mathbf{V}_{\cdot a}\right\|_1\cdot\left\|\frac{1}{n}\mathbf{Z}^{\top}\mathbf{K}_h\mathbf{V}_{\cdot b}\right\|_{\infty} \nonumber\\
\leq&|\hat{J}_n|^{\frac{1}{2}}\sqrt{\Phi(|\hat{J}_n|,J_n)}\cdot\left\|\frac{1}{n}\mathbf{Z}(\hat{J}_n)^{\top}\mathbf{K}_h\mathbf{V}_{\cdot a}\right\|\cdot\left\|\frac{1}{n}\mathbf{Z}(\hat{J}_n)^{\top}\mathbf{K}_h\mathbf{V}_{\cdot b}\right\|_{\infty} \nonumber \\
\leq&|\hat{J}_n|\sqrt{\Phi(|\hat{J}_n|,J_n)}\cdot\left\|\frac{1}{n}\mathbf{Z}^{\top}\mathbf{K}_h\mathbf{V}_{\cdot a}\right\|_{\infty}\cdot\left\|\frac{1}{n}\mathbf{Z}^{\top}\mathbf{K}_h\mathbf{V}_{\cdot b}\right\|_{\infty}. \label{eq:stand21}
\end{align}
By assumption we have $\Phi(|\hat{J}_n|,J_n)=O_P(1)$ and we just have to deal with the two infinity-norms. From Lemma \ref{lem:ZV} we find that
\begin{align*}
&|\hat{J}_n|\left(\sup_{a=1,...,4}\left\|\frac{1}{n}\mathbf{Z}^{\top}\mathbf{K}_h\mathbf{V}_{\cdot a}\right\|_{\infty}\right)^2=O_P\left(\frac{|\hat{J}_n|\log p_n}{nh}\right)+O_P(|\hat{J}_n|h^4)=o_P(1)
\end{align*}
by the assumptions on $|\hat{J}_n|$ and $h$. Thus, we have shown \eqref{eq:mat1}. In order to show \eqref{eq:mat2} we use very similar arguments: Firstly,
$$\frac{1}{n}\mathbf{V}^{\top}\mathbf{K}_h^{\frac{1}{2}}\mathbf{M}_n(\hat{J}_n)\mathbf{K}_h^{\frac{1}{2}}\mathbf{r}_n=\frac{1}{n}\mathbf{V}^{\top}\mathbf{K}_h\mathbf{r}_n-\frac{1}{n}\mathbf{V}^{\top}\mathbf{K}_h\mathbf{Z}(\hat{J}_n)\left(\mathbf{Z}(\hat{J}_n)^{\top}\mathbf{K}_h\mathbf{Z}(\hat{J}_n)\right)^{-1}\mathbf{Z}(\hat{J}_n)^{\top}\mathbf{K}_h\mathbf{r}_n.$$
The first part equals exactly the first part of \eqref{eq:mat2}. Thus we have left to prove that the second part above has the rate which appears in \eqref{eq:mat2}. This can be seen by the same reasoning as in \eqref{eq:stand21} and statement \eqref{eq:ZV2} from Lemma \ref{lem:ZV}. More precisely, let $a\in\{1,...,4\}$ be arbitrary, then we can finish the proof of the proposition as follows:
\begin{align*}
&\left[\frac{1}{n}\mathbf{V}^{\top}\mathbf{K}_h\mathbf{Z}(\hat{J}_n)\left(\mathbf{Z}(\hat{J}_n)^{\top}\mathbf{K}_h\mathbf{Z}(\hat{J}_n)\right)^{-1}\mathbf{Z}(\hat{J}_n)^{\top}\mathbf{K}_h\mathbf{r}_n\right]_a \\
\leq&|\hat{J}_n|\sqrt{\Phi(|\hat{J}_n|,J_n)}\cdot\left\|\frac{1}{n}\mathbf{Z}^{\top}\mathbf{K}_h\mathbf{V}_{\cdot a}\right\|_{\infty}\cdot\left\|\frac{1}{n}\mathbf{Z}^{\top}\mathbf{K}_h\mathbf{r}_n\right\|_{\infty} \\
=&O_P\left(\left\|\frac{1}{n}\mathbf{Z}^{\top}\mathbf{K}_h\mathbf{r}_n\right\|_{\infty}\cdot\left(|\hat{J}_n|\sqrt{\frac{\log p_n}{nh}}+h^2|\hat{J}_n|\right)\right)
\end{align*}
\end{proof}

\begin{lemma}
\label{lem:VR}
Suppose that $f_X$ is continuous and that for some $C\in(0,\infty)$,
\begin{equation}
\sup_{n\in\IN}\sup_{u\in[0,1]}\left|\IE(r_i(J_n,h)^2|X_i=uh)\right|+\left|\IE(r_i(J_n,h)^2|X_i=-uh)\right|\leq C. \label{eq:rsqbounded}
\end{equation}
Then,
\begin{equation}
\label{eq:VR}
\frac{1}{n}\mathbf{V}^{\top}\mathbf{K}_h\mathbf{r}_n=O_P\left(\frac{1}{\sqrt{nh}}\right).
\end{equation}
\end{lemma}
\begin{proof}
We prove \eqref{eq:VR} by an application of Markov's Inequality. Note that since $\mathbf{r}_n$ is a residual we have that $\IE\left(\mathbf{V}^{\top}\mathbf{K}_h\mathbf{r}_n\right)=0$. Since $\mathbf{V}$ has only four rows, we may just work for each row individually. We hence keep $a\in\{1,...,4\}$ arbitrary but fixed in the following. By Assumption \eqref{eq:rsqbounded} we find that
$$\IE\left(\frac{1}{h}K\left(\frac{X_i}{h}\right)^2V_{i,a}^2r_i(J_n,h)^2\right)=O(1)$$
Thus, for any $\epsilon>0$ we have by Markov's Inequality and independence
\begin{align*}
&\IP\left(\left|\frac{1}{n}\mathbf{V}_{\cdot a}^{\top}\mathbf{K}_h\mathbf{r}_n\right|>\epsilon\right)\leq\frac{1}{\epsilon^2}\IE\left(\left(\frac{1}{n}\sum_{i=1}^n\frac{1}{h}K\left(\frac{X_i}{h}\right)V_{i,a}r_i(J_n,h)\right)^2\right) \\
=&\frac{1}{nh\epsilon^2}\IE\left(\frac{1}{h}K\left(\frac{X_i}{h}\right)^2V_{i,a}^2r_i(J_n,h)^2\right)=O\left((nh\epsilon^2)^{-1}\right).
\end{align*}
From the above we conclude \eqref{eq:VR} and the proof of the lemma is complete.
\end{proof}

The following result is our version of Lemma 4 in \citet{BC13}. The proof is similar but we give it here for completeness.
\begin{lemma}
\label{lem:Q}
Suppose that $\textrm{RSE}(m_n,J_n,h)$ holds for some (possibly random) sequence $m_n$. Let furthermore the conditions of Lemma \ref{lem:VR} hold. Denote $\alpha=(\theta^{\top} , \gamma^{\top})^{\top}$ for $\theta\in\IR^4$ and $\gamma\in\IR^{p_n}$. Then we have for all $\alpha$ with $\|\gamma_{J_n^c}\|_0\leq m_n$,
\begin{align}
&\frac{1}{n}\Bigg|\left\|\mathbf{K}_h^{\frac{1}{2}}\left(\mathbf{Y}-\begin{pmatrix} \mathbf{V} & \mathbf{Z}\end{pmatrix}\left(\begin{pmatrix}
\theta_{0,n} \\ \gamma_{0,n}
\end{pmatrix}+\alpha\right)\right)\right\|_2^2-\left\|\mathbf{K}_h^{\frac{1}{2}}\left(\mathbf{Y}-\begin{pmatrix} \mathbf{V} & \mathbf{Z}\end{pmatrix}\begin{pmatrix}
\theta_{0,n} \\ \gamma_{0,n}
\end{pmatrix}\right)\right\|_2^2 \nonumber \\
&\quad\quad\quad\quad-\left\|\mathbf{K}_h^{\frac{1}{2}}\begin{pmatrix} \mathbf{V} & \mathbf{Z}\end{pmatrix}\alpha\right\|_2^2\Bigg|\leq\rho_n(m_n)\frac{1}{\sqrt{n}}\left\|\mathbf{K}_h^{\frac{1}{2}}\begin{pmatrix} \mathbf{V} & \mathbf{Z}\end{pmatrix}\alpha\right\|_2, \nonumber
\end{align}
where
$$\rho_n(m_n)=O_P\left(\sqrt{|J_n|+m_n}\left(\frac{1}{\sqrt{nh}}+\frac{1}{n}\left\|\mathbf{Z}^{\top}\mathbf{K}_h\mathbf{r}_n\right\|_{\infty}\right)\right).$$
\end{lemma}
\begin{proof}
For ease of notation we write $\mathbf{D}=\begin{pmatrix}
\mathbf{V} & \mathbf Z
\end{pmatrix}$. Then, we have
\begin{align*}
&\frac{1}{n}\left|\left\|\mathbf{K}_h^{\frac{1}{2}}\left(\mathbf{Y}-\mathbf{D}\left(\begin{pmatrix}
\theta_{0,n} \\ \gamma_{0,n}
\end{pmatrix}+\alpha\right)\right)\right\|_2^2-\left\|\mathbf{K}_h^{\frac{1}{2}}\left(\mathbf{Y}-\mathbf{D}\begin{pmatrix}
\theta_{0,n} \\ \gamma_{0,n}
\end{pmatrix}\right)\right\|_2^2-\left\|\mathbf{K}_h^{\frac{1}{2}}\mathbf{D}\alpha\right\|_2^2\right| \\
=&\frac{2}{n}\left|\alpha^{\top}\mathbf{D}^{\top}\mathbf{K}_h\left(\mathbf{Y}-\mathbf{D}\begin{pmatrix}
\theta_{0,n} \\ \gamma_{0,n}
\end{pmatrix}\right)\right|=\frac{2}{n}\left|\alpha^{\top}\begin{pmatrix}
\mathbf{V}^{\top} \\ \mathbf{Z}^{\top} \end{pmatrix}\mathbf{K}_h\mathbf{r_n}\right| \\
\leq&2\|\alpha\|_1\left(\frac{1}{n}\left\|\mathbf{V}^{\top}\mathbf{K}_h\mathbf{r}_n\right\|_{\infty}+\frac{1}{n}\left\|\mathbf{Z}^{\top}\mathbf{K}_h\mathbf{r}_n\right\|_{\infty}\right).
\end{align*}
Since $\|\alpha_{J_n^c}\|_0\leq m_n$, we have
$$\|\alpha\|_1\leq\sqrt{|J_n|+m_n}\cdot\|\alpha\|_2\leq\frac{\sqrt{|J_n|+m_n}}{\sqrt{\phi(m_n,J_n)}}\cdot\frac{1}{\sqrt{n}}\left\|\mathbf{K}_h^{\frac{1}{2}}\mathbf{D}\alpha\right\|_2,$$
where $\phi(m_n,J_n)$ is defined in the restricted sparse eigenvalue condition. Since we assume $\textrm{RSE}(m_n,J_n,h)$ and since we can apply Lemma \ref{lem:VR} the proof is complete.
\end{proof}

\subsection{Computing the Bias}
\subsubsection{The Result and Proof Structure}
The following lemma shows the form of the bias.
\begin{lemma}
\label{lem:bias}
Let $h\to0, nh\to\infty$, $|J_n|^{1/2}h^2\to0$ and let $K$ be symmetric with $\kappa(K)$ invertible and $K^{(4)},(K^2)^{(2)}<\infty$. Suppose that $f_X$ is three times differentiable with $f_X(0)>0$. Suppose furthermore that $\|\IE(K_h(X_i)Z_i(J_n)Z_i(J_n)^{\top}\|_2=O(1)$ and that the CEFs $\mu_{Z^{(k)}}(x)=\IE(Z_i^{(k)}|X_i=x)$ are differentiable with $\mu_{Z^{(k)}}(0)=\mu_{Z_i^{(k)}}'(0)=0$ and three times one-sided differentiable such that the third derivatives extend continuously to zero with
\begin{align}
\sup_{n\in\IN}\sup_{k\in J_n}\sup_{u\in[0,1]}\left|\left(\mu_{Z^{(k)}}f_X\right)''(uh)\right|+\left|\left(\mu_{Z^{(k)}}f_X\right)''(-uh)\right|<\infty, \label{eq:zsd} \\
\sup_{n\in\IN}\sup_{k\in J_n}\sup_{u\in[0,1]}\left|\left(\mu_{Z^{(k)}}f_X\right)'''(uh)\right|+\left|\left(\mu_{Z^{(k)}}f_X\right)'''(-uh)\right|<\infty, \label{eq:ztd} \\
\sup_{n\in\IN}\sup_{k\in J_n}\sup_{u\in[0,1]}\left|\IE\left(Z_i^{(k)}Y_i\big|X_i=uh\right)\right|+\left|\IE\left(Z_i^{(k)}Y_i\big|X_i=-uh\right)\right|<\infty. \label{eq:yzcv}
\end{align}
Suppose that $\mu_Y(x)=\IE(Y_i|X_i=x)$ is three times one-sided differentiable. Denote
\begin{align*}
\mathcal{B}_n&=\frac{1}{2}\frac{K_+^{(3)}-2K_+^{(1)}K_+^{(2)}}{K_+^{(2)}-2\left(K_+^{(1)}\right)^2}\Bigg[\mu_{Y+}''-\mu_{Y-}'' \\
&\quad\quad-\sum_{k\in J_n}\left(\mu_{Z^{(k)}+}''-\mu_{Z^{(k)}-}''\right)\left[\IE(K_h(X_i)Z_iZ_i^{\top})^{-1}\IE(K_h(X_i)Z_iY_i)\right]_k\Bigg] \\
&\quad\quad\quad\quad+o(1)+O(|J_n|h^2)+O(|J_n|^{\frac{1}{2}}h).
\end{align*}
 Then,
$$\theta_{0,n}^{(2)}=\tau+h^2\mathcal{B}_n.$$
\end{lemma}
In order to reduce the notation in the proof we omit, in this subsection only, the subscript $J_n$ on the covariates $Z_i(J_n)$. Thus, in the following $Z_i\in\IR^{n\times |J_n|}$.
\begin{proof}[Proof of Lemma \ref{lem:bias}]
Let $a_1,a_2,b_1,b_2$ be as in Lemma \ref{lem:kernel} and define
\begin{align*}\kappa_{h,b}(K)=&\left[\left(I+\frac{f_X'(0)}{f_X(0)}h\begin{pmatrix}
0 & 0 & a_1 & 0 \\
0 & 0 & 0   & a_1 \\
1 & 0 & -a_2 & 0 \\
0 & 1 & 2a_2 & a_2
\end{pmatrix}+h^2\frac{f_X''(0)}{2f_X(0)}\begin{pmatrix}
a_1  & 0   & -b_1 & 0   \\
0    & a_1 & 2b_1 & b_1 \\
-a_2 & 0   &  b_2 & 0   \\
2a_2 & a_2 &    0 & b_2
\end{pmatrix}\right)^{-1}\right]_{2\cdot} \\
&\quad\quad\quad\quad-\begin{pmatrix}
0 & 1 & 0 & 0
\end{pmatrix}.
\end{align*}
By using least squares algebra, we obtain
\begin{align}
&\theta_{0,n}^{(2)} \nonumber \\
=&\left[\left(\IE(K_h(X_i)V_iV_i^{\top})-\IE(K_h(X_i)V_iZ_i^{\top})\IE(K_h(X_i)Z_iZ_i^{\top})^{-1}\IE(K_h(X_i)Z_iV_i^{\top})\right)^{-1}\right]_{2\cdot} \nonumber \\
&\times\left(\IE(K_h(X_i)V_iY_i)-\IE(K_h(X_i)V_iZ_i^{\top})\IE(K_h(X_i)Z_iZ_i^{\top})^{-1}\IE(K_h(X_i)Z_iY_i)\right) \nonumber \\
=&\left[\left(I-\IE\left(K_h(X_i)V_iV_i^{\top}\right)^{-1}\IE(K_h(X_i)V_iZ_i^{\top})\IE(K_h(X_i)Z_iZ_i^{\top})^{-1}\IE(K_h(X_i)Z_iV_i^{\top})\right)^{-1}\right]_{2\cdot} \nonumber \\
&\quad\quad\times\left(\kappa(K)^{-1}\IE(K_h(X_i)V_iV_i^{\top})\right)^{-1} \nonumber \\
&\times\kappa(K)^{-1}\left(\IE(K_h(X_i)V_iY_i)-\IE(K_h(X_i)V_iZ_i^{\top})\IE(K_h(X_i)Z_iZ_i^{\top})^{-1}\IE(K_h(X_i)Z_iY_i)\right) \label{eq:begin_lem12}
\end{align}
By the approximations from Lemma \ref{lem:help2} and \ref{lem:help3} we obtain
\begin{align}
&\left[\left(I-\IE\left(K_h(X_i)V_iV_i^{\top}\right)^{-1}\IE(K_h(X_i)V_iZ_i^{\top})\IE(K_h(X_i)Z_iZ_i^{\top})^{-1}\IE(K_h(X_i)Z_iV_i^{\top})\right)^{-1}\right]_{2\cdot} \nonumber \\
&\quad\quad\times\left(\kappa(K)^{-1}\IE(K_h(X_i)V_iV_i^{\top})\right)^{-1} \nonumber \\
=&\left(\begin{pmatrix}
0 & 1 & 0 & 0
\end{pmatrix}+O(|J_n|h^4)\right) \nonumber \\
&\times\left(f_X(0)I+f_X'(0)h\begin{pmatrix}
0 & 0 & a_1 & 0 \\
0 & 0 & 0   & a_1 \\
1 & 0 & -a_2 & 0 \\
0 & 1 & 2a_2 & a_2
\end{pmatrix}+h^2\frac{f_X''(0)}{2}\begin{pmatrix}
a_1  & 0   & -b_1 & 0   \\
0    & a_1 & 2b_1 & b_1 \\
-a_2 & 0   &  b_2 & 0   \\
2a_2 & a_2 &    0 & b_2
\end{pmatrix}+o(h^2)\right)^{-1} \nonumber \\
=&\frac{1}{f_X(0)}\left[\left(I+\frac{f_X'(0)}{f_X(0)}h\begin{pmatrix}
0 & 0 & a_1 & 0 \\
0 & 0 & 0   & a_1 \\
1 & 0 & -a_2 & 0 \\
0 & 1 & 2a_2 & a_2
\end{pmatrix}+h^2\frac{f_X''(0)}{2f_X(0)}\begin{pmatrix}
a_1  & 0   & -b_1 & 0   \\
0    & a_1 & 2b_1 & b_1 \\
-a_2 & 0   &  b_2 & 0   \\
2a_2 & a_2 &    0 & b_2
\end{pmatrix}\right)^{-1}\right]_{2\cdot} \nonumber \\
&\quad\quad\quad\quad\quad\quad\quad\quad\quad\quad\quad\quad\quad\quad\quad\quad\quad\quad\quad\quad+o(h^2)+O\left(|J_n|h^4\right) \nonumber \\
=&\frac{1}{f_X(0)}\left(\begin{pmatrix}
0 & 1 & 0 & 0
\end{pmatrix}+\kappa_{h,b}(K)\right)+o(h^2)+O\left(|J_n|h^4\right). \label{eq:c1}
\end{align}
From Lemma \ref{lem:help1} we obtain with $A=Y_i$
\begin{align}
\kappa(K)^{-1}\IE(K_h(X_i)V_iY_i)=\begin{pmatrix}
f_X(0)\mu_{Y-} \\ f_X(0)\tau \\ h\left[\mu_Yf_X\right]_-' \\ h\left(\left[\mu_Yf_X\right]_+'-\left[\mu_Yf_X\right]_-'\right)
\end{pmatrix}+h^2B(K,Y)+O(h^3). \label{eq:C2}
\end{align}
We apply now Lemma \ref{lem:help1} with $A=Z_i^{(k)}$ for each $k\in J_n$. Recall for that case  that $\mu_{Z^{(k)}}(0)=\mu_{Z^{(k)}}'(0)=0$. Let furthermore $B_k$ be the vector $B(K,A)$ for $A=Z_i^{(k)}$ and denote by $B\in\IR^{4\times |J_n|}$ the matrix which has the vectors $B_k$ as columns. With these definitions, we obtain from Lemma \ref{lem:help1}
\begin{align}
&\kappa(K)^{-1}\IE(K_h(X_i)V_iZ_i^{\top})\IE(K_h(X_i)Z_iZ_i^{\top})^{-1}\IE(K_h(X_i)Z_iY_i) \nonumber \\
=&h^2B\IE(K_h(X_i)Z_iZ_i^{\top})^{-1}\IE(K_h(X_i)Z_iY_i)+O(h^3)\IE(K_h(X_i)Z_iZ_i^{\top})^{-1}\IE(K_h(X_i)Z_iY_i). \label{eq:c3}
\end{align}
where $O(h^3)$ denotes a $4\times |J_n|$ matrix of entries of order $O(h^3)$ where for all of them the same constant can be used because we assume boundedness of the corresponding derivatives. Similarly, the entries of $B$ are uniformly bounded. We know from Lemma~\ref{lem:bounded_gamma} that
$$\left\|\IE\left(K_h(X_i)Z_iZ_i^\top\right)^{-1}\IE\left(K_h(X_i)Z_iY_i\right)\right\|_2=O(1).$$
Hence, by using the Cauchy-Schwarz Inequality we obtain for each entry $a\in\{1,...,4\}$ that
$$\left|[\eqref{eq:c3}]^{(a)}\right|\leq h^2\|B_{a\cdot}\|_2O(1)+\left\|\left[O(h^3)\right]_{a\cdot}\right\|_2O(1)=O\left(|J_n|^{\frac{1}{2}}h^2\right)+O\left(|J_n|^{\frac{1}{2}}h^3\right).$$
Hence, by bringing together the considerations \eqref{eq:c1}-\eqref{eq:c3} and by using that $\kappa_{h,b}(K)=O(h)$ and $|J_n|^{\frac{1}{2}}h^2\to0$ we finally obtain that
\begin{align}
&\theta_{0,n}^{(2)}=\left[\frac{1}{f_X(0)}\left(\begin{pmatrix}
0 & 1 & 0 & 0
\end{pmatrix}+\kappa_{h,b}(K)\right)+o(h^2)+O\left(|J_n|h^4\right)\right] \nonumber \\
&\times\kappa(K)^{-1}\left(\IE(K_h(X_i)V_iY_i)-\IE(K_h(X_i)V_iZ_i^{\top})\IE(K_h(X_i)Z_iZ_i^{\top})^{-1}\IE(K_h(X_i)Z_iY_i)\right) \nonumber \\
=&\tau+\frac{h^2}{f_X(0)}B^{(2)}(K,Y)-\frac{h^2}{f_X(0)}B_{2\cdot}\IE(K_h(X_i)Z_iZ_i^{\top})^{-1}\IE(K_h(X_i)Z_iY_i) \nonumber \\
&\quad+\kappa_{h,b}(K)\begin{pmatrix}
\mu_{Y-} \\ \tau \\ \frac{h}{f_X(0)}\left[\mu_Yf_X\right]_-' \\ \frac{h}{f_X(0)}\left(\left[\mu_Yf_X\right]_+'-\left[\mu_Yf_X\right]_-'\right)
\end{pmatrix}+o(h^2)+O(|J_n|h^4)+O(|J_n|^{\frac{1}{2}}h^3). \label{eq:almostdone}
\end{align}
The completion of the proof is straight forward but computationally tedious. We firstly need to compute $\kappa_{h,b}(K)$ explicitly. This can be done by using a computer algebra system and we report here only the result (see Lemma \ref{lem:kernel} for a definition of $a_1$)
$$\mathcal{\kappa}_{h,b}(K)=\begin{pmatrix}
o(h^2) & h^2a_1\left(\frac{f_X'(0)^2}{f_X(0)^2}-\frac{f_X''(0)}{2f_X(0)}\right)+o(h^2) & O(h^2) & -ha_1\frac{f_X'(0)}{f_X(0)}+O(h^2)
\end{pmatrix}.$$
Moreover, we compare the vectors in the definition of $B(K,A)$ with the columns of the matrix in \eqref{eq:mateq1}. Thus, we obtain from Lemma \ref{lem:kernel} that
\begin{align*}
B^{(2)}(K,A)=&\frac{a_1}{2}\Big(\mu_{A+}''f_X(0)+2\mu_{A+}'f_X'(0)+\mu_{A+}f_X''(0) \\
&-\mu_{A-}''f_X(0)-2\mu_{A-}'f_X'(0)-\mu_{A-}f_X''(0)\Big).
\end{align*}
Hence, in particular
\begin{align}
B^{(2)}(K,Y)=&\frac{a_1}{2}\Big(\mu_{Y+}''f_X(0)+2\mu_{Y+}'f_X'(0)+\mu_{Y+}f_X''(0) \nonumber \\
&-\mu_{Y-}''f_X(0)-2\mu_{Y-}'f_X'(0)-\mu_{Y-}f_X''(0)\Big), \nonumber \\
B^{(2)}(K,Z^{(k)})=&\frac{a_1}{2}\Big(\mu_{Z^{(k)}+}''f_X(0)-\mu_{Z^{(k)}-}''f_X(0)\Big). \label{eq:BZk}
\end{align}
By combining the previous equations and replacing $\tau=\mu_{Y+}-\mu_{Y-}$, we finally obtain the desired expression for the bias: Firstly,
\begin{align*}
&\frac{h^2}{f_X(0)}B^{(2)}(K,Y)+\kappa_{h,b}(K)\begin{pmatrix}
\mu_{Y-} \\ \tau \\ \frac{h}{f_X(0)}\left[\mu_Yf_X\right]_-' \\ \frac{h}{f_X(0)}\left(\left[\mu_Yf_X\right]_+'-\left[\mu_Yf_X\right]_-'\right)
\end{pmatrix} \\
=&h^2a_1\Bigg(\frac{1}{2}\left(\mu_{Y+}''-\mu_{Y-}''\right)+\frac{f_X'(0)}{f_X(0)}\left(\mu_{Y+}'-\mu_{Y-}'\right)+\frac{1}{2}\frac{f_X''(0)}{f_X(0)}\left(\mu_{Y+}-\mu_{Y-}\right) \\
&-\frac{1}{2}\frac{f_X''(0)}{f_X(0)}\left(\mu_{Y+}-\mu_{Y-}\right)+\frac{f_X'(0)^2}{f_X(0)^2}\left(\mu_{Y+}-\mu_{Y-}\right) \\
&-\frac{f_X'(0)}{f_X(0)^2}\left(\mu_{Y+}'f_X(0)+\mu_{Y+}f_X(0)'-\mu_{Y-}'f_X(0)-\mu_{Y-}f_X(0)'\right)\Bigg)+o(h^2) \\
=&h^2\frac{a_1}{2}\left(\mu_{Y+}''-\mu_{Y-}''\right)+o(h^2).
\end{align*}
And, secondly, we see that $B_{2\cdot}$ has entries given by \eqref{eq:BZk} and hence,
\begin{align*}
&\frac{h^2}{f_X(0)}B_{2\cdot}\IE(K_h(X_i)Z_iZ_i^{\top})^{-1}\IE(K_h(X_i)Z_iY_i) \\
=&h^2\frac{a_1}{2}\sum_{k\in J_n}\left(\mu_{Z^{(k)}+}''-\mu_{Z^{(k)}-}''\right)\left[\IE(K_h(X_i)Z_iZ_i^{\top})^{-1}\IE(K_h(X_i)Z_iY_i)\right]_k.
\end{align*}
Replacing the above two expressions in \eqref{eq:almostdone} completes the proof.
\end{proof}

\subsubsection{Supporting Results}
\begin{lemma}
\label{lem:bounded_gamma}
Let $\left\|\IE\left(K_h(X_i)Z_iZ_i^\top\right)^{-1}\right\|_2=O(1)$. Then,
$$\left\|\IE\left(K_h(X_i)Z_iZ_i^{\top}\right)^{-1}\IE\left(K_h(X_i)Z_iY_i\right)\right\|_2=O(1).$$
If $\left|\left(\sigma_{Z-}^2+\sigma_{Z+}^2\right)^{-1}\right\|_2=O(1)$ we have that $\|\gamma_n\|_2=O(1)$, where $\gamma_n$ is defined below \eqref{eq:approx_bias_var}.
\end{lemma}
\begin{proof}
Denote $\beta=\IE\left(K_h(X_i)Z_iZ_i^{\top}\right)^{-1}\IE\left(K_h(X_i)Z_iY_i\right)$ and let $\epsilon=Y_i-Z_i^\top\beta$. Then, $\IE\left(K_h(X_i)Z_i\epsilon\right)=0$. We want to proof that
$$\|\beta\|_2^2=\IE\left(K_h(X_i)Z_i^\top Y_i\right)\IE\left(K_h(X_i)Z_iZ_i^\top\right)^{-2}\IE\left(K_h(X_i)Z_iY_i\right)=O(1).$$
Using all these properties and notation, we get for any constant $c>0$
\begin{align*}
&\IE\left(K_h(X_i)Y_i^2\right)-c\|\beta\|_2^2=\IE\left(K_h(X_i)\epsilon^2\right)+\beta^\top\IE\left(K_h(X_i)Z_iZ_i^\top\right)\beta-c\|\beta\|_2^2 \\
\geq&\IE\left(K_h(X_i)Z_i^\top Y_i\right)\IE\left(K_h(X_i)Z_iZ_i^\top\right)^{-1}\IE\left(K_h(X_i)Z_iY_i\right)-c\|\beta\|_2^2 \\
=&\IE\left(K_h(X_i)Z_i^\top Y_i\right)\left(\IE\left(K_h(X_i)Z_iZ_i^\top\right)^{-1}-c\IE\left(K_h(X_i)Z_iZ_i^\top\right)^{-2}\right)\IE\left(K_h(X_i)Z_iY_i\right).
\end{align*}
Hence, the proof is complete if there is a constant $c>0$ such that
$$M_c=\IE\left(K_h(X_i)Z_iZ_i^\top\right)^{-1}-c\IE\left(K_h(X_i)Z_iZ_i^\top\right)^{-2}$$
is positive semi-definite. Recall that $\|A\|_2$ denote the largest eigenvalue (in absolute value) of $A$ if $A$ is a symmetric matrix. By assumption we can choose $0<c\leq\left\|\IE\left(K_h(X_i)Z_iZ_i^\top\right)^{-1}\right\|_2^{-1}$, that is, whenever $\mu$ is an eigenvalue of $\IE\left(K_h(X_i)Z_iZ_i^\top\right)^{-1}$ it holds that $c\leq\mu^{-1}$. Let now $(\mu,v)$ be an eigenvalue-eigenvector pair of $\IE\left(K_h(X_i)Z_iZ_i^\top\right)^{-1}$. We get
$$M_cv=\left(\IE\left(K_h(X_i)Z_iZ_i^\top\right)^{-1}-c\IE\left(K_h(X_i)Z_iZ_i^\top\right)^{-2}\right)v=(\mu-c\mu^2)v.$$
Since for each symmetric matrix an orthogonal basis of eigenvectors can be found, we see that all eigenvalues of $M_c$ are of the form $\mu-c\mu^2$ where $\mu$ is an eigenvalue of $\IE\left(K_h(X_i)Z_iZ_i^\top\right)^{-1}$. By the choice of $c$ we have $\mu-c\mu^2\geq0$ and hence we conclude that $M_c$ is positive semi-definite.

The proof for $\gamma_n$ can be carried out along the same lines. Here the starting point is $\epsilon_+=Y_i-\mu_{Y+}-(Z_i-\mu_{Z+})^\top\gamma_n$ and $\epsilon_-=Y_i-\mu_{Y-}-(Z_i-\mu_{Z-})^\top\gamma_n$. Then,
\begin{align*}
&\IE(\epsilon_+(Z_i-\mu_{Z+})|X_i=0+)+\IE(\epsilon_-(Z_i-\mu_{Z-})|X_i=0-) \\
=&\sigma_{YZ+}^2+\sigma_{ZY-}^2-\left(\sigma_{Z+}^2+\sigma_{Z-}^2\right)\gamma_n=0
\end{align*}
and consequently
\begin{align*}
&\IE(Y_i^2|X_i=0+)+\IE(Y_i^2|X_i=0-) \\
=&\IE(\epsilon_+^2|X_i=0+)+\IE(\epsilon_-^2|X_i=0-)+\gamma_n^\top\left(\sigma_{Z+}^2+\sigma_{Z-}^2\right)\gamma_n.
\end{align*}
Now we can proceed in the same way as before.
\end{proof}

\begin{lemma}
\label{lem:help1}
Let $f_X$ be three times continuously differentiable. $A$ be an arbitrary random variable such that $\mu_A(x)=\IE(A|X_i=x)$ is well defined and is three times one-sided differentiable at $0$. The derivatives extend continuously to $0$ and the third derivatives are bounded around $0$. Suppose moreover that the kernel $K$ is symmetric with $K^{(4)}<\infty$ (in particular: $\kappa(K)$ is invertible by Lemma \ref{lem:kappa}). Define $\tau_A=\mu_{A+}-\mu_{A-}$ and
$$B(K,A)=\frac{1}{2}\kappa(K)^{-1}\left[\begin{pmatrix}
K_+^{(2)} \\ K_+^{(2)} \\ K_+^{(3)} \\ K_+^{(3)}
\end{pmatrix}\left[\mu_Af_X\right]_+''+\begin{pmatrix}
K_-^{(2)} \\ 0 \\ K_-^{(3)} \\ 0 \end{pmatrix}\left[\mu_Af_X\right]_-''\right].$$
Then,
\begin{align*}
&\kappa(K)^{-1}\IE(K_h(X_i)V_iA)=\begin{pmatrix}
f_X(0)\mu_{A-} \\ f_X(0)\tau_A \\ h\left[\mu_Af_X\right]_-' \\ h\left(\left[\mu_Af_X\right]_+'-\left[\mu_Af_X\right]_-'\right)\end{pmatrix}+h^2B(K,A)+O(h^3).
\end{align*}
\end{lemma}
\begin{proof}
By using an argument as in \eqref{eq:kernel} we obtain
\begin{align}
\IE\left(K_h(X_i)V_iA\right)=&\begin{pmatrix}
K_+^{(0)} \\ K_+^{(0)} \\ K_+^{(1)} \\ K_+^{(1)}
\end{pmatrix}\mu_{A+}f_X(0)+\begin{pmatrix}
K_-^{(0)} \\ 0 \\ K_-^{(1)} \\ 0 \end{pmatrix}\mu_{A-}f_X(0) \label{eq:b1} \\
&+h\left[\begin{pmatrix}
K_+^{(1)} \\ K_+^{(1)} \\ K_+^{(2)} \\ K_+^{(2)}
\end{pmatrix}\left[\mu_Af_X\right]_+'+\begin{pmatrix}
K_-^{(1)} \\ 0 \\ K_-^{(2)} \\ 0 \end{pmatrix}\left[\mu_Af_X\right]_-'\right] \label{eq:b2} \\
&+\frac{h^2}{2}\left[\begin{pmatrix}
K_+^{(2)} \\ K_+^{(2)} \\ K_+^{(3)} \\ K_+^{(3)}
\end{pmatrix}\left[\mu_Af_X\right]_+''+\begin{pmatrix}
K_-^{(2)} \\ 0 \\ K_-^{(3)} \\ 0 \end{pmatrix}\left[\mu_Af_X\right]_-''\right]+O(h^3). \label{eq:b3}
\end{align}
We treat the expression above line by line. For \eqref{eq:b1} we note that $\mu_{A+}=\tau_A+\mu_{A-}$. Thus we obtain by comparing with the columns of $\kappa(K)$ and by using the kernel properties that
\begin{align*}
\kappa(K)^{-1}\eqref{eq:b1}=f_X(0)\kappa(K)^{-1}\left[\tau_A\begin{pmatrix}
K_+^{(0)} \\ K_+^{(0)} \\ K_+^{(1)} \\ K_+^{(1)}
\end{pmatrix}+\mu_{A-}\begin{pmatrix}
1 \\ K_+^{(0)} \\ 0 \\ K_+^{(1)}
\end{pmatrix}\right]=f_X(0)\begin{pmatrix}
\mu_{A-} \\ \tau_A \\ 0 \\ 0
\end{pmatrix}.
\end{align*}
For \eqref{eq:b2} we obtain by the same argument
\begin{align*}
\kappa(K)^{-1}\eqref{eq:b2}=&h\kappa(K)^{-1}\left[\begin{pmatrix}
K_+^{(1)} \\ K_+^{(1)} \\ K_+^{(2)} \\ K_+^{(2)}
\end{pmatrix}\left[\mu_Af_X\right]_++\left(\begin{pmatrix}
0 \\ K_+^{(1)} \\ K^{(2)} \\ K_+^{(2)} \end{pmatrix}-\begin{pmatrix}
K_+^{(1)} \\ K_+^{(1)} \\ K_+^{(2)} \\ K_+^{(2)} \end{pmatrix}\right)\left[\mu_Af_X\right]_-'\right] \\
=&h\begin{pmatrix}
0 \\ 0 \\ \left[\mu_Af_X\right]_-' \\ \left[\mu_Af_X\right]_+'-\left[\mu_Af_X\right]_-'
\end{pmatrix}.
\end{align*}
This proves the statement because $\kappa(K)^{-1}\eqref{eq:b3}=h^2B(K,A)+O(h^3)$.
\end{proof}

\begin{lemma}
\label{lem:help2}
Let $K$ be symmetric with $K^{(4)},(K^2)^{(2)}<\infty$, $f_X$ three times differentiable with $f_X(0)\neq0$, $\|\IE(K_h(X_i)Z_iZ_i^{\top})^{-1}\|_2=O(1)$ and $h\to0,nh\to\infty$. Suppose that for all $n\in\IN$ and all $k\in J_n$ the functions $\mu_{Z^{(k)}}(x)=\IE(Z_i^{(k)}|X_i=x)$ are differentiable with $\mu_{Z^{(k)}}(0)=\mu_{Z^{(k)}}'(0)=0$ and one-sided differentiable up to order three. The third derivatives extend continuously to zero and fulfil \eqref{eq:zsd} and \eqref{eq:ztd}. Then,
\begin{align*}
&\left[\left(I-\IE\left(K_h(X_i)V_iV_i^{\top}\right)^{-1}\IE(K_h(X_i)V_iZ_i^{\top})\IE(K_h(X_i)Z_iZ_i^{\top})^{-1}\IE(K_h(X_i)Z_iV_i^{\top})\right)^{-1}\right]_{2\cdot} \\
=&\begin{pmatrix}
0 & 1 & 0 & 0 \end{pmatrix}+O(|J_n|h^4).
\end{align*}
\end{lemma}
\begin{proof}
Let $a,b\in\{1,...,4\}$ and $k\in J_n$ be arbitrary. We have by an expansion of the type \eqref{eq:kernel} for the choices $L(u)=K(u)$, $L(u)=K(u)\Ind(u\geq0)$, $L=K(u)u$ and $L(u)=K(u)u\Ind(u\geq0)$ because $\mu_{Z^{(k)}}(0)=\mu_{Z^{(k)}}'(0)=0$ that
\begin{align*}
&\left[\gamma_{V,n}\right]_{k,\cdot}=\IE\left(K_h(X_i)V_iZ_i^{(k)}\right)^{\top} \\
=&\frac{1}{2}h^2\begin{pmatrix}\left[\mu_{Z^{(k)}}f_X\right]_-'' & \left[\mu_{Z^{(k)}}f_X\right]_+''-\left[\mu_{Z^{(k)}}f_X\right]_-'' \end{pmatrix}\begin{pmatrix}
K^{(2)}   & K_+^{(2)} & 0         & K_+^{(3)} \\
K_+^{(2)} & K_+^{(2)} & K_+^{(3)} & K_+^{(3)}
\end{pmatrix}+O(h^3).
\end{align*}
Here $O(h^3)$ has to be understood as row-vector where all entries are $O(h^3)$. By the assumptions on the third derivatives, we have, in addition, that the $O(h^3)$ has the same constants for all choices of $k\in J_n$. Thus, we find a constant $C>0$ such that for all $a,b\in\{1,...,4\}$
\begin{align*}
&\left[\gamma_{V,n}^{\top}\IE(K_h(X_i)Z_iZ_i^{\top})^{-1}\gamma_{V,n}\right]_{a,b}^2\leq\left\|\left[\gamma_{V,n}\right]_{\cdot a}\right\|_2^2\left\|\IE(K_h(X_i)Z_iZ_i^{\top})^{-1}\right\|_2^2\left\|\left[\gamma_{V,n}\right]_{\cdot,b}\right\|_2^2 \\
\leq&|J_n|^2Ch^8\left\|\IE(K_h(X_i)Z_iZ_i^{\top})^{-1}\right\|_2^2=O(|J_n|^2h^8)
\end{align*}
since $\left\|\IE(K_h(X_i)Z_iZ_i^{\top})^{-1}\right\|_2^2=O(1)$ by assumption. Hence we obtain that
$$\kappa(K)^{-1}\gamma_{V,n}^{\top}\IE(K_h(X_i)Z_iZ_i^{\top})^{-1}\gamma_{V,n}=O(|J_n|h^4),$$
where $O(|J_n|h^4)$ means here a $4\times4$ matrix which entries are each of order $O(|J_n|h^4)$. We obtain from \eqref{eq:kernel_exp} in Lemma \ref{lem:SigmaV} that (use that matrix inversion is a continuous operation)
$$\IE\left(K_h(X_i)V_iV_i^{\top}\right)^{-1}=\frac{1}{f_X(0)}\kappa(K)^{-1}+O(h).$$
Bringing the previous two results together, we get
\begin{align*}
 &I-\IE\left(K_h(X_i)V_iV_i^{\top}\right)^{-1}\gamma_{V,n}^{\top}\IE(K_h(X_i)Z_iZ_i^{\top})^{-1}\gamma_{V,n} \\
=&I-\frac{1}{f_X(0)}\kappa(K)^{-1}\gamma_{V,n}^{\top}\IE(K_h(X_i)Z_iZ_i^{\top})^{-1}\gamma_{V,n} \\
&\quad\quad\quad\quad-\left(\IE\left(K_h(X_i)V_iV_i^{\top}\right)^{-1}-\frac{1}{f_X(0)}\kappa(K)^{-1}\right)\gamma_{V,n}^{\top}\IE(K_h(X_i)Z_iZ_i^{\top})^{-1}\gamma_{V,n} \\
=&I+O(|J_n|h^4).
\end{align*}
By using the formula for the inverse of block matrices we can read off for the second row
\begin{align*}
&\left[\left(I-\IE\left(K_h(X_i)V_iV_i^{\top}\right)^{-1}\gamma_{V,n}^{\top}\IE\left(K_h(X_i)Z_iZ_i^\top\right)^{-1}\gamma_{V,n}\right)^{-1}\right]_{2\cdot} \\
=&\begin{pmatrix}
0 & 1 & 0 & 0 \end{pmatrix}+O(|J_n|h^4).
\end{align*}
\end{proof}

\begin{lemma}
\label{lem:kernel}
Let $K$ be a symmetric kernel such that $\kappa(K)$ is invertible and $K^{(4)}<\infty$. Then,
$$\kappa(K)^{-1}=\frac{1}{\left(K_+^{(1)}\right)^2-\frac{1}{2}K_+^{(2)}}\begin{pmatrix}
-K_+^{(2)} &   K_+^{(2)} & -K_+^{(1)}   & K_+^{(1)}   \\
 K_+^{(2)} & -2K_+^{(2)} &  K_+^{(1)}   &           0   \\
-K_+^{(1)} &   K_+^{(1)} & -\frac{1}{2} & \frac{1}{2} \\
 K_+^{(1)} &           0 &  \frac{1}{2} &          -1 
\end{pmatrix}.$$
In particular, for
\begin{align*}
a_2=&\frac{K_+^{(3)}-2K_+^{(1)}K_+^{(2)}}{K_+^{(2)}-2\left(K_+^{(1)}\right)^2},\quad a_1=\frac{2\left(K_+^{(2)}\right)^2-2K_+^{(1)}K_+^{(3)}}{K_+^{(2)}-2\left(K_+^{(1)}\right)^2} \\
b_2=&\frac{K_+^{(4)}-2K_+^{(1)}K_+^{(3)}}{K_+^{(2)}-2\left(K_+^{(1)}\right)^2},\quad b_1=\frac{2K_+^{(2)}K_+^{(3)}-2K_+^{(1)}K_+^{(4)}}{K_+^{(2)}-2\left(K_+^{(1)}\right)^2}
\end{align*}
it holds that
\begin{equation}
\label{eq:mateq1}
\kappa(K)^{-1}\begin{pmatrix}
0         & K_+^{(1)} & 2K_+^{(2)}   & K_+^{(2)} \\
K_+^{(1)} & K_+^{(1)} & K_+^{(2)} & K_+^{(2)} \\
2K_+^{(2)}   & K_+^{(2)} & 0   & K_+^{(3)} \\
K_+^{(2)} & K_+^{(2)} & K_+^{(3)} & K_+^{(3)}
\end{pmatrix}=\begin{pmatrix}
0 & 0 & a_1 & 0 \\
0 & 0 & 0   & a_1 \\
1 & 0 & -a_2 & 0 \\
0 & 1 & 2a_2 & a_2
\end{pmatrix}.
\end{equation}
and
$$\kappa(K)^{-1}\begin{pmatrix}
2K_+^{(2)}   & K_+^{(2)} & 0   & K_+^{(3)} \\
K_+^{(2)} & K_+^{(2)} & K_+^{(3)} & K_+^{(3)} \\
0   & K_+^{(3)} & 2K_+^{(4)}   & K_+^{(4)} \\
K_+^{(3)} & K_+^{(3)} & K_+^{(4)} & K_+^{(4)}
\end{pmatrix}=\begin{pmatrix}
a_1  & 0   & -b_1 & 0   \\
0    & a_1 & 2b_1 & b_1 \\
-a_2 & 0   &  b_2 & 0   \\
2a_2 & a_2 &    0 & b_2
\end{pmatrix}.$$
\end{lemma}
\begin{proof}
Note that by Jensen's Inequality $2\left(K_+^{(1)}\right)^2<K_+^{(2)}$ and thus we do not divide by zero. The remainder of the proof is direct calculation.
\end{proof}

\begin{lemma}
\label{lem:help3}
Let all conditions of Lemma \ref{lem:SigmaV} hold and suppose that the kernel $K$ is symmetric such that $\kappa(K)$ is invertible and $K^{(4)}<\infty$. Then,
\begin{align*}
&\kappa(K)^{-1}\IE(K_h(X_i)V_iV_i^{\top}) \\
=&f_X(0)I+f_X'(0)h\begin{pmatrix}
0 & 0 & a_1 & 0 \\
0 & 0 & 0   & a_1 \\
1 & 0 & -a_2 & 0 \\
0 & 1 & 2a_2 & a_2
\end{pmatrix}+h^2\frac{f_X''(0)}{2}\begin{pmatrix}
a_1  & 0   & -b_1 & 0   \\
0    & a_1 & 2b_1 & b_1 \\
-a_2 & 0   &  b_2 & 0   \\
2a_2 & a_2 &    0 & b_2
\end{pmatrix}+o(h^2),
\end{align*}
where $a_1,a_2,b_1,b_2$ are defined in Lemma \ref{lem:kernel}.
\end{lemma}
\begin{proof}
By \eqref{eq:kernel_exp} from Lemma \ref{lem:SigmaV} and Lemma \ref{lem:kernel} we obtain that (use symmetry of the kernel)
\begin{align*}
&\kappa(K)^{-1}\IE(K_h(X_i)V_iV_i^{\top}) \\
=&f_X(0)I+f_X'(0)h\kappa(K)^{-1}\begin{pmatrix}
0         & K_+^{(1)} & K^{(2)}   & K_+^{(2)} \\
K_+^{(1)} & K_+^{(1)} & K_+^{(2)} & K_+^{(2)} \\
K^{(2)}   & K_+^{(2)} & 0         & K_+^{(3)} \\
K_+^{(2)} & K_+^{(2)} & K_+^{(3)} & K_+^{(3)}
\end{pmatrix} \\
&+\frac{h^2}{2}f_X''(0)\kappa(K)^{-1}\begin{pmatrix}
K^{(2)}   & K_+^{(2)} &0          & K_+^{(3)} \\
K_+^{(2)} & K_+^{(2)} & K_+^{(3)} & K_+^{(3)} \\
0         & K_+^{(3)} & K^{(4)}   & K_+^{(4)} \\
K_+^{(3)} & K_+^{(3)} & K_+^{(4)} & K_+^{(4)}
\end{pmatrix}+o(h^2) \\
=&f_X(0)I+f_X'(0)h\begin{pmatrix}
0 & 0 & a_1 & 0 \\
0 & 0 & 0   & a_1 \\
1 & 0 & -a_2 & 0 \\
0 & 1 & 2a_2 & a_2
\end{pmatrix}+h^2\frac{f_X''(0)}{2}\begin{pmatrix}
a_1  & 0   & -b_1 & 0   \\
0    & a_1 & 2b_1 & b_1 \\
-a_2 & 0   &  b_2 & 0   \\
2a_2 & a_2 &    0 & b_2
\end{pmatrix}+o(h^2).
\end{align*}
\end{proof}

\subsection{The Lasso as Model Selector}
\label{subsec:LSM}
\subsubsection{The Result and Proof Structure}
The following two theorems show that the number of covariates selected by the Lasso is comparable to the size of the set $J_n$ and that $C_n$ as defined in \eqref{eq:defC} converges to zero quick enough. These properties of the Lasso estimator are relevant for showing that it can be used as a model selection procedure. Note that $\IE(K_b(X_i)Z_i)=0$ is not guaranteed. However, since we only care about $\tilde{\gamma}_n$ and since we argued in the proof of Theorem \ref{thm:main} that the value of $\tilde{\gamma}_n$ does not change if we change the centralization of the $Z_i$, we may assume in the following without loss of generality that $\IE(K_b(X_i)Z_i)=0$. We will also use the abbreviation $\mathbf{r}_n(b)=\mathbf{r}_n(J_n,b)$.

Define for a sequence $\lambda_n\geq0$ and numbers $0<w^{(l)}\leq1<w^{(u)}<\infty$ the event
\begin{align*}
\mathcal{T}(b)=&\Bigg\{2\left\|\frac{1}{n}\mathbf{V}^{\top}\mathbf{K}_b\mathbf{r}_n(b)\right\|_{\infty}\leq\frac{1}{2}\lambda_n\textrm{ and } \\
&\quad\quad\quad2\sup_{k=1,...,p_n}\left|\frac{1}{nb}\sum_{i=1}^n\hat{\omega}_{n,k}^{-1}Z_i^{(k)}K\left(\frac{X_i}{b}\right)r_i(J_n,b)\right|\leq\frac{1}{2}\lambda_n\Bigg\}.
\end{align*}
Moreover, we denote by $\tilde{\mathcal{T}}(b)$ the intersection of the events $\mathcal{T}(b)$ and
$$w^{(l)}\leq\min_{k\in J_n^c}\hat{\omega}_{n,k}\leq\max_{k=1,...,p_n}\hat{\omega}_{n,k}\leq w^{(u)}.$$
We will show in Corollary \ref{cor:T}, that we can choose $w^{(l)}\leq1<w^{(u)}$ and $C>0$ such that for $\lambda_n=C\sqrt{\frac{\log p_n}{nb}}$, $\IP(\tilde{\mathcal{T}}(b))\to1$.

\begin{theorem}
\label{thm:my_sparsity}
Let (CTB), (AS), (CV) and (TCS, \eqref{eq:equicont1}, \eqref{eq:equicont2} for $h=b$) as well as $\textrm{RSE}(|J_n|\log n,J_n,b)$ and $\textrm{CC}\left(\bar{w},J_n\right)$ hold and suppose that $f_X$ is continuous, $p_n\to\infty$, $b\to0$ and $\log p_n/nb\to0$. Then, $|\hat{J}_n|=O_P(|J_n|)$.
\end{theorem}
\begin{proof}
Note firstly that we may restrict to the event $\tilde{\mathcal{T}}(b)$ because of Corollary \ref{cor:T}. Moreover $\Phi(|J_n|\log n,J_n)=O_P(1)$ by Assumption $\textrm{RSE}(|J_n|\log n,J_n,b)$. We may thus also restrict to the event $\Phi(|J_n|\log n,J_n)\leq\Phi_0$ for some (possibly large but fixed) $\Phi_0>0$. Similarly, since $\textrm{CC}\left(\bar{w},J_n\right)$ holds, we may assume that $k\left(\bar{w},J_n\right)^{-1}\leq k\left(\bar{w}\right)^{-1}<\infty$. On these events, we have for all $m\leq |J_n|\log n$ (see Lemma \ref{lem:sparsity} for a definition of $L_n$)
$$2L_n|J_n|\Phi(\min(m,n),J_n)\leq 2\left(\frac{4w^{(u)}}{w^{(l)}}\right)^2\frac{4}{k(\bar{w})}\Phi_0|J_n|.$$
Thus, for $n$ large enough, there are $m\in\IN$ which fulfill $m\leq |J_n|\log n$ and $m\in\mathcal{M}$. For each such $m$ we get from Lemma \ref{lem:sparsity} that
$$\left|\hat{J}_n\right|\leq\left|J_n\right|+\left|\hat{J}_n\setminus J_n\right|\leq |J_n|\left(1+L_n\Phi(\min(m,n),J_n)\right)\leq|J_n|\left(1+\left(\frac{4w^{(u)}}{w^{(l)}}\right)^2\frac{4}{k(\bar{w})}\Phi_0\right)$$
which finishes the proof.
\end{proof}

\begin{theorem}
\label{thm:LassoCn}
Let (AS), (CTB), (CV), (MS), (BW), (TCS \eqref{eq:equicont1}, \eqref{eq:equicont2} for $h$ and $b$), (D conditions on $\mu_Z$ and $\mu_Z'$), CC$(\bar{w},J_n)$, $\textrm{RSE}(|J_n|\log n,J_n,b)$ and RSE$(0,J_n,h)$ for $\tilde{Z}_i$ hold. Suppose that $f_X$ is continuous and that $p_n\to\infty$. Then, $\IP(\hat{J}_n\supseteq J_{0,n})\to1$ and
$$\left|C(J_{0,n})\right|=O_P\left(|J_n\setminus J_{0,n}|\cdot|J_n|\frac{\log p_n}{ng}\right).$$
\end{theorem}
\begin{proof}
Since all assumptions of Theorem \ref{thm:my_sparsity} are assumed, we may use that $|\hat{J}_n|=O(|J_n|)$. Thus we have that $\IP(|\hat{J}_n|\leq\log n|J_n|)\to1$ and therefore we may restrict to the event $|\hat{J}_n|\leq\log n|J_n|$. By Corollary \ref{cor:T} we may also restrict to the event $\tilde{\mathcal{T}}(b)$. Hence, we obtain on $\tilde{\mathcal{T}}(b)\cap\{|\hat{J}_n|\leq|J_n|\log n\}$
\begin{align*}
\|\gamma_0(J_n,b)-\tilde{\gamma}_n\|_2^2\leq&\frac{1}{\phi(|J_n|\log n,J_n)}\frac{1}{n}\left\|\mathbf{K}_b^{\frac{1}{2}}\begin{pmatrix}\mathbf{V} & \mathbf{Z}\end{pmatrix}\begin{pmatrix}
\theta_0(J_n,b)-\tilde{\theta}_n \\ \gamma_0(J_n,b)-\tilde{\gamma}_n\end{pmatrix}\right\|_2^2 \\
\leq&\frac{4\lambda_n^2|J_n|(w^{(u)})^2}{k\left(\bar{w},J_n\right)\phi(|J_n|\log n,J_n)}.
\end{align*}
Thus we find that on on $\tilde{\mathcal{T}}(b)\cap\{|\hat{J}_n|\leq|J_n|\log n\}$, we have for any $k\in\{1,...,p_n\}$
$$\left|\gamma_0^{(k)}(J_n,b)-\tilde{\gamma}_n^{(k)}\right|\leq \frac{2Cw^{(u)}}{\sqrt{k(\bar{w},J_n)\phi(|J_n|\log n,J_n)}}\sqrt{\frac{|J_n|\log p_n}{c_{g,1}ng}},$$
where we used the properties of $b$ specified in (BW). Since we assume CC$(\bar{w},J_n)$ and $\textrm{RSE}(|J_n|\log n,J_n,b)$ and by Assumption (MS, \eqref{eq:vic}), and the definition of $J_{0,n}$, we find that for $k\in J_{0,n}$ necessarily $\tilde{\gamma}_n^{(k)}\neq0$ and hence we have that on $\tilde{\mathcal{T}}_n(b)\cap\{|\hat{J}_n|\leq|J_n|\log n\}$, $\hat{J}_n\supseteq J_{0,n}$ and in particular $\hat{J}_n\cap J_{0,n}=J_{0,n}$. This proves $\IP(\hat{J}_n\supseteq J_{0,n})\to1$. Moreover, we obtain on $\tilde{\mathcal{T}}(b)$
\begin{align*}
&\left|C_n(J_{0,n})\right| \\
=&\Bigg|\frac{1}{n}\left\|\mathbf{K}_h^{\frac{1}{2}}\left(\mathbf{Y}-\mathbf{V}\check{\theta}_0(J_n,h)-\tilde{\mathbf{Z}}\left(\check{\gamma}_0(J_n,h)\right)_{\hat{J}_n\cap J_{0,n}}\right)\right\|_2^2 \\
&\quad\quad\quad\quad\quad\quad\quad\quad-\frac{1}{n}\left\|\mathbf{K}_h^{\frac{1}{2}}\left(\mathbf{Y}-\mathbf{V}\check{\theta}_0(J_n,h)-\tilde{\mathbf{Z}}\check{\gamma}_0(J_n,h)\right)\right\|_2^2\Bigg| \\
=&\left|\frac{1}{n}\left\|\mathbf{K}_h^{\frac{1}{2}}\left(\mathbf{r}_n(h)+\tilde{\mathbf{Z}}\left(\gamma_0(J_n,h)\right)_{J_n\setminus J_{0,n}}\right)\right\|_2^2-\frac{1}{n}\left\|\mathbf{K}_h^{\frac{1}{2}}\mathbf{r}_n(h)\right\|_2^2\right| \\
=&\frac{1}{n}\left\|\mathbf{K}_h^{\frac{1}{2}}\tilde{\mathbf{Z}}\left(\gamma_0(J_n,h)\right)_{J_n\setminus J_{0,n}}\right\|_2^2+\frac{2}{n}\mathbf{r}_n(h)^{\top}\mathbf{K}_h\tilde{\mathbf{Z}}\left(\gamma_0(J_n,h)\right)_{J_n\setminus J_{0,n}} \\
\leq&\tilde{\Phi}(0,J_n)\left\|\left(\gamma_0(J_n,h)\right)_{J_n\setminus J_{0,n}}\right\|_2^2+2\left\|\frac{1}{n}\mathbf{r}_n(h)^{\top}\mathbf{K}_h\tilde{\mathbf{Z}}\right\|_{\infty}\left\|\left(\gamma_0(J_n,h)\right)_{J_n\setminus J_{0,n}}\right\|_1,
\end{align*}
where $\tilde{\Phi}$ denotes the restricted eigenvalue from Definition \ref{def:RSE} for $Z_i$ replaced by $\tilde{Z}_i$. From (MS), we obtain for some constant $C_0>0$
\begin{align*}
\left\|\left(\gamma_0(J_n,h)\right)_{J_n\setminus J_{0,n}}\right\|_2^2\leq&C_0|J_n\setminus J_{0,n}|\frac{|J_n|\log p_n}{ng}, \\
\left\|\left(\gamma_0(J_n,h)\right)_{J_n\setminus J_{0,n}}\right\|_1\leq& C_0|J_n\setminus J_{0,n}|\sqrt{\frac{|J_n|\log p_n}{ng}}.
\end{align*}
The statement of the lemma follows by invoking Lemma \ref{lem:ZR} and RSE$(0,J_n,h)$ for $\tilde{Z}_i$.
\end{proof}

\subsubsection{Supporting Results}
\label{subsec:suppresLasso}
In the proofs below we will work all the time with the same bandwidth sequence $b$ and the same target set $J_n$. To simplify notation, we write therefore $(\theta_{0,n},\gamma_{0,n})$ instead of $(\theta_0(J_n,b),\gamma_0(J_n,b))$ and $\mathbf{r}_n(b)$ instead of $\mathbf{r}_n(J_n,b)$.

\begin{lemma}
\label{lem:weights}
Let (CTB), (CV), $\IE(K_b(X_i)Z_i)=0$, $p_n\to\infty$, $b\to0$ and $\log p_n/nb\to0$ be true. Then, there are numbers $0<w^{(l)}\leq1<w^{(u)}<\infty$ such that with probability converging to one
$$w^{(l)}\leq\min_{k=1,...,p_n}\hat{\omega}_{n,k}\leq\max_{k=1,...,p_n}\hat{\omega}_{n,k}\leq w^{(u)}.$$
\end{lemma}
\begin{proof}
We note that
\begin{equation}
\label{eq:weight_split}
\hat{\omega}_{n,k}^2=\frac{1}{nb}\sum_{i=1}^nK\left(\frac{X_i}{b}\right)^2\left(Z_i^{(k)}\right)^2-b\left(\frac{1}{nb}\sum_{i=1}^nK\left(\frac{X_i}{b}\right)Z_i^{(k)}\right)^2.
\end{equation}
We may apply Lemma \ref{lem:expbound} with $x>1$ since (CTB, \eqref{eq:B2}) holds. Hence,
\begin{align*}
&\IP\left(\max_{k\in\{1,...,p_n\}}\left|\frac{1}{nb}\sum_{i=1}^n\left(K\left(\frac{X_i}{b}\right)^2\left(Z_i^{(k)}\right)^2-\IE\left(K\left(\frac{X_i}{b}\right)^2\left(Z_i^{(k)}\right)^2\right)\right)\right|\geq x\sqrt{\frac{\log p_n}{nb}}\right) \\
\leq&p_n^{1-x}\to0.
\end{align*}
Since $\log p_n/nb\to0$, we conclude from assumption (CV) that the first part of \eqref{eq:weight_split} is uniformly bounded with probability converging to $1$. Similarly, (CTB, \eqref{eq:B0}) together with $\IE(K_b(X_i)Z_i^{(k)})=0$ implies the same for the second part of \eqref{eq:weight_split}. Combining these two yields the result.
\end{proof}

\begin{corollary}
\label{cor:T}
Suppose that all assumptions of Lemmas \ref{lem:ZR} and \ref{lem:weights} hold and let the conditions of Lemma \ref{lem:VR} for $h=b$ hold. Then, there are choices of $0<w^{(l)}\leq1<w^{(u)}<\infty$ and $C\in(0,\infty)$ such that for $\lambda_n=C\sqrt{\frac{\log p_n}{nb}}$, $\IP(\tilde{\mathcal{T}}(b))\to1$.
\end{corollary}
\begin{proof}
Note firstly that by \eqref{eq:VR} in Lemma \ref{lem:VR} the first condition in the definition of $\mathcal{T}$ is true with probability converging to one for any choice $C>0$. Lemmas \ref{lem:weights} and \ref{lem:ZR} show that also the other events hold with probability converging to one for $w^{(l)},w^{(u)}$ as in Lemma \ref{lem:weights} and $C>0$ large enough.
\end{proof}
In the following we will always assume that we have chosen $0<w^{(l)}\leq1<w^{(u)}<\infty$ and $C>0$ as in Corollary \ref{cor:T}. The following lemma is a version of a classical result about the performance of the Lasso estimator and the proof is along the lines of similar results like e.g. in Chapter 6.2.2 in \citet{vdGB11} or in Theorem 1 of \citet{BC13}. However, since we have here a localized, partially penalized model with variable weights, we provide the proof for completeness.
\begin{lemma}
\label{lem:reg_error}
On the event $\tilde{\mathcal{T}}(b)$ we have for $\bar{w}=3w^{(u)}/w^{(l)}$ that
$$\frac{1}{n}\left\|\mathbf{K}_b^{\frac{1}{2}}\begin{pmatrix}\mathbf{V} & \mathbf{Z}\end{pmatrix}\begin{pmatrix}
\theta_0(J_n,b)-\tilde{\theta}_n \\ \gamma_0(J_n,b)-\tilde{\gamma}_n\end{pmatrix}\right\|_2^2+w^{(l)}\lambda_n\left\|\begin{pmatrix}
\theta_0(J_n,b)-\tilde{\theta}_n \\ \gamma_0(J_n,b)-\tilde{\gamma}_n\end{pmatrix}\right\|_1\leq4\frac{\lambda_n^2|J_n|(w^{(u)})^2}{k\left(\bar{w},J_n\right)}.$$
\end{lemma}
\begin{proof}[Proof of Lemma \ref{lem:reg_error}]
The proof is analogous to the results in Chapter 6.2.2 of \citet{vdGB11}. Define $\tilde{\alpha}_n=\begin{pmatrix}\theta_{0,n} & \gamma_{0,n}\end{pmatrix}^{\top}-\begin{pmatrix}\tilde{\theta}_n & \tilde{\gamma}_n\end{pmatrix}^{\top}$. When indexing $\tilde{\alpha}_n$ by a set $S$ we implicitly mean that $S=(S_1,S_2)$ comprises two index sets. The first index set $S_1\subseteq\{1,...,4\}$ indicates which indices of the first part $\theta_{0,n}-\tilde{\theta}_n$ are included. The second index set $S_2\subseteq\{1,...,p_n\}$ indicates which indices of the second part $\gamma_{0,n}-\tilde{\gamma}_n$ are included. From now on, we let $S_n=(\{1,...,4\},J_n)$ and $S_n^c=(\emptyset,J_n^c)$.

On $\tilde{\mathcal{T}}(b)$ and by using that $(\gamma_{0,n})_{J_n^c}=0$, we obtain the following inequality chain (the first inequality is frequently called basic inequality)
\begin{align}
&\frac{2}{n}\left\|\mathbf{K}_b^{\frac{1}{2}}\begin{pmatrix}\mathbf{V} & \mathbf{Z}\end{pmatrix}\tilde{\alpha}_n\right\|_2^2\leq-4\frac{1}{n}\mathbf{r}^{\top}_n\mathbf{K}_b\begin{pmatrix}\mathbf{V} & \mathbf{Z}\end{pmatrix}\tilde{\alpha}_n+2\lambda_n\sum_{k=1}^{p_n}\hat{\omega}_{n,k}\left(|\gamma_{0,n}^{(k)}|-|\tilde{\gamma}_n^{(k)}|\right) \nonumber \\
\leq&4\left\|\frac{1}{n}\mathbf{V}^{\top}\mathbf{K}_b\mathbf{r}_n(b)\right\|_{\infty}\cdot\left\|\theta_{0,n}-\tilde{\theta}_n\right\|_1 \nonumber \\
&\quad\quad\quad+4\sup_{k=1,...,p_n}\left|\frac{1}{nb}\sum_{i=1}^n\hat{\omega}_{n,k}^{-1}Z_i^{(k)}K\left(\frac{X_i}{b}\right)r_i(J_n,b)\right|\cdot\sum_{k=1}^{p_n}\hat{\omega}_{n,k}\left|\gamma_{0,n}^{(k)}-\tilde{\gamma}_n^{(k)}\right| \nonumber \\
&\quad\quad\quad+2\lambda_n\sum_{k\in J_n}\hat{\omega}_{n,k}|\gamma_{0,n}^{(k)}-\tilde{\gamma}_n^{(k)}|-2\lambda_n\sum_{k\in J_n^c}\hat{\omega}_{n,k}|\gamma_{0,n}^{(k)}-\tilde{\gamma}_n^{(k)}| \nonumber \\
\leq&\lambda_n\left\|\theta_{0,n}-\tilde{\theta}_n\right\|_1+\lambda_n\sum_{k=1}^{p_n}\hat{\omega}_{n,k}\left|\gamma_{0,n}^{(k)}-\tilde{\gamma}_n^{(k)}\right| \nonumber \\
&\quad\quad\quad+2\lambda_n\sum_{k\in J_n}\hat{\omega}_{n,k}|\gamma_{0,n}^{(k)}-\tilde{\gamma}_n^{(k)}|-2\lambda_n\sum_{k\in J_n^c}\hat{\omega}_{n,k}|\gamma_{0,n}^{(k)}-\tilde{\gamma}_n^{(k)}| \nonumber \\
\leq&3\lambda_n\left(\left\|\theta_{0,n}-\tilde{\theta}_n\right\|_1+\sum_{k\in J_n}\hat{\omega}_{n,k}|\gamma_{0,n}^{(k)}-\tilde{\gamma}_n^{(k)}|\right)-\lambda_n\sum_{k\in J_n^c}\hat{\omega}_{n,k}\left|\gamma_{0,n}^{(k)}-\tilde{\gamma}_n^{(k)}\right|. \label{eq:be}
\end{align}
Since the left hand side of the above inequality chain is non-negative as the square of a norm, we conclude from the above that (use that $w^{(u)}\geq1$)
\begin{align*}
\left\|\left(\tilde{\alpha}_n\right)_{S_n^c}\right\|_1\leq&\frac{1}{w^{(l)}}\sum_{k\in J_n^c}\hat{\omega}_{n,k}\left|\gamma_{0,n}^{(k)}-\tilde{\gamma}_n^{(k)}\right|\leq\frac{3}{w^{(l)}}\left(\left\|\theta_{0,n}-\tilde{\theta}_n\right\|_1+\sum_{k\in J_n}\hat{\omega}_{n,k}|\gamma_{0,n}^{(k)}-\tilde{\gamma}_n^{(k)}|\right) \\
\leq&\frac{3w^{(u)}}{w^{(l)}}\left\|\left(\tilde{\alpha}_n\right)_{S_n}\right\|_1.
\end{align*}
Thus, by definition of $k\left(\bar{w},J_n\right)$ in (CC), we conclude (use that $w^{(u)}\geq1$)
\begin{align*}
&\left\|\theta_{0,n}-\tilde{\theta}_n\right\|_1+\sum_{k\in J_n}\hat{\omega}_{n,k}|\gamma_{0,n}^{(k)}-\tilde{\gamma}_n^{(k)}|\leq w^{(u)}\left\|\left(\tilde{\alpha}_n\right)_{S_n}\right\|_1 \\
\leq&w^{(u)}\sqrt{\frac{|J_n|}{nk\left(\bar{w},J_n\right)}}\left\|\mathbf{K}_b^{\frac{1}{2}}\begin{pmatrix}\mathbf{V} & \mathbf{Z}\end{pmatrix}\tilde{\alpha}_n\right\|_2.
\end{align*}
Using the above in \eqref{eq:be}, we obtain (use $w^{(l)}\leq1$)
\begin{align*}
&\frac{2}{n}\left\|\mathbf{K}_b^{\frac{1}{2}}\begin{pmatrix}\mathbf{V} & \mathbf{Z}\end{pmatrix}\tilde{\alpha}_n\right\|_2^2+\lambda_nw^{(l)}\left\|\tilde{\alpha}_n\right\|_1 \\
\leq&\frac{2}{n}\left\|\mathbf{K}_b^{\frac{1}{2}}\begin{pmatrix}\mathbf{V} & \mathbf{Z}\end{pmatrix}\tilde{\alpha}_n\right\|_2^2+\lambda_n\left\|\theta_{0,n}-\tilde{\theta}_n\right\|_1+\lambda_n\sum_{k\in J_n}\hat{\omega}_{n,k}|\gamma_{0,n}^{(k)}-\tilde{\gamma}_n^{(k)}| \\
&\quad\quad\quad\quad\quad\quad\quad+\lambda_n\sum_{k\in J_n^c}\hat{\omega}_{n,k}\left|\gamma_{0,n}^{(k)}-\tilde{\gamma}_n^{(k)}\right| \\
\leq&4\lambda_n\left(\left\|\theta_{0,n}-\tilde{\theta}_n\right\|_1+\sum_{k\in J_n}\hat{\omega}_{n,k}|\gamma_{0,n}^{(k)}-\tilde{\gamma}_n^{(k)}|\right) \\
\leq&4w^{(u)}\lambda_n\sqrt{\frac{|J_n|}{nk\left(\bar{w},J_n\right)}}\left\|\mathbf{K}_b^{\frac{1}{2}}\begin{pmatrix}\mathbf{V} & \mathbf{Z}\end{pmatrix}\tilde{\alpha}_n\right\|_2 \\
\leq&\frac{1}{n}\left\|\mathbf{K}_b^{\frac{1}{2}}\begin{pmatrix}\mathbf{V} & \mathbf{Z}\end{pmatrix}\tilde{\alpha}_n\right\|_2^2+\frac{4\lambda_n^2|J_n|(w^{(u)})^2}{k\left(\bar{w},J_n\right)}.
\end{align*}
Subtracting the regression error on both sides yields the statement of the lemma.
\end{proof}

The next lemma is a version of Lemma 2 in \citet{BC13} tailored to our situation.
\begin{lemma}
\label{lem:BC1}
On the set $\tilde{\mathcal{T}}(b)$ we have
$$\sqrt{|\hat{J}_n|}\leq2\sqrt{\Phi(|\hat{J}_n|,J_n)}\frac{4w^{(u)}}{w^{(l)}}\frac{\sqrt{|J_n|}}{k\left(\bar{w},J_n\right)^{\frac{1}{2}}}.$$
\end{lemma}
\begin{proof}
The proof is based on the proof of Lemma 2 in \citet{BC13}. We know from the KKT conditions (cf. Lemma 2.1 in \cite{vdGB11} that
\begin{align*}
\frac{2}{n}\left|\sum_{i=1}^nK_b(X_i)Z_i^{(k)}\left(Y_i-V_i^{\top}\tilde{\theta}_n-Z_i^{\top}\tilde{\gamma}_n\right)\right|=&\lambda_n\hat{\omega}_{n,k}, \\
\frac{2}{n}\left|\sum_{i=1}^nK_b(X_i)V_i^{(l)}\left(Y_i-V_i^{\top}\tilde{\theta}_n-Z_i^{\top}\tilde{\gamma}_n\right)\right|=&0,
\end{align*}
for all $k\in \hat{J}_n$ and all $l=1,...,4$. From this observation, we conclude that
\begin{align}
&\lambda_n\sqrt{|\hat{J}_n|}=\left(\sum_{k\in\hat{J}_n}\left(\frac{2}{n}\sum_{i=1}^nK_b(X_i)\frac{Z_i^{(k)}}{\hat{w}_{n,k}}\left(Y_i-V_i^\top\tilde{\theta}_n-Z^\top_i\tilde{\gamma}_n\right)\right)^2\right)^{\frac{1}{2}} \nonumber \\
\leq& \left(\sum_{k\in\hat{J}_n}\left(\frac{2}{n}\sum_{i=1}^nK_b(X_i)\frac{Z_i^{(k)}}{\hat{w}_{n,k}}r_i(J_n,b)\right)^2\right)^{\frac{1}{2}} \nonumber \\
&+\left(\sum_{k\in\hat{J}_n}\left(\frac{2}{n}\sum_{i=1}^nK_b(X_i)\frac{Z_i^{(k)}}{\hat{w}_{n,k}}\left(V_i^\top\left(\theta_{0,n}(J_n,b)-\tilde{\theta}_n\right)+Z^\top_i\left(\gamma_{0,n}(J_n,b)-\tilde{\gamma}_n\right)\right)\right)^2\right)^{\frac{1}{2}} \nonumber \\
\leq& \left(\sum_{k\in\hat{J}_n}\left(\frac{2}{n}\sum_{i=1}^nK_b(X_i)\frac{Z_i^{(k)}}{\hat{w}_{n,k}}r_i(J_n,b)\right)^2\right)^{\frac{1}{2}} \nonumber \\
&\quad\quad+\frac{1}{w^{(l)}}\frac{2}{n}\underbrace{\left\|\begin{pmatrix}\mathbf{V}^\top \\ \mathbf{Z}(\hat{J}_n)^\top\end{pmatrix}\mathbf{K}_b\begin{pmatrix}\mathbf{V} & \mathbf{Z}\end{pmatrix}\begin{pmatrix}\theta_{0,n}(J_n,b)-\tilde{\theta}_n \\ \gamma_{0,n}(J_n,b)-\tilde{\gamma}_n \end{pmatrix}\right\|_2}_{=:c}  \label{eq:abc}
\end{align}
Let $\alpha=(\theta^{\top} \, \gamma^{\top})^{\top}\in\IR^{4+p_n}$ be defined via
\begin{align*}
\theta=&c^{-1}\mathbf{V}^{\top}\mathbf{K}_b\begin{pmatrix} \mathbf{V} &\mathbf{Z}\end{pmatrix}\begin{pmatrix}\theta_0(J_n,b)-\tilde{\theta}_n \\ \gamma_0(J_n,b)-\tilde{\gamma}_n\end{pmatrix}, \\
\gamma_{\hat{J}_n}=&c^{-1}\mathbf{Z}(\hat{J}_n)^{\top}\mathbf{K}_b\begin{pmatrix} \mathbf{V} &\mathbf{Z}\end{pmatrix}\begin{pmatrix}\theta_0(J_n,b)-\tilde{\theta}_n \\ \gamma_0(J_n,b)-\tilde{\gamma}_n\end{pmatrix}, \quad\gamma_{\hat{J}_n^c}=0.
\end{align*}
Thus, $\|\gamma\|_0\leq|\hat{J}_n|$ and $\|\alpha\|_2$=1. Moreover, by definition of $\alpha$ and of the restricted sparse eigenvalue, we get from Lemma \ref{lem:reg_error}
\begin{align*}
&\frac{1}{n}\left\|\begin{pmatrix} \mathbf{V}^{\top} \\ \mathbf{Z}(\hat{J}_n)^{\top}\end{pmatrix}\mathbf{K}_b\begin{pmatrix} \mathbf{V} &\mathbf{Z}\end{pmatrix}\begin{pmatrix}\theta_0(J_n,b)-\tilde{\theta}_n \\ \gamma_0(J_n,b)-\tilde{\gamma}_n\end{pmatrix}\right\|_2 \\
=&\frac{1}{n}\left|\alpha^{\top}\begin{pmatrix} \mathbf{V}^{\top} \\ \mathbf{Z}^{\top}\end{pmatrix}\mathbf{K}_b\begin{pmatrix} \mathbf{V} &\mathbf{Z}\end{pmatrix}\begin{pmatrix}\theta_0(J_n,b)-\tilde{\theta}_n \\ \gamma_0(J_n,b)-\tilde{\gamma}_n\end{pmatrix}\right| \\
\leq&\frac{1}{n}\left\|\alpha^{\top}\begin{pmatrix} \mathbf{V}^{\top} \\ \mathbf{Z}^{\top}\end{pmatrix}\mathbf{K}_b^{\frac{1}{2}}\right\|_2\left\|\mathbf{K}_b^{\frac{1}{2}}\begin{pmatrix} \mathbf{V} &\mathbf{Z}\end{pmatrix}\begin{pmatrix}\theta_0(J_n,b)-\tilde{\theta}_n \\ \gamma_0(J_n,b)-\tilde{\gamma}_n\end{pmatrix}\right\|_2 \\
\leq&\sqrt{\Phi(|\hat{J}_n|,J_n)}2w^{(u)}\frac{\lambda_n\sqrt{|J_n|}}{k\left(\bar{w},J_n\right)^{\frac{1}{2}}}.
\end{align*}
Using the above, we obtain on $\tilde{\mathcal{T}}(b)$
\begin{align*}
\eqref{eq:abc}\leq&\frac{1}{2}\lambda_n\sqrt{|\hat{J}_n|}+4\sqrt{\Phi(|\hat{J}_n|,J_n)}\frac{\lambda_n\sqrt{|J_n|}w^{(u)}}{k\left(\bar{w},J_n\right)^{\frac{1}{2}}w^{(l)}}.
\end{align*}
Thus, we obtain by rearranging
\begin{align*}
&\frac{1}{2}\lambda_n\sqrt{|\hat{J}_n|}\leq4\sqrt{\Phi(|\hat{J}_n|,J_n)}\frac{\lambda_n\sqrt{|J_n|}w^{(u)}}{k\left(\bar{w},J_n\right)^{\frac{1}{2}}w^{(l)}} \\
\Leftrightarrow&\sqrt{|\hat{J}_n|}\leq8\sqrt{\Phi(|\hat{J}_n|,J_n)}\frac{\sqrt{|J_n|}w^{(u)}}{k\left(\bar{w},J_n\right)^{\frac{1}{2}}w^{(l)}}.
\end{align*}
\end{proof}

Having established this result, we can also proof our own version of Theorem 3 in \citet{BC13}.
\begin{lemma}
\label{lem:sparsity}
Denote
$$L_n=\frac{4}{k\left(\bar{w},J_n\right)}\left(\frac{4w^{(u)}}{w^{(l)}}\right)^2.$$
On $\tilde{\mathcal{T}}(b)$ we have
$$\left|\hat{J}_n\setminus J_n\right|\leq L_n|J_n|\min_{m\in \mathcal{M}}\Phi(\min(m,n),J_n),$$
where $\Phi$ is defined in Definition \ref{def:RSE} and
$$\mathcal{M}=\left\{m\in\IN: m>|J_n|\Phi(\min(m,n),J_n)2L_n\right\}.$$
\end{lemma}
\begin{proof}
The proof can be carried out along the same lines as the proof of Theorem 3 in \citet{BC13} by noting that their equation (A.2) reads in our case (this is a consequence of Lemma \ref{lem:BC1})
$$|\hat{J}_n\setminus J_n|\leq |J_n| \Phi(|\hat{J}_n|,J_n)L_n.$$
\end{proof}

\section{Details on Lasso Parameter Choice}
\label{sec:app_implementation}
In this section, we give a formal description of the tuning parameter choices for the localized Lasso that we mention in Section \ref{sec:numerical_results}. Since the implementation of cross-validation is straightforward, we focus on the two other methods, which both aim at finding a value of $\lambda$ such that the set $\mathcal{T}(b)$  defined in Section \ref{subsec:LSM} has a large probability. We begin with describing our adaptation of the method from \citet{BCH14} to our RD setting. In this algorithm $\hat{w}_{n,k}$ is replaced by an estimate of
$$w_{n,k}=\sqrt{\IE\left(K_h(X)\left(Z_i^{(k)}r_i(J_n,b)\right)^2\right)}.$$
One then sets $\lambda=2c\sqrt{nb}\Phi^{-1}(1-\gamma/2p_n)$, where $\Phi$ is the standard normal distribution function and, following \citet{BCH14},  $c=1.1$ and $\gamma=0.05$. In order to construct estimates $\hat{w}_{n,k}$ of the infeasible weights $w_{n,k}$, we employ the following algorithm:
\begin{enumerate}
	\item Obtain residuals $\hat{r}_i$ from estimating a standard local linear RD regression without covariates.
	\item Compute
	$$\hat{w}_{n,k}=\sqrt{\frac{1}{nb}\sum_{i=1}^nK\left(\frac{X_i}{b}\right)^2\left(Z^{(k)}_i\hat{r}_i\right)^2}.$$
	\item Fit the model with covariates, using $\hat{w}_{n,k}$ as penalty loadings and $\lambda$ as described above. Compute new residuals $\hat{r}_i$ and let $\hat{J}_n$ be the set of covariates which receive a non-zero parameter.
	\item Update the penalty loadings as
	$$\hat{w}_{n,k}=\sqrt{\frac{1}{nb}\sum_{i=1}^nK\left(\frac{X_i}{b}\right)^2\left(Z_i^{(k)}\hat{r}_i\right)^2}\cdot\sqrt{\frac{nb}{nb-|\hat{J}|_n+4}}.$$
	\item Repeat steps 3-4 until either the absolute change in the updated penalty loadings is smaller than $\nu$, or after $K$ repetitions.  In Section \ref{sec:numerical_results}, we choose $\nu=10^{-5}$ and $K=10$.
\end{enumerate}

We also adapt the method of \citet{VL20} to our RD setting. Here all notation is the same as in the main paper. The algorithm is as follows:
\begin{enumerate}
	\item Define a sequence of $M$ values $0<\lambda_1<...<\lambda_M$ such that for $\lambda_M$ all parameters are estimated to be zero.
	\item Compute for all $m=1,...,M$ the estimators using $\lambda_k$ as tuning parameter and compute the corresponding empirical residuals $\hat{r}_i(\lambda_k)$.
	\item Compute $e^{(1)},...,e^{(L)}$, where each $e_l$ comprises of $n$ i.i.d. standard normal random variables.
	\item Compute for each $m=1,...,M$:
	$$\left\{\max_{k=1,...,p_n}\left|\frac{2}{nb}\sum_{i=1}^n\hat{w}_{n,k}K\left(\frac{X_i}{b}\right)Z_i^{(k)}\hat{r}_i(\lambda_m)e_i^{(l)}\right|:l=1,...,L\right\},$$
	and let $\hat{q}_{\alpha}(\lambda_m)$ be the empirical $\alpha$-quantile of the above set.
	\item Let $\hat{m}=\min\{m: \hat{q}_{\alpha}(\lambda_{m'})\leq \lambda_{m'}\textrm{ for all }m'\geq m\}$ if $\hat{q}_{\alpha}(\lambda_M)\leq\lambda_M$ and $\hat{m}=M$ otherwise.
	\item Choose $\lambda_{\hat{m}}$ as the value of the tuning parameter. In Section \ref{sec:numerical_results}, we choose $M=5p_n$, $L=100$ and $\alpha=0.05$.
\end{enumerate}

\section{Additional Simulations}
\subsection{Simulation Results with Robust Bias Correction}
\label{sec:rdrobust_empirics}

In this section we show the same analysis as in Section \ref{sec:simulations}, except that  we use the bandwidth selection and standard error computations as implemented in the package {\tt rdrobust}. 
The results are shown in Table \ref{tab1rdrobust}. The interpretation of the results is very similar to those presented in the main text. An important difference, however, is that the standard errors from {\tt rdrobust} underestimate the true standard deviation of \emph{all} estimators, including the baseline and linear adjustment estimators with only a single covariate, by a factor that is at least about 10\%, and increases with the number of selected covariates. Correspondingly, confidence intervals based on the different estimators undercover slightly in cases where the number of selected covariates is small, and exhibit substantial distortions otherwise. Otherwise, the different methods for choosing the Lasso tuning parameter show the same tendency as observed in Section \ref{sec:simulations}: Cross-Validation (CV) chooses many covariates while the other two methods select a lower number of covariates.

 \begin{table}[!t]
 	\centering\caption{Simulation Results for {\tt rdrobust}}\label{tab1rdrobust}
 	\begin{tabular}{l|cccccc}
 		\toprule
 		Covariate Selection &  \#Cov. & Bias  & SD    & Avg.\  SE &  CI Length & Coverage           \\
 		\toprule
 		Lasso (CV)   & 11.4 & 0.0149 & 0.0335 & 0.0227 & 0.1295 & 86.9 \\ 
 		Lasso (BCH)  & 2.5 & 0.0155 & 0.0341 & 0.0281 & 0.1617 & 93.0 \\ 
 		Lasso (LV)   & 3.2 & 0.0154 & 0.0331 & 0.0270 & 0.1550 & 92.6 \\ 
 		\midrule
 		Fixed: No Covariates     & 0.0 & 0.0171 & 0.0588 & 0.0524 & 0.3015 & 94.0 \\ 
 		Fixed: Covariate 1       & 1.0 & 0.0164 & 0.0379 & 0.0329 & 0.1519 & 91.8 \\ 
 		Fixed: Covariates 1--10  & 10.0 & 0.0149 & 0.0309 & 0.0245 & 0.1126 & 88.3 \\ 
 		Fixed: Covariates 1--30  & 30.0 & 0.0132 & 0.0332 & 0.0221 & 0.1011 & 82.2 \\ 
 		Fixed: Covariates 1--50  & 50.0 & 0.0117 & 0.0374 & 0.0193 & 0.0892 & 72.0 \\ 
 		Fixed: Optimal Covariate & - & 0.0159 & 0.0318 & 0.0270 & 0.1555 & 93.4 \\ 
 		\bottomrule
 	\end{tabular}
 
 	\raggedright \footnotesize Results based on $10000$ Monte Carlo replications when the method from {\tt rdrobust} is used for inference. For each estimator, the table shows shows average number of selected covariates (\#Cov.), the bias (Bias), the standard deviation (SD), the average value of the final estimator's standard error (SE), the average length of the corresponding confidence interval for the parameter of interest (CI Length), and the share of simulation runs in which the respective confidence interval covered the true parameter value (Coverage). 	
 \end{table}
 
\subsection{Simulation in a Non-Sparse Setting}
\label{sec:non_sparse-simulations}
We complement the results from Section \ref{sec:simulations} by investigating a non-sparse setup. We use the same data generating process as in Section \ref{sec:simulations}, except that we chose $\alpha$ differently, namely
$$\alpha=(\underbrace{\alpha_0 \, \dots \,\alpha_0}_{50{\rm times}},\underbrace{0,\,\dots,\,0}_{150{\rm times}}),$$
%$\alpha_0=0.3883765$
where $\alpha_0=0.3883765$. This choice guarantees that $Y$ has the same variance as with the choice of $\alpha$ from Section \ref{sec:simulations}. As before we use a sample size of $n=1,000$ and we consider $10,000$ Monte-Carlo repetitions. In this setup, there are thus 50 covariates of roughly equal importance, which is a large number given the rather moderate sample size. The tables \ref{tab5} (using {\tt RDHonest}) and \ref{tab4} (using {\tt rdrobust}) show the simulation results. We see that the performance of the Lasso critically depends on the method used to determine the tuning parameter. While Cross-Validation (CV) chooses too many covariates and results in biased standard errors and CI under-coverage, the method (BCH) is very conservative in selecting very few covariates. While this is no problem for the coverage, the standard deviation could be lower when more covariates are incorporated: The method (LV) selects more covariates but not too many to cause problems for the coverage. These patterns appear for both {\tt RDHonest} and {\tt rdrobust}.

  \begin{table}[!t]
 	\centering\caption{Simulation Results for {\tt RDHonest}}\label{tab5}
 	\begin{tabular}{l|cccccc}
 		\toprule
 		Covariate Selection &  \#Cov. & Bias  & SD    & Avg.\ SE &  CI Length & Coverage           \\
 		\toprule
 		Lasso (CV)   & 28.2 & 0.0068 & 0.0638 & 0.0294 & 0.1576 & 73.0 \\ 
 		Lasso (BCH)  &  0.5 & 0.0067 & 0.0720 & 0.0693 & 0.3096 & 96.3 \\ 
 		Lasso (LV)   &  1.2 & 0.0063 & 0.0684 & 0.0636 & 0.2861 & 95.9\\ 
 		\midrule
 		Fixed: No Covariates     &  0.0 & 0.0070 & 0.0749 & 0.0739 & 0.3294 & 96.7 \\ 
 		Fixed: Covariate 1       &  1.0 & 0.0063 & 0.0676 & 0.0658 & 0.2940 & 96.4 \\ 
 		Fixed: Covariates 1--10  & 10.0 & 0.0058 & 0.0599 & 0.0513 & 0.2343 & 94.2 \\ 
 		Fixed: Covariates 1--30  & 30.0 & 0.0059 & 0.0548 & 0.0345 & 0.1692 & 87.0 \\ 
 		Fixed: Covariates 1--50  & 50.0 & 0.0062 & 0.0489 & 0.0214 & 0.1181 & 76.9 \\ 
 		Fixed: Optimal Covariate & -    & 0.0050 & 0.0532 & 0.0515 & 0.2305 & 96.5 \\ 
 		\bottomrule
 	\end{tabular} 
 	\raggedright \footnotesize Results based on $10000$ Monte Carlo replications. For each estimator, the table shows shows average number of selected covariates (\#Cov.), the bias (Bias), the standard deviation (SD), the average value of the final estimator's standard error (SE), the average length of the corresponding confidence interval for the parameter of interest (CI Length), and the share of simulation runs in which the respective confidence interval covered the true parameter value (Coverage). 	
 \end{table}
 
   \begin{table}[!t]
 	\centering\caption{Simulation Results for {\tt rdrobust}}\label{tab4}
 	\begin{tabular}{l|cccccc}
 		\toprule
 		Covariate Selection &  \#Cov. & Bias  & SD    & Avg.\ SE &  CI Length & Coverage           \\
 		\toprule
 		Lasso (CV)   & 65.4 & 0.0117 & 0.1742 & 0.0166 & 0.0939 & 49.9 \\ 
 		Lasso (BCH)  &  1.9 & 0.0156 & 0.0532 & 0.0443 & 0.2544 & 93.2 \\ 
 		Lasso (LV)   &  3.1 & 0.0155 & 0.0513 & 0.0417 & 0.2394 & 92.8 \\ 
 		\midrule
 		Fixed: No Covariates     &  0.0 & 0.0169 & 0.0595 & 0.0528 & 0.3036 & 94.2 \\ 
 		Fixed: Covariate 1       &  1.0 & 0.0166 & 0.0525 & 0.0463 & 0.2143 & 92.3 \\ 
 		Fixed: Covariates 1--10  & 10.0 & 0.0153 & 0.0453 & 0.0371 & 0.1714 & 89.8 \\ 
 		Fixed: Covariates 1--30  & 30.0 & 0.0136 & 0.0413 & 0.0280 & 0.1287 & 83.6 \\ 
 		Fixed: Covariates 1--50  & 50.0 & 0.0116 & 0.0375 & 0.0194 & 0.0896 & 71.7 \\ 
 		Fixed: Optimal Covariate & -    & 0.0165 & 0.0392 & 0.0340 & 0.1960 & 93.9 \\ 
 		\bottomrule
 	\end{tabular} 
 	\raggedright \footnotesize Results based on $10000$ Monte Carlo replications. For each estimator, the table shows shows average number of selected covariates (\#Cov.), the bias (Bias), the standard deviation (SD), the average value of the final estimator's standard error (SE), the average length of the corresponding confidence interval for the parameter of interest (CI Length), and the share of simulation runs in which the respective confidence interval covered the true parameter value (Coverage). 	
 \end{table}
\end{document}